\documentclass[a4paper,12pt]{article}

\usepackage{amsmath, amsfonts, amssymb, amsthm}
\usepackage{authblk}
\usepackage{enumerate, color}
\usepackage{multido}
\usepackage{url}
\usepackage{mathrsfs}
\usepackage{mathtools}
\usepackage{verbatim}
\usepackage{setspace}
\usepackage{mathabx}
\usepackage{cleveref}
\usepackage[top=2.5cm,bottom=2.5cm,left=2.5cm,right=2.5cm]{geometry}
\usepackage{tikz}
\usepackage{subcaption}
\usepackage{pdfpages}
\usepackage{enumitem}
\usepackage[normalem]{ulem}
\usepackage{tikz}
\usetikzlibrary{positioning}

\newtheorem{theorem}{Theorem}
\newtheorem{corollary}{Corollary}
\newtheorem{proposition}{Proposition}
\newtheorem{lemma}{Lemma}
\newtheorem{remark}{Remark}

\numberwithin{equation}{section}

\newcommand{\C}{\mathbf{C}}
\newcommand{\R}{\mathbf{R}}

\newcommand{\Z}{\mathbf{Z}}

\newcommand{\e}{{\mathrm e}}

\newcommand{\imag}{\mathrm i}

\newcommand{\supp}{{\mathrm{supp}}}
\newcommand{\prob}{\mathbb{P}}
\newcommand{\expect}{\mathbb E}

\newcommand{\Const}{C}
\newcommand{\triad}{t}
\newcommand{\dist}{\mathrm{dist}}
\newcommand{\caS}{\mathcal S}

\newcommand{\bbE}{\mathbb E}
\newcommand{\caN}{\mathcal N}
\newcommand{\caJ}{\mathcal J}

\newcommand{\caT}{\mathcal T}

\newcommand{\caP}{\mathcal P}
\newcommand{\caW}{\mathcal W}
\newcommand{\caC}{\mathcal C}
\newcommand{\caK}{\mathcal K}

\newcommand{\caA}{\mathcal A}
\newcommand{\caV}{\mathcal V}
\newcommand{\caR}{\mathcal R}

\newcommand{\fral}{{\mathbf{l}}}
\newcommand{\frar}{{\mathbf{r}}}

\newcommand{\nonres}{\mathrm{NR}}
\newcommand{\fatnonres}{\mathbf{NR}}
\newcommand{\res}{\mathrm{R}}
\newcommand{\fatres}{\mathbf{R}}
\newcommand{\barres}{\widehat{\mathbf{R}}}
\newcommand{\tree}{{\Upsilon}}

\begin{document}

\title{Absence of Normal Heat Conduction in Strongly Disordered Interacting Quantum Chains}

\author[1]{Wojciech De Roeck}
\author[2]{Lydia Giacomin}
\author[3]{François Huveneers}
\author[4]{Oskar Pro\'sniak}
\affil[1]{Institute for Theoretical Physics, KU Leuven, 3000 Leuven, Belgium}
\affil[2]{Ceremade, UMR-CNRS 7534, Université Paris Dauphine, PSL Research
University, 75775 Paris, France}
\affil[3]{Department of Mathematics, King’s College London, Strand, London WC2R 2LS, United Kingdom}
\affil[4]{Department of Physics and Materials Science, University of Luxembourg, L-1511 Luxembourg}

\date{\today}
\maketitle
\begin{abstract}
	 We prove that in strongly disordered, interacting, quantum chains, the conductance of a chain of length $L$ vanishes faster than $1/L$. This means that transport is anomalous in such chains. This phenomenon was first claimed in \cite{Basko2006,gornyi2005interacting} and a pioneering treatment appeared in \cite{Imbrie2016a}. 
\end{abstract}

\section{Introduction}
Starting with the seminal work \cite{Anderson1958} of P. Anderson in the '50s, it was realized that non-interacting electrons in a disordered potential landscape can exhibit a vanishing conductivity. This is connected to the spectral phenomenon now known as \emph{Anderson localization}, which has been extensively studied both in the physics literature (see \cite{abrahams201050} for a review) and the mathematics literature \cite{gol1977pure,kunz1980spectre,Frohlich1983,Aizenman1993}, where it corresponds to a transition between point spectrum and continuous spectrum. 
In case the electrons do interact with each other, the problem is way more complicated and the mathematical setup that was so successful in studying Anderson localization is no longer applicable.  Since the works \cite{Basko2006,gornyi2005interacting,oganesyan2007localization} in 2005-2007, there has been, at least for a few years, a consensus in the physics community that weak interactions do not suffice to restore nonzero conductivity. 
Recently, this consensus has been challenged, see Section~\ref{sec: MBL} for references on the ensuing debate. 


This debate motivates our work. We consider strongly disordered quantum spin chains and we prove that the conductivity indeed vanishes for such systems, at any temperature.
To keep our results as transparent as possible, we have chosen to state the vanishing of conductivity in a non-equilibrium setup: we show that the time-averaged heat flow through a chain of length $L$ vanishes faster than $1/L$, with large probability.   This result is Theorem \ref{thm: absence of conduction} in the next section. 

We conclude this introduction with some remarks.

\subsection{Robustness} Our result is formulated for a wide class of Hamiltonians. This is fairly important because of the following: it is quite straightforward to design spin chain Hamiltonians for which the conductivity is manifestly vanishing. A first example would be free spins, corresponding to the hamiltonian 
$$H=\sum_{i=1}^L Z_i$$  (notation is explained in the next section), or the random field transverse Ising model 
$$
H=  \sum_{i=1}^L\theta_i Z_i +\sum_{i=1}^{L-1} X_iX_{i+1}
$$
where $\theta_i$ are i.i.d.\ random variables uniformly distributed in $[0,1]$. 
The latter system can be mapped to free fermions via a Jordan-Wigner transformation and so the considerations on Anderson localization apply to this model. These are straightforward examples, but one can cook up models that are less well-known but for which the conductivity vanishes as well. This is even possible in classical models, see e.g.\ \cite{WDR_Huveneers_Olla} (see also \cite{nachtergaele_reschke}).
Therefore, it is important to stress that the debate we referred to above, pertains to \emph{robust} phenomena, as fine-tuned models can always behave in a deviant way.  For this reason, we have allowed fairly general interaction terms in our Hamiltonian.


\subsection{The Role of Spatial Dimension} 
As mentioned, our result is stated for one-dimensional systems whereas the original papers \cite{Basko2006,gornyi2005interacting} predicted zero conductivity in any spatial dimension, provided the disorder is strong enough. We see no reason to believe that our technique could be extended to higher dimension, however. This is in line with earlier non-rigorous work questioning the validity of the reasoning in higher spatial dimensions \cite{de2017stability,luitz2017small,ponte2017thermal}.   
A discussion of the difficulties that would arise in higher dimension is beyond the scope of this introduction.

\subsection{Many-Body Localization}\label{sec: MBL}
The issue of vanishing conductivity described above is often situated within the context of a stronger property that has come to be known as \emph{many-body localization} (MBL).
See \cite{Basko2006, gornyi2005interacting, vznidarivc2008many, oganesyan2007localization, pal2010many, kjall2014many, Serbyn2013, Huse2014, luitz2015many, ros2015integrals, Imbrie2016a, Imbrie2016b, Schreiber2015} for early works, \cite{ RevModPhys.91.021001, doi:10.1146/annurev-conmatphys-031214-014726} for reviews, and \cite{vsuntajs2020quantum, Sierant2020,KieferEmmanouilidis2020,abanin2021distinguishing,  PhysRevB.106.L020202, PhysRevB.105.174205, PhysRevB.105.224203, long2023phenomenology} for recent debates on the existence of MBL. 
The MBL property is stronger than vanishing conductivity in the following sense: zero conductivity still allows for subdiffusive transport through the chain, i.e.\@ the heat flow could be $1/L^{a}$, with $a>1$,  whereas MBL forces the heat flow to be suppressed exponentially in the length $L$ of the chain.  
The current paper does not rule out nor confirm the existence of MBL. Yet, the concept of MBL is still used as a technical tool. 
We prove the MBL property in certain rare regions of the chain; see Theorem~\ref{thm: locality of u}. 
The vanishing of conductivity then follows from the presence of these regions via an argument originally devised in \cite{agarwal2015, Agarwal2017, vznidarivc2016diffusive}.

\subsection{Earlier Mathematical Work}

Our proof is strongly inspired by the earlier work \cite{Imbrie2016a}, which introduced a KAM-like method applicable to extensive quantum systems. See also \cite{Imbrie2016} for an implementation of this method in the non-interacting setting. Like our work, \cite{Imbrie2016a} addresses systems without restriction on the spectral range.
Several earlier works have investigated localization in interacting systems near spectral edges \cite{frohlich1986localization, poschel1990small, elgart2022localization, beaud2017low}, as well as in spin glasses; see, e.g., \cite{manai_warzel_2023, yin2024eigenstate}, which build in part on \cite{laumann_pal_2014, baldwin_laumann_2017}. From a technical standpoint, these settings pose rather different challenges.

\subsection{Acknowledgements}
We enthusiastically thank our colleagues at the MBL meeting in Oxford in August 2024, organised by S.~Sondhi and V.~Khemani. Their feedback, encouragement, and stimulating questions were crucial to the development of this paper.  
We are also grateful to J.~Imbrie for enlightening discussions on the topic of this paper.

W.D.R. and O.A.P. (while the latter was a PhD student at KULeuven) were supported  by the FWO (Flemish Research Fund) grant G098919N, 
the FWO-FNRS EOS research project G0H1122N EOS 40007526 CHEQS, the KULeuven  Runners-up grant iBOF DOA/20/011, and the internal KULeuven grant C14/21/086.

\section{Model and Main Results}\label{sec: model and main results}

\subsection{Model}
We consider a quantum spin chain with finite length $L$,
where $L$ is a positive integer, and we write $\Lambda_L=\{1,\ldots, L\}$. 
{The length $L$ is arbitrary but fixed, and to keep the notation light, we will not indicate its dependence in all expressions below.}
Let $\mathcal H=\mathcal H_L$ be the Hilbert space  
\begin{equation}\label{eq: basic hilbert space}
    \mathcal H_L \; = \;  \C^2 \otimes \ldots \otimes \C^2 \qquad \text{($L$ tensors).}
\end{equation}
We single out a preferred orthonormal basis $|\pm 1\rangle$ in $\C^2$ and we define the Pauli matrices 
$X,Z$ by 
\[
    Z |\sigma_0\rangle \; = \; \sigma_0 |\sigma_0\rangle, 
    \qquad  
    X |\sigma_0\rangle \; = \; |-\sigma_0\rangle
\]
with $\sigma_0 \in \{\pm 1\}$. 
{The Pauli matrix $Y$ can be defined as $Y = \imag XZ$, but we will not use it explicitly, as it will always be expressed in terms of $X$ and $Z$.}
Given $x\in\Lambda_L$, we write $Z_x$ and $X_x$ for copies of these operators acting on the $x^{\mathrm{th}}$ leg of the tensor product in \eqref{eq: basic hilbert space}, and extended with identity on the other legs. We say that $X_x,Z_x$ are supported on $\{x\}$. 
In general, an operator $O$ is supported on a discrete interval $I\subset\Lambda_L$  if it can be written as $O_I \otimes 1_{I^c}$, where $O_I$ acts on the legs labeled by $x \in I$ and $1_{I^c}$ is the identity on all legs $x \in I^c$ of the tensor product.
{Here and throughout, a \emph{discrete interval}, or simply \emph{interval} when the context is clear, denotes a non-empty set of the form $K\cap\Z$ with $K$ a real interval.}

We consider the Hamiltonian operator on $\mathcal H$,
\begin{equation}\label{eq: main Hamiltonian}
    H 
    \; = \;
    \sum_{x=1}^L \theta_x Z_x +  \sum_{x=1}^{L-1}  \kappa_x Z_x Z_{x+1} 
    + \sum_{\substack{I\subset\Lambda_L\\ \text{interval}}}  (\gamma/2)^{|I|}  W_I
\end{equation}
where 
\begin{enumerate}
    \item 
    $(\theta_x)_{x\in\Lambda_L}$ is a sequence of i.i.d. random variables, uniformly distributed in $[0,1]$.
    
    \item\label{item: defining the constant CJ}
    The $\kappa_x$ are real numbers satisfying $\max_{x\in\Lambda_L} |\kappa_x|<C_\kappa$ for some fixed $C_\kappa<\infty$ that does not depend on $L$. 
    
    \item   
    $|I|$ denotes the cardinality of the discrete interval $I$.
    
    \item
    The operators $W_I$ are hermitian, supported on $I$, and their norm satisfies $\|W_I\|\leq 1$. 
    
    \item  
    The \emph{coupling strength} $\gamma >0$ is the main parameter of our model. It will be taken small enough.
\end{enumerate}

In the sequel, we will view the operator $H$ defined in \eqref{eq: main Hamiltonian} as a random operator: Its matrix elements are functions of the random variables $(\theta_x)_{x\in\Lambda_L}$.
We denote by $\prob$ the law of the variables $(\theta_x)_{x\in\Lambda_L}$ and by $\expect$ the corresponding expectation. 
We notice that $\prob$ is a uniform measure on the sample space $\Omega=[0,1]^L$.

{Below, we will use the term \emph{constants} to refer to positive real numbers that may depend only on $C_\kappa$, and on nothing else.
In particular, they never depend on $L$ or $\gamma$. 
We will often use the letters $C$ and $c$ to denote generic constants, with the understanding that their values may vary from one instance to another. 
Usually, we use the letter $C$ to stress that the constant needs to be taken large enough, and the letter $c$ to stress that it needs to be taken small enough.}

\subsection{Result on Many-Body Localization}
The Hamiltonian $H$ is hermitian and it can therefore be diagonalized in the joint eigenbasis of $(Z_x)_{x\in \Lambda_L}$, i.e.\@ there is an unitary $U$ such that
\begin{equation}\label{eq: H diagonalization}
    U^\dagger H U \; = \; D
\end{equation}
where $D$ is diagonal, i.e.\@ $[D,Z_x]=0$ for all $x\in\Lambda_L$.
Our main technical result states that we can choose $U$ to have the following strong locality property with an exponentially small probability as a function of $L$:

\begin{theorem}[Locality of $U$]\label{thm: locality of u}
There exist constants $\gamma_0, c, c' > 0$ such that, 
for any $\gamma < \gamma_0$ and any random Hamiltonian of the form~\eqref{eq: main Hamiltonian}, 
the following property holds with probability at least $e^{-\gamma^{c'}L}$:
For any discrete interval $I \subset \Lambda_L$ and any operator $O$ supported in $I$, there is sequences of operators $(O_{n})_{n\ge 0}$ and $(O_n')_{n\ge 0}$ satisfying
\begin{enumerate} 
    \item   
    \[
        \| O_{n}\|, \|O'_n\| \;\leq\; \gamma^{cn} \|O\|, 
    \]
    \item
    \[
        UO U^\dagger \; = \; \sum_{n=0}^{\infty} O_{n}, 
        \qquad 
        U^\dagger O U \; = \; \sum_{n=0}^{\infty} O'_{n},
    \]
    \item  
    $O_{n}$ and $O'_n$ are supported on $I_{n}=\{x \in\Lambda_L : \dist(x,I) \leq n\}$.
\end{enumerate} 
\end{theorem}

{
On the event where the conclusions of this theorem hold, we can construct a complete set of \emph{local integrals of motion} (LIOMs), as introduced in \cite{Serbyn2013,Huse2014}. 
A possible choice is
\[
    \mathcal Z_x = U Z_x U^\dagger, \qquad x \in \Lambda_L.
\]
These operators form a complete set of mutually commuting conserved quantities, since $[H, \mathcal Z_x] = [\mathcal Z_x,\mathcal Z_y]= 0$ for all $x,y \in \Lambda_L$. 
Our theorem states that they are quasi-local, i.e., they can be expanded as an exponentially converging sum of local operators. 
Moreover, since $D$ can be cast as
\begin{equation}\label{eq: D sum local}
    D \; =  \;\sum_{S \subset \Lambda_L} D_S \prod_{x \in S} Z_x
\end{equation}
for some real coefficients $D_S$ 
(where the sets $S$ are not required to be intervals), 
the Hamiltonian can be expressed as
\[
    H \; = \;  \sum_{S \subset \Lambda_L} D_S \prod_{x \in S} \mathcal Z_x.
\]}

It also follows from our theorem that the coefficients $|D_S|$ decay exponentially with the \emph{diameter} of $S$, denoted $d_S$, which is defined as the {length} of the smallest {real} interval containing $S$ {(note that this is one less than the cardinality of the smallest discrete interval containing $S$)}.

To see this, it pays to introduce some lighter notation, and to write
\begin{equation}\label{eq: decomposition H sys intervals}
H \; = \; 
\sum_{\substack{I\subset \Lambda_L\\\text{interval}}}  H_I, 
\qquad  
H_I \; = \; H^{(0)}_I+   (\gamma/2)^{|I|}  W_I
\end{equation}
where each term $H_I$ is supported on $I$ and $H^{(0)}_I$ is defined as
\[
H^{(0)}_I \; = \; 
    \begin{cases}
    \theta_x Z_x  &  \text{whenever $I=\{x\}$ for some $x$}, \\
    J_x Z_x Z_{x+1}  &  \text{whenever $I=\{x,x+1\}$ for some $x$}, \\
    0 & \text{otherwise.}
\end{cases}
\]
We can now apply our theorem to each local operator $H_I$ in~\eqref{eq: decomposition H sys intervals}, and conclude that $U^\dagger H_I U$ is quasi-local. 
Observing that the off-diagonal part of $\sum_I U^\dagger H_I U$ cancels out by definition of $U$, and noting that the operator norm of an operator bounds the norm of its diagonal part, we deduce by identification with~\eqref{eq: D sum local} that there exists a constant $C$ such that 
\begin{equation}\label{eq: bound local D}
    |D_S| \;\le\; C \min \left\{ 1,  d_S^2 \gamma^{c({d_S - 1})/2} \right\},
    \qquad S \ne \varnothing,
\end{equation}
where the prefactor $d_S^2$ accounts for the number of terms in the expansion~\eqref{eq: decomposition H sys intervals} of $H$ that can contribute to the term labeled by $S$ in \eqref{eq: D sum local}.

\subsection{Result on Absence of Heat Conduction}\label{sec: absence of conduction}

We couple the chain to heat baths on the left and right side of the chain. 
Let us make this more precise.  
We have finite-dimensional bath Hilbert spaces  $\mathcal H_{B,\fral},\mathcal H_{B,\frar}$
and hermitian operators $H_{B,\fral},V_{B,\fral},H_{B,\frar},V_{B,\frar}$ 
acting on $\mathcal H_{B,\fral}, \mathcal H_{B,\frar}$, respectively.
These operators satisfy the constraint  
\begin{equation}\label{eq: bounded coupling}
\|V_{B,\fral}\| , \| V_{B,\frar}\| \;\leq\; 1.
\end{equation}
The total Hamiltonian of the system is 
\begin{equation}\label{eq: def of H tot}
H_{\mathrm{tot}} \; = \;   
H_{B,\fral}+V_{B,\fral}\otimes X_1 + 
H_{B,\frar}+V_{B,\frar}\otimes X_L+
H_{\mathrm{sys}},
\end{equation}
with $H_{\mathrm{sys}} = H$ the bulk Hamiltonian defined in
\eqref{eq: main Hamiltonian}. 
The Hamiltonian $H_{\mathrm{tot}}$ acts on the Hilbert space 
\[
\mathcal H_{\mathrm{tot}}
\; = \; 
\mathcal H_{B,\fral}
\otimes 
\mathcal H_{L}
\otimes 
\mathcal H_{B,\frar}. 
\]

We will decompose $H_{\mathrm{sys}}$ and $H_{\mathrm{tot}}$ in a left and a right part.
Decomposing $H_{\mathrm{sys}} = H$ as in \eqref{eq: decomposition H sys intervals}, we define the left and right part of $H_{\mathrm{sys}}$ as
%
\[
    H_{\mathrm{sys},\fral}
    \; = \;
    \sum_{I: \min I \leq L/2} H_I, 
    \qquad 
    H_{\mathrm{sys},\frar} \; = \; 
    H_{\mathrm{sys}} - H_{\mathrm{sys},\fral}.
\]
Now all the terms in the total Hamiltonian are associated either to the left or to the right, hence we can also decompose $H_{\mathrm{tot}}$ as $H_{\mathrm{tot}}=H_{\fral}+H_{\frar}$ with 
\begin{equation}\label{eq: decompoisition H tot left right}
H_{\fral} \; = \;   
H_{B,\fral}+V_{B,\fral}\otimes X_1+ H_{\mathrm{sys},\fral}, 
\qquad
H_{\frar} \; = \; H_{\mathrm{tot}}-H_{\fral}.
\end{equation}

Our main quantity of interest is the {time}-averaged heat current: 
\[
\frac{1}{T}\int_0^T d t J(t), \qquad  J(t)=   e^{\imag t H_{\mathrm{tot}}} J  e^{-\imag t H_{\mathrm{tot}}}
\]
{for $T>0$,}
between left and right part of the system, where the instantaneous current operator is defined as
\begin{equation}\label{eq: current J definition}
 J \; = \;   
 \imag [H_{\mathrm{tot}}, H_{\fral}] 
 \; = \; 
 \imag [H_{\frar}, H_{\fral}]
\end{equation}
Let $\mathcal R$ be the set of density operators on $\mathcal H_{\mathrm{tot}}$, 
i.e.\@ positive hermitian operators $\rho$ with $\mathrm{Tr}(\rho)=1$. 
The expectation value of operators $O$ with respect to a density matrix $\rho$ is given by $\mathrm{Tr}(\rho O)$. 
Further, let us denote by $B$ all data describing the baths, that is: 
the choice of finite-dimensional Hilbert spaces 
$\mathcal H_{B,\fral}, \mathcal H_{B,\frar}$ and operators $H_{B,\fral},V_{B,\fral},H_{B,\frar},V_{B,\frar}$ satisfying the constraint \eqref{eq: bounded coupling}.  
By $\sup_B (\cdot)$ in \eqref{eq: def of aveage current} below, we indicate that we take the supremum over all these choices.

Now, we can define the maximal  value of the {long time}-averaged heat current:
\begin{equation}\label{eq: def of aveage current}
\langle J^{(L)} \rangle_{\mathrm{ne}} 
\; = \;     
\limsup_{T \to \infty} \sup_{B} \sup_{\rho \in \mathcal R } 
\frac{1}{T}    \int_0^T dt  \mathrm{Tr}[\rho J(t)]  
\end{equation}
where the superscript $(L)$ reminds us that this quantity is computed for a fixed chain length $L$ and the subscript ``ne" makes explicit that this value corresponds to  non-equilibrium setup.
We are now ready to state the theorem on absence of heat conduction. 
\begin{theorem}\label{thm: absence of conduction}
Let the constants $\gamma_0,c,c'>0$ be as in Theorem~\ref{thm: locality of u}.
There exists a constant $C>0$ such that, 
for $\gamma<\gamma_0$ and with probability at least 
\[
    1-\exp\left(-\frac{L^{1-\gamma^{c'}}}{\log L}\right) ,
\]
we have
\[  
    \left|\langle J^{(L)} \rangle_{\mathrm{ne}}\right|
    \;\leq\; 
    C L^{-\frac{c}{5}\log (1/\gamma)} .
\]
In particular, if $\gamma$ is small enough, the conductance vanishes:
\begin{equation}
 \lim_{L\to\infty}\expect \left( L \langle J^{(L)} \rangle_{\mathrm{ne}} \right) \; = \;  0.
\end{equation} 
\end{theorem}

To interpret this theorem, let us make two remarks:

\begin{remark}
Our definition \eqref{eq: def of aveage current} takes a supremum over the bath data before taking the limit $T\to\infty$. 
This is important because finite systems can only exchange a finite amount of energy, hence the long time-average heat flux would vanish if we would take $T\to \infty$ before sending the size of the baths to infinity, and our result would be of little interest. 
\end{remark}

\begin{remark}
The most recognizable choice for the density matrix $\rho$ would be to take it thermal in the left and right baths, at two different temperatures. In this case, we expect the system to reach a non-equilibrium stationary state if the size of the baths is sent to infinity. Moreover, in a normal diffusive system, Ohm's law holds i.e.\@ the stationary current is proportional to $1/L$. Our quantity $\langle J^{(L)}\rangle_{\mathrm{ne}}$ constitutes an upper bound for the stationary current, and Theorem~\ref{thm: absence of conduction} shows thus that conduction is anomalous.
\end{remark}

\section{Outline of the Paper}\label{sec: scheme outline}

The remainder of this paper is dedicated to proving Theorems \ref{thm: locality of u} and \ref{thm: absence of conduction}. 

Our result on the vanishing conductivity, i.e.\@ Theorem \ref{thm: absence of conduction}, follows from many-body localization of the Hamiltonian on atypical segments whose size are logarithmic relative to the total system size. These segments are atypical in that they are free of resonance (see below).
Theorem \ref{thm: locality of u} shows that the Hamiltonian \eqref{eq: main Hamiltonian} is many-body localized  with a probability that remains at least exponentially small in the total length. In other words, it is likely to find a segment that is logarithmic in size relative to the total system where this result holds.
The transition from localization on a small segment to subdiffusion across the entire chain is described in the final Section \ref{sec: absence heat conduction}. 
The proof of this part relies on reasonably straightforward arguments. 

The proof of Theorem \ref{thm: locality of u} constitutes the most challenging part of our paper and spans Sections \ref{sec: diagrams} to \ref{sec: probability of absence of resonances}, as well as Sections \ref{sec: bounds on generators} and \ref{sec: proof of theorem one}, where the proof is concluded using classical tools such as the Lieb-Robinson bound.
We now outline the strategy of proof of this theorem.


\subsection{Renormalization Scheme}
We will establish a renormalization scheme by performing an infinite sequence of unitary changes of basis, aimed at reducing the amplitude and eventually removing all off-diagonal elements while maintaining locality, as stated in Theorem \ref{thm: locality of u}. As previously emphasized, this scheme is reminiscent of the Newton iteration scheme used in classical mechanics to prove the KAM theorem and is known as Schrieffer-Wolff transformation in our context.

In our scheme, we start with a Hamiltonian $H^{(0)} = H$, with $H$ defined in \eqref{eq: main Hamiltonian}, and we construct a sequence of Hamiltonians $(H^{(k)})_{k\ge 0}$ as well as skew-Hermitians operators, called \emph{generators}, $(A^{(k+1)})_{k\ge 0}$ acting on $\mathcal H_L$ such that 
\begin{equation}\label{eq: fundamental relation of the scheme}
    H^{(k+1)} 
    \; = \;
    \e^{A^{(k+1)}} H^{(k)}\e^{-A^{(k+1)}}
    \; = \;
    \e^{[A^{(k+1)},\cdot]}H^{(k)}
\end{equation}
for all $k\ge 0$. 
The unitary transformation $U$ featured in Theorem~\ref{thm: locality of u} is eventually defined as $U = \lim_{k\to\infty}\e^{-A^{(k)}}\dots \e^{-A^{(1)}}$. 
We will be able to control our scheme provided $\gamma$ is taken small enough and with probability at least exponentially small in system size. 

To develop an intuition on how this scheme works, let us carry out the first step at a formal level. 
We write the Hamiltonian $H^{(0)}$ as $H^{(0)} = E^{(0)} + V^{(0)}$ {(recall that $V^{(0)}$ is of order $\gamma$)},
a decomposition that can be directly inferred from \eqref{eq: main Hamiltonian}. 
We find
\[
    \e^{[A^{(1)},\cdot]}H^{(0)}
    \; = \; 
    E^{(0)} + \left( V^{(0)} + [A^{(1)},E^{(0)}]\right) + \mathcal{O} (\gamma^2).
\]
The expression in parentheses will vanish if $A^{(1)}$ is defined to satisfy
\begin{equation}\label{eq: naive expression for A}
    \langle\sigma'|A^{(1)}|\sigma\rangle
    \; = \;
    \gamma \frac{\langle \sigma'|V^{(0)}|\sigma\rangle}{E^{(0)}(\sigma') - E^{(0)}(\sigma)}
\end{equation}
for any configurations $\sigma, \sigma'$, using the notation $E^{(0)}(\sigma) = \langle \sigma|E^{(0)}|\sigma\rangle$.

At this step, a resonance occurs when the denominator $E^{(0)}(\sigma') - E^{(0)}(\sigma)$ becomes smaller in absolute value than some threshold, denoted by $\varepsilon$ in the sequel. 
Our analysis relies on the complete absence of resonance in the system. 
It is then rather straightforward to figure out that this can only happen with exponentially small probability as a function of the system size, and that the unitary transformation $\e^{[A^{(1)},\cdot]}$ preserves locality provided the resonance threshold $\varepsilon$ is large compared to the coupling strength $\gamma$.

It is worthwhile to observe that the expression \eqref{eq: naive expression for A} can only possibly make sense if the operator $V^{(0)}$ is off-diagonal, as otherwise the denominator would vanish identically for $\sigma = \sigma'$. This leads to the renormalization of the energy: Diagonal elements of the perturbation are incorporated into the renormalized energy $E^{(k)}$, the analog of $E^{(0)}$ at later steps of the scheme (more precisely only low-order elements are incorporated into $E^{(k)}$ at each steps).
This fact has rather far-reaching implications from a technical point of view: 
While a spin flip at site $x$ only changes the bare energy $E^{(0)}$ at the sites $x-1, x, x+1$, at later steps it will affect the renormalized energy $E^{(k)}$ on longer and longer stretches as $k$ grows large. This makes controlling the locality of our operators much more delicate at later steps of the procedure.  
The step parameter $k$ is called the \emph{scale} throughout the paper.

\subsection{Formal Set-Up and Inductive Control} \label{subsec: formal setup and inductive control}
Our renormalization scheme is introduced precisely in Sections~\ref{sec: diagrams} to \ref{sec: A k+1}. 
The description of our scheme is based on a full Taylor expansion of the exponential in $\e^{[A^{(k+1)},\cdot]}$, leading to expansions for the Hamiltonians $H^{(k)}$ and the generators $A^{(k+1)}$ at each step. 
In Section~\ref{sec: diagrams}, we introduce the notion of \emph{diagrams} and \emph{triads}, which will serve as bookkeeping devices to organize these expansions: 
The Hamiltonian $H^{(k)}$ is written as a sum over terms labelled by diagrams $g$, and the generator $A^{(k+1)}$ is expanded as a (finite) sum over triads $t$. 
The explicit construction of these Hamiltonians and generators is performed in Sections~\ref{sec: Hamiltonian k}-\ref{sec: A k+1}. 
It is worth noting that the material in these two sections is primarily algebraic; the control of resonances and the convergence of the expansions, particularly in a volume-independent manner, are not addressed at this stage.


At a technical level, the issue of locality brought up by the energy renormalization discussed above is responsible for the introduction of \emph{gap diagrams} in Section~\ref{sec: diagrams}, that make triads more complex objects than simple diagrams. 
More concretely, the second representation for the generators $A^{(k+1)}$ introduced in Section~\ref{sec: A k+1} is the construction that allows us to maintain control over the locality of all our operators.

The number of terms generated in the expansions is controlled in Section~\ref{sec: counting diagrams}. To exploit this result for establishing their convergence, we rely on the fact that the $n^{\mathrm{th}}$ order of the Taylor expansion of the exponential function involves the prefactor $1/n!$, which will be used to balance the proliferation of diagrams (we only use here a basic argument showing that the number of diagrams grows itself like $n!$; this may not be optimal since our system is one-dimensional).

In Section~\ref{sec: inductive bounds}, we define precise notions of resonance, cf.\@ the non-resonance conditions $\mathrm{NR}_{\mathrm{I}}$ and $\mathrm{NR}_{\mathrm{II}}$ in (\ref{eq: NR I}) and (\ref{eq: NR II}), respectively. We establish inductive bounds on the perturbation at each step, assuming that these non-resonance conditions hold. Together with the control on the number of diagrams established in Section~\ref{sec: counting diagrams}, these bounds allow us to fully establish the convergence of the expansions on the non-resonance event, as carried over in Section~\ref{sec: bounds on generators}. 

It is crucial to realize that the non-resonance condition $\mathrm{NR}_{\mathrm{II}}(t)$ in \eqref{eq: NR II} is itself an upper bound on the norm of $A^{(k+1)}(t)$ for a given triad $t$. 
As it turns out, this bound is even more stringent than the corresponding bound \eqref{eq: inductive bound A tilde non crowded} on the norm of $A^{(k+1)}(t)$ propagated in Proposition~\ref{pro: main inductive bounds} later in the same section. The impossibility of setting up an inductive scheme based only on the much more natural non-resonance condition $\mathrm{NR}_{\mathrm{I}}$ was already observed in \cite{Imbrie2016a}. 
This fact can probably be considered the biggest conceptual difficulty in devising a rigorous scheme.

The non-resonance condition $\mathrm{NR}_{\mathrm{II}}$ concerns the so-called \emph{non-crowded} diagrams, which span a spatial region of nearly maximal size (what “nearly maximal” means will be quantified by the introduction of the parameter $\beta$, see Sections~\ref{sec: parameters and scales}~and~\ref{subsec: aximatic diagrams}). 
The terms represented by these diagrams will thus not be estimated inductively, and their control relies on direct probabilistic estimates, that will occupy the second part of our work.

At a technical level, let us point out that the necessity to control the derivatives of our matrix elements with respect to the disorder in Proposition~\ref{pro: main inductive bounds} also originates from the energy renormalization. These estimates will be used in the proof of \eqref{eq: main probabilistic bound NRI} in Section~\ref{subsec: bounds on probability of resonances for equivalence classes}, and in the proof of \eqref{eq: main probabilistic bound NRII} in Section~\ref{sec: conclusion proof probability resonance}.

\subsection{Probabilistic Estimates}\label{sec: probabilistic estimates outline}
We are left with estimating the probability of resonances, which is the content of Sections~\ref{sec: probabilistic estimates} to \ref{sec: probability of absence of resonances}. The probability of the most natural non-resonance condition $\mathrm{NR}_{\mathrm I}$ is rather straightforward to control, see the proof in Section~\ref{subsec: bounds on probability of resonances for equivalence classes} (where the set-up could have been further simplified if we would only have to deal with that). 

Let us thus focus on the condition $\mathrm{NR}_{\mathrm{II}}$. 
Before coming to the core of the problem, let us point out a technical difficulty that introduces some complications in this part. 
On the one hand, as we said, to sum over diagrams, we are taking advantage of the prefactor $1/n!$ stemming from the Taylor expansion of the exponential function. 
On the other hand, the lower bound on the probability of having no resonance goes through a union bound over triads $t$ of the probability of $(\mathrm{NR}_{\mathrm{II}}(t))^c$ (where the superscript $c$ denotes the complement).
In other words, this means that we have to sum over the probabilities that each triad brings a resonant transition. 
However, the prefactor $1/n!$ cannot be ``used twice'', leading naively to a divergence. 

The solution to this issue is to realize that different triads can lead to the same probability event, resulting in a significant reduction in the number of terms in the union bound. This motivates the introduction of equivalence classes in Section~\ref{sec: equivalence of diagrams}, corresponding to triads that lead to the same event. 
The precise way to assign an event to each class, in such a way that each event depends only on the class, is itself not fully straightforward, and constitutes the content of Sections~\ref{sec: expansion of vertices into children}~and~\ref{sec: main bound matrix elements}. 
Furthermore, the number of equivalence classes is counted in Section~\ref{sec: counting equivalence classes of diagrams}, where it is shown that no inverse factorial is needed anymore.

Let us now address the central problem of estimating the probability of $\mathrm{NR}_{\mathrm{II}}(t)$ for given triad $t$. 
The operator $A^{(k+1)}(t)$ labelled by the triad $t$ at scale $k$ contains a product of several denominators:
\begin{equation}\label{eq: product of denominators}
    \prod_{i}\frac{1}{\Delta E_i}.
\end{equation}
Here, $\Delta E_i$ are energy differences analogous to those introduced earlier when discussing the first step, but generated at all scales $k' \le k$.
Considering this expression, we can now reiterate the statement made earlier: 
the scheme cannot be controlled inductively. 
Indeed, if we were to estimate each of the denominators using the basic non-resonance condition $\mathrm{NR_I}$, i.e.\@ by lower-bounding each denominator by a threshold of the type $\varepsilon^{L_{k'}}$ (where $L_{k'}$ corresponds to the scale at which this denominator was generated), the scheme would diverge.
Instead, we need to proceed with a direct probabilistic estimate. 
Following an idea already put forward in \cite{Imbrie2016a}, we partition the set of denominators $\Delta E_i$ into two sets, called $\mathcal{S}_{\mathrm{pro}}$ and $\mathcal{S}_{\mathrm{ind}}$, where ``pro/ind'' stand for ``probabilistic/inductive.''

Denominators that belong to $\mathcal{S}_{\mathrm{ind}}$ are still estimated inductively. 
The reason for estimating them inductively is either that we have an enhanced inductive bound on them corresponding to \eqref{eq: inductive bound A tilde crowded} (a case that is not problematic), or that they overlap too much with other denominators, leading to a difficulty in finding enough independent variables (a case that is potentially problematic). 
It is acceptable to estimate denominators of the latter kind inductively, as long as there are not too many of them, which is established in Section~\ref{sec: reduction number of variables}.

For the denominators in $\mathcal{S}_{\mathrm{pro}}$, the bound is manifestly not inductive: we have to take into account the randomness of the denominators, regardless of the scale at which they were created. To do this, we assign to each denominator $\Delta E_i$ in $\mathcal{S}_{\mathrm{pro}}$ a site $x$ distinct for each denominator, such that the leading part of $\Delta E_i$ does depend on the disorder variable $\theta_x$, i.e.\@ $\partial \Delta E_i/\partial \theta_x$ does not vanish.
More precisely, we need to control the determinant of the Jacobian matrix $(\partial \Delta E_i/\partial \theta_x)_{i,x}$, which allows us to perform a change of variables and integrate over independent denominators instead of independent disorder variables. This introduces additional conditions that are not straightforward to verify, and we accomplish this task in Section \ref{sec: selecting integration variables}.
The conclusion of the proof of the estimate of the probability of $\mathrm{NR}_{\mathrm{II}}(t)$ is carried out in Section~\ref{sec: conclusion proof probability resonance}.

Finally, we need to establish that the probability of having no resonance at all scales is at least exponentially small as a function of the length. 
This probability can be expressed as the partition function of a polymer system, cf.~\eqref{eq: bound on partition function}.  
The natural way to estimate it is via cluster expansion, and this is carried out in Section~\ref{sec: probability of absence of resonances}. 
It is worth pointing out though that the application of the Kotecky-Preiss criterion leads to somewhat unusual considerations in this case.

\section{Technical Preliminaries}\label{sec: technical}
In this section, we collect some considerations on auxiliary parameters used in the proof, as well as operators that appear throughout the paper.

\subsection{Parameters and Scales}\label{sec: parameters and scales}
Let us summarize the main parameters introduced in the proof. 
In addition to the coupling strength $\gamma > 0$ that appears in the Hamiltonian~\eqref{eq: main Hamiltonian}, we introduce the parameters $\delta$, $\varepsilon$ and $\beta$, with the constraints
\begin{equation}\label{eq: first constraints parameters}
    \gamma \; < \; \delta, \varepsilon \; < \; 1, 
    \qquad 
    \frac{1}{2} \;\le\; \beta \; < \; 1. 
\end{equation}
The role of these parameters may be understood as follows: 
\begin{enumerate}
    \item 
    The parameter $\varepsilon$ serves as a resonance threshold. 
    See the non-resonance conditions in \eqref{eq: NR I} and \eqref{eq: NR II} in Section~\ref{subsec: main results on inductive bounds}, where $\varepsilon$ is first used. 

    \item 
    The parameter $\delta$ does not have an intrinsic interpretation.
    It allows us to express that all our matrix elements decay at least exponentially as a function of the \emph{bare order} in perturbation, as introduced in Section~\ref{sec: diagrams}.
    This role becomes evident in the inductive bounds on the matrix elements in equations~\eqref{eq: inductive bound V tilde} to~\eqref{eq: inductive bound A tilde non crowded} below, 
    where $\gamma$ should be chosen small enough for given values of $\delta$ and $\varepsilon$ so that the right-hand sides scale at least as fast as $\delta^{\|g\|}$ or $\delta^{\|t\|}$.

    \item 
    The parameter $\beta$ is used to define the scales: 
    \begin{equation}\label{eq: definition of the scales}
    L_k \; = \;  (1+\beta)^k
    \end{equation}
    for all $k\ge 0$. 
    A value of $\beta$ close to 1 corresponds to a case where matrix elements that cannot be estimated inductively are rare and cover a large spatial domain, enabling probabilistic estimates. 
    See also Section~\ref{subsec: aximatic diagrams}.
\end{enumerate}

The values of these parameters can be determined according to the following logic:  
We first choose $\beta$ sufficiently close to $1$, then take $\delta$ and $\varepsilon$ sufficiently small, and finally choose $\gamma$ small enough, while ensuring that $\gamma$ remains a polynomial function of $\varepsilon$.  
Anticipating the needs of our proof, we now go through the various constraints that will arise. 
In doing so, we fix the value of $\beta$, as well as express $\delta$ and $\varepsilon$ as functions of $\gamma$, leading to stronger constraints than those in \eqref{eq: first constraints parameters}.  
As a result, all arguments in the proof will hold provided $\gamma$ is sufficiently small, uniformly in the system size $L$; we will not repeat this condition in each intermediate result.

First, $\beta$ needs to be close enough to $1$ so that the bound the bound \eqref{eq: corolary norm loss lemma} holds. 
This holds true as soon as $\beta \ge 1 - 1/312$, and we fix 
\begin{equation}\label{eq: value of beta}
    \beta \; = \; 1 - \frac{1}{312}.
\end{equation}
Second, $\delta$ and $\varepsilon$ will have to be large enough as compared to $\gamma$ so that the following inequalities are satisfied: 
\begin{align}
    & 
    \frac{\gamma}{\varepsilon^{1+b}} \;<\; \delta,
    \qquad b = 9, 
    \qquad \text{cf.\@ Proof of Corollary~\ref{cor: derivative tilded energy} in Section~\ref{subsec: main results on inductive bounds}},
    \label{eq: epsilon delta 1st}\\
    & 
    32 \left(\frac{\gamma}{\delta}\right)^{1-\beta} \frac{1}{\varepsilon^4} 
    \;<\; 1, 
    \qquad \text{cf.\@ \eqref{eq: first smallness gamma}},
    \label{eq: epsilon delta 2nd}\\
    & 
    \left(\frac{\gamma}{\delta}\right)^{1 - \beta}
    \frac1{\varepsilon^7} \; < \; 1
    \qquad \text{cf.\@ \eqref{eq: next non trivial constrain on eps and delt}}.
    \label{eq: epsilon delta 3rd}
\end{align}
Observe that the first constraint above also implies 
\begin{equation}\label{eq: implied constraint gamma delta epsilon}
    \frac{\gamma}{\varepsilon \delta} \; < \; 1,
\end{equation}
a constraint that will be used in Section~\ref{sec: proof of main theorems}. 

Third, the validity of \eqref{eq: bound on FNR} requires that $\gamma$ can be expressed as a (possibly high) power of $\varepsilon$.
All the constraints above can be satisfied if we set 
\begin{equation}\label{eq: delta epsilon power law gamma}
     \delta \; = \; \varepsilon \; = \; \gamma^{a(\beta)}
     \qquad \text{with} \qquad 
     a(\beta) \; < \; \frac{1 - \beta}{7 + 1 - \beta} \; < \; \frac1{7\times 312}.
\end{equation}

Finally, we recall that all constants are independent of \(L\) and \(\gamma\), and therefore also independent of \(\delta\) and \(\varepsilon\). 
Before the bound~\eqref{eq: corolary norm loss lemma}, where \(\beta\) is fixed, the constants can also be regarded as independent of \(\beta\).

\subsection{Operators}
Let us gather some technical information about the operators that we will be dealing with. 
Every operator on $\mathcal H_L$ can be written as a unique linear combination of the operators 
\begin{equation}\label{eq: fundamental monomials}
	X_1^{i_1}\dots X_L^{i_L}Z_1^{j_1}\dots Z_L^{j_L}
\end{equation}
with $i_k,j_k\in\{0,1\}$ for $1\le k \le L$. 
We notice the commutation relation for Pauli matrices
\begin{equation}\label{eq: anit commutation}
    Z_i X_j \; = \; (-1)^{\delta_{i,j}}X_j Z_i
\end{equation}
valid for all $1 \le i,j \le L$, and we deduce from this the more general form
\begin{equation}\label{eq: anti commutation generalized}
    f(Z_1,\dots,Z_L) X_1^{i_1}\dots X_L^{i_L}
    \; = \; 
    X_1^{i_1}\dots X_L^{i_L} f((-1)^{i_1}Z_1,\dots,(-1)^{i_L}Z_L)
\end{equation}
for all $i_k\in\{0,1\}$ for $1\le k \le L$ and any function $f$ on $\{\pm 1\}^L$.

We will make a frequent use of a special kind of operators: 
we say that an operator $A$ on $\mathcal H_L$ is a $X$-\emph{monomial} if it is of the form 
\begin{equation}\label{eq: X monomial definition}
    A \; = \; X_1^{i_1}\dots X_L^{i_L} f(Z_1,\dots,Z_L)
\end{equation}
for some $i_k\in\{0,1\}$ for $1\le k \le L$ and some function $f$ on $\{\pm 1\}^L$. 
We use a basis for $\mathcal H_L$ labeled by classical configurations $\sigma \in \{\pm 1\}^L$: the vector $|\sigma\rangle \in \mathcal H_L$ is characterized by $Z_i |\sigma\rangle=\sigma_i |\sigma\rangle $ for all $i=1,\ldots,L$. For an $X$-monomial $A$ as in \eqref{eq: X monomial definition}, we note that, for a given configuration $\sigma$, there is a unique configuration $\sigma'$, such that possibly $\langle \sigma'|A|\sigma\rangle \ne 0$, namely $|\sigma' \rangle = X_1^{i_1}\dots X_L^{i_L} |\sigma\rangle$. 
To lighten the notation, we will write 
\begin{equation}\label{eq: notation matrix element X monomial}
    \langle A |\sigma\rangle \; := \; \langle \sigma' |A|\sigma\rangle.
\end{equation}
%
An $X$-monomial $A$ as defined in \eqref{eq: X monomial definition} is hermitian if and only if
\begin{equation}\label{eq: self adjoint X monomial}
    f^*(Z_1,\dots,Z_L)= f((-1)^{i_1}Z_1,\dots,(-1)^{i_L}Z_L)  
\end{equation}
where $f^*$ denotes the complex conjugate of $f$. 

We also note that every hermitian/skew-hermitian operator on $\mathcal H_L$ can be expressed as a linear combination of hermitian/skew-hermitian $X$-monomials respectively. 

We will say that an interval $I\subset \Lambda_L$ is the \emph{support} of an operator $A$ on $\mathcal H_L$, and we write $I=\supp(A)$, if $I$ is the smallest interval (w.r.t.\ the inclusion) such that $A$ is supported on it.

Given an operator $A$ on $\mathcal H_L$, we denote by $\|A\|$ the operator norm on $\mathcal H_L$: 
\begin{equation}\label{eq: norm general}
    \|A\| \; = \; 
    \sup\left\{\|A\psi\|_2, \|\psi\|_2 \le 1\right\}
\end{equation}
where $\|\cdot \|_2$ is the Hilbert space norm on $\mathcal H_L$, i.e.\ $\|\psi\|_2^2=\sum_{\sigma \in \{\pm 1\}^L} |\langle \sigma|\psi\rangle|^2$.
If $A$ is an $X$-monomial with associated function $f$, then
\begin{equation}\label{eq: norm as maximum}
    \|A\|
\; =\;  \max_{\sigma}| \langle A |\sigma \rangle | 
\; =\;
\max_{\sigma}| f(\sigma)|.
\end{equation}

\section{Diagrams and Triads}\label{sec: diagrams}
In this section, we define what are diagrams and triads. 
We introduce them in an axiomatic way and construct their concrete implementation subsequently.

\subsection{Diagrams}\label{subsec: aximatic diagrams}
Let $k\ge 0$. 
We let $\mathcal G^{(k)}$ be a set whose elements are called \emph{diagrams at scale} $k$. This set will be defined inductively in $k$.
Diagrams have some attributes and properties that we list now. 
These properties will be guaranteed to hold for $k=0$ and their validity for higher $k$ will follow by construction.
\begin{enumerate}
	\item 
	Each diagram $g\in\mathcal G^{(k)}$ has an \emph{order}, denoted by $|g|$.
	It is a positive real number that satisfies
	\begin{equation}\label{eq: lower bound norm diagram}
		|g| \; \ge \; L_k.
	\end{equation}

        \item 
        Each diagram $g\in\mathcal G^{(k)}$ has an \emph{bare order}, denoted by $\|g\|$, and 
        \begin{equation} \label{eq: relation order bare order}
            \| g \| \; \ge \; |g|. 
        \end{equation}
	
	\item 
	To each diagram $g\in \mathcal G^{(k)}$ is associated a \emph{domain} 
	\[
		I(g) \; \subset \; \Lambda_L.
	\] 
        Given the domain $I(g)$, we also define 
        \[
        \overline{I}(g) = [\min I(g)-1,\max I(g) + 1] \cap \Lambda_L.
        \]
	Moreover, 
	\begin{equation}\label{eq: bound length of I}
		|g| \; \ge \; |I(g)|
	\end{equation}
	where $|I(g)|$ denotes the number of points in $I(g)$ (it is thus one unit larger than the usual length of the interval). 
	
	\item 
	A diagram $g\in \mathcal G^{(k)}$ is called \emph{crowded} if
	\begin{equation}\label{eq: crowded condition}
		|g| \; \ge \; \frac1\beta |I(g)|, 
	\end{equation}
	otherwise it is called \emph{non-crowded}. 
	
	\item 
	To each diagram $g\in \mathcal G^{(k)}$ is associated a (possibly empty) set $\mathcal A (g) \subset I(g)$ of \emph{active spins}.
	
	\item 
	A diagram $g\in \mathcal G^{(k)}$ is \emph{diagonal} if its set of active spins is empty, otherwise it is \emph{off-diagonal}. 
	
    \item 
    To each diagram $g\in \mathcal G^{(k)}$ is associated a \emph{diagram factorial}, denoted by $g!$. This is a positive integer. 

    \item 
    If $g$ is crowded, off-diagonal and such that $|g|<L_{k+1}$, 
    its \emph{reduced order} is denoted by $|g|_{\mathrm r}$ and is defined to be equal to 
    \begin{equation}\label{eq: norm jump diagram}
        |g|_{\mathrm r} \; = \; \max\{|I(g)|,\beta L_k\}.
    \end{equation}
\end{enumerate}

Let us make two remarks about the reduced order of a diagram $g$ (for which the notion makes sense). 
First, the bound \eqref{eq: lower bound norm diagram} may not be satisfied anymore for $|g|$ replaced by $|g|_{\mathrm r}$, while \eqref{eq: bound length of I} is satisfied, i.e.\ 
	\begin{equation}\label{eq: bound length of I possibly reduced}
		\min{(|g|, |g|_{\mathrm r})} \; \ge \; |I(g)|
	\end{equation}
Second, the bound 
\begin{equation}\label{eq: bound g prime and g}
	|g|_{\mathrm r} \;\le\; \beta |g| 
\end{equation}
holds. 
Indeed, if $|g|_{\mathrm r}=|I(g)|$, then $|g|_{\mathrm r} = |I(g)| \le \beta |g|$ using condition \eqref{eq: crowded condition}. 
If instead $|g|_{\mathrm r} = \beta L_k$, then $|g|_{\mathrm r}=\beta L_k \le \beta |g|$ using condition \eqref{eq: lower bound norm diagram}.

\subsection{Triads}\label{subsec: triads}
In order to define recursively the sets $\mathcal G^{(k)}$ for $k\ge 1$, we need some intermediate definitions. 
Let $k\ge 0$. 
First, let 
\begin{equation}\label{eq: participating diagonal diagrams}
    \mathcal D^{(k)} \; = \; \bigcup_{j=1}^k \left\{ 
    g \in \mathcal G^{(j-1)}, g \text{ is diagonal and }|g|< L_j\right\}
\end{equation}
with the convention $\mathcal D^{(0)} = \varnothing$. 
Let then $g\in \mathcal G^{(k)}$ be off-diagonal and such that $|g|< L_{k+1}$. 
We first define
\begin{equation}\label{eq: left gap diagrams}
    \mathcal L^{(k)}(g) 
    \; = \;
    \{g' \in \mathcal D^{(k)} : \mathcal A(g) \cap \overline I(g')\ne \varnothing, \min I(g') < \min I(g) \} \cup \{\varnothing\}. 
\end{equation}
Second, given also some $g'\in \mathcal L^{(k)}(g)$, we define 
\begin{multline}\label{eq: right gap diagrams}
    \mathcal R^{(k)}(g,g')
    \; = \; 
    \big\{ g'' \in \mathcal D^{(k)} : \mathcal A(g) \cap \overline I(g''))\ne \varnothing, \\ 
    \min I(g'')\ge \min I(g'),\max I(g'') > \max I(g) \big\}
    \cup \{\varnothing\}
\end{multline}
with the convention $\min I(g') = \min I(g)$ if $g' = \varnothing$. 
See Figure~\ref{fig: triad}.

\begin{figure}[!h]
    \centering

    \begin{tikzpicture}[xscale=0.8, yscale=0.8]

    \draw[dashed] (9.5, -1) -- (9.5, 2); 

    \def\a{1.5}
    \def\b{6.5}
    \def\ap{0.5}
    \def\bp{3.5}
    \def\app{4.5}
    \def\bpp{8}

    \def\topY{1.5}
    \def\lowY{0}

    \begin{scope}
    \draw[ultra thick] (\a, \topY) -- (\b, \topY);
    \draw[thick] (\a, \topY+0.1) -- (\a, \topY-0.1);
    \draw[thick] (\b, \topY+0.1) -- (\b, \topY-0.1);
    \node at ({(\a+\b)/2}, \topY+0.4) {$I(g)$};

    \draw[thick] (\ap, \lowY) -- (\bp, \lowY);
    \draw[thick] (\ap, \lowY+0.1) -- (\ap, \lowY-0.1);
    \draw[thick] (\bp, \lowY+0.1) -- (\bp, \lowY-0.1);
    \node at ({(\ap+\bp)/2}, \lowY+0.4) {$I(g')$};

    \draw[thick] (\app, \lowY) -- (\bpp, \lowY);
    \draw[thick] (\app, \lowY+0.1) -- (\app, \lowY-0.1);
    \draw[thick] (\bpp, \lowY+0.1) -- (\bpp, \lowY-0.1);
    \node at ({(\app+\bpp)/2}, \lowY+0.4) {$I(g'')$};
    \end{scope}

    \begin{scope}[xshift=10cm]
    \def\apnew{0.5}
    \def\bpnew{7.5} 

    \def\appnew{4.5}
    \def\bppnew{6.9} 

    \def\midY{0.5}    
    \def\lowYb{-0.5}  

    \draw[ultra thick] (\a, \topY) -- (\b, \topY);
    \draw[thick] (\a, \topY+0.1) -- (\a, \topY-0.1);
    \draw[thick] (\b, \topY+0.1) -- (\b, \topY-0.1);
    \node at ({(\a+\b)/2}, \topY+0.4) {$I(g)$};

    \draw[thick] (\apnew, \midY) -- (\bpnew, \midY);
    \draw[thick] (\apnew, \midY+0.1) -- (\apnew, \midY-0.1);
    \draw[thick] (\bpnew, \midY+0.1) -- (\bpnew, \midY-0.1);
    \node at ({(\apnew+\bpnew)/2}, \midY+0.4) {$I(g')$};

    \draw[thick] (\appnew, \lowYb) -- (\bppnew, \lowYb);
    \draw[thick] (\appnew, \lowYb+0.1) -- (\appnew, \lowYb-0.1);
    \draw[thick] (\bppnew, \lowYb+0.1) -- (\bppnew, \lowYb-0.1);
    \node at ({(\appnew+\bppnew)/2}, \lowYb+0.4) {$I(g'')$};
    \end{scope}

    \end{tikzpicture}
    \caption{
    Domains of $g, g', g''$ constituting a triad $t=(g,g',g'')$, in two different cases. 
    These domains follow the set rules. 
    The intervals $\overline I (g')$ and $I(g)$ overlap, and the interval $I(g')$ extends further to the left than $I(g)$. 
    The intervals  $\overline I (g'')$ and $I(g)$ overlap, the interval $I(g'')$ extends further to the right than $I(g)$ and no further to the left than $I(g')$.}
\label{fig: triad}
\end{figure}

Given $g$ as above, given $g'\in\mathcal L^{(k)}(g)$ and $g''\in\mathcal R^{(k)}(g,g')$, we say that the triplet $t = (g,g',g'')$ is a \emph{triad} at scale $k$. 
The set of triads at scale $k$ is denoted by $\mathcal T^{(k)}$.
We also define three functions $\mathsf c,\mathsf l,\mathsf r$ on $\mathcal T^{(k)}$ such that, if $t\in\mathcal T^{(k)}$ writes $t=(g,g',g'')$, it holds that 
\[
    \mathsf c(t) \; = \; g, \quad
    \mathsf l(t) \; = \; g', \quad
    \mathsf r(t) \; = \; g''.
\]

A triad $t=(g,g',g'')$ has attributes similar to diagrams: 
\begin{enumerate}
    \item 
    The order of $t$ is denoted by $|t|$ and is given by 
    \begin{align*}
    &|t| \; = \; |g| + |g'| + |g''| \qquad \text{if $g$ is non-crowded}, \\
    &|t| \; = \; |g|_{\mathrm r} + |g'| + |g''| \qquad \text{if $g$ is crowded.}
    \end{align*}

    \item 
    The bare order of $t$ is defined to be
    \[
        \|t\| \; = \; \| g\| + \| g'\| + \| g'' \|.
    \]

    \item 
    The domain of $t$ is denoted by $|I(t)|$ and is given by
    \[
    I(t) \; = \; I(g) \cup I(g') \cup I(g'').
    \]
    We use also the notation $\overline{I}(t)=[\min I(t) - 1,\max I(t)+1]\cap \Lambda_L$. 

    \item 
    The set of active spins of $t$ is equal to the set of active spins of $g$, and is denoted by $\mathcal A(t)$. 

    \item 
    The associated factorial is defined by 
    \[
        t ! \; = \; g! g'! g'' ! 
    \]
\end{enumerate}
In the above definitions, we have used the conventions $|\varnothing|=\|\varnothing\| = 0$, $I(\varnothing)=\varnothing$ and $\varnothing ! = 1$. 
For later use, we note that a triad $t\in \caT^{(k)}$ satisfies the bounds
\begin{equation}\label{eq: bound on order triads} 
  \beta L_k \;\leq\; |t| \; <\; L_{k+1}+2L_k \; < \; 3L_{k+1},  
\end{equation}
and 
\begin{equation}\label{eq: bound on extended support triads}
   |I(t)| \; \leq \;  |t| \; <\;  3L_{k+1}, \qquad   |\overline{I}(t)| \; < \; 3 L_{k+1} + 2 \; \le \; 5L_{k+1}.  
\end{equation}

\subsection{Diagrams at Scale $k=0$}\label{sec: diagrams at scale zero}
A diagram $g \in \mathcal G^{(0)}$ is a couple $g=(S,I)$ with $S\subset I\subset\Lambda_L$ with $I\ne\varnothing$. 
We define 
\begin{enumerate}
    \item 
    $|g| = |I|$,

    \item 
    $\|g\| = |g|$,

    \item 
    $I(g) = I$, 

    \item 
    $\mathcal A(g)=S$,

    \item 
    $g!=1$.
\end{enumerate}
We notice that $|g|=|I(g)|$ and that all diagrams in $\mathcal G^{(0)}$ are thus non-crowded.
We notice also that the bounds (\ref{eq: lower bound norm diagram}-\ref{eq: bound length of I}) hold at the scale $k=0$.

\subsection{From Diagrams at Scale $k$ to Diagrams at Scale $k+1$}\label{subsec: diagrams iteratively}
We now assume that the set $\mathcal G^{(k)}$ has been defined for some $k\ge 0$, and we define the set $\mathcal G^{(k+1)}$. 
It consists of $n+1$-tuples of the form
\begin{equation}\label{eq: concatenation}
	g \; = \; ({t_0},t_1,\dots,t_n)
\end{equation}
for some $n\ge 0$, 
where ${t_0}\in\mathcal G^{(k)}$, where $t_j\in\mathcal T^{(k)}$ for $1\le j \le n$, and such that the following conditions are satisfied: 
\begin{enumerate}
    \item 
    If $n=0$, then $|t_0|\ge L_{k+1}$.
    
    
    \item 
    For all $1\le j \le n$, there exists $0 \le i < j$ such that
    \begin{equation}\label{eq: locality constraint diagrams}
        \mathcal A(t_i,t_j) \; := \; 
        \left(\mathcal A(t_i)\cap\overline{I}(t_j)\right)\cup
        \left(\mathcal A(t_i)\cap\overline{I}(t_j)\right)
        \; \ne \; \varnothing.
    \end{equation} 
\end{enumerate} 

Given $g = (t_0, t_1, \dots, t_n) \in \mathcal{G}^{(k+1)}$, we say that distinct diagram/triads $t_i \ne t_j$ are \emph{adjacent} if they satisfy the constraint \eqref{eq: locality constraint diagrams} (notice that we do not require $i<j$ though). We will use the term ``adjacent'' exclusively to describe diagram/triads at some scale $k$ that are components of a given diagram at scale $k+1$, i.e.\@ those that are ``colleagues'',  using the terminology introduced in Section~\ref{sec: tree structure} below.

We also extend this terminology slightly, as it is sometimes more natural to refer to the central diagrams rather than the triads themselves. If $t_i$ and $t_j$ are adjacent, we say that the central diagrams $\mathsf{c}(t_i)$ and $\mathsf{c}(t_j)$ are adjacent as well (we adopt the convention $\mathsf{c}(t_i) = t_i$ when $i = 0$). Note, however, that this extended usage does not imply that the constraint \eqref{eq: locality constraint diagrams} must hold for $\mathsf{c}(t_i)$ and $\mathsf{c}(t_j)$ replacing $t_i$ and $t_j$, respectively.



\begin{remark}
    In \eqref{eq: concatenation} and below, we have used the notation $t_0$ for the first component of $g$ despite the fact that $t_0$ is a diagram and not a triad. 
    We will most often keep using this notation as it is convenient to put $t_0,t_1,\dots,t_n$ on the same footing.
\end{remark}

\begin{remark}
If $n=0$ in the above construction, then we say that such diagrams have been \emph{taken over} from the scale $k$, or that they have been \emph{regenerated} at the scale $k+1$. 
In the sequel, we will view the set of all diagrams as the disjoint union over $k\ge 0$ of $\mathcal G^{(k)}$, so that all diagrams come with a unique, well-defined scale. In particular, a diagram and its regenerated version will be viewed as different diagrams.
\end{remark}

Let $g=(t_0,t_1,\dots,t_n)$ be a diagram in $\mathcal G^{(k+1)}$ for some $k\ge 0$.
If $n=0$ in \eqref{eq: concatenation}, all the attributes of $g$ (order, bare order, domain,set of active spins and factorial) are simply these of $t_0$. If instead $n>0$, we define 
\begin{enumerate}
\item
The order of $g$: 
\[
	|g| \; = \; |t_0| + |t_1| + \dots + |t_n|.
\]

\item 
The bare order of $g$: 
\[
    \| g \| \; = \; \|t_0\| + \|t_1\| + \dots + \|t_n\|.
\]

\item 
The domain of $g$: 
\[
	I(g) \; = \; I(t_0) \cup I(t_1) \cup \dots \cup I(t_n). 
\]

\item 
The active spins of $g$: A site is an active spin of $g$ if it is an active spin for an odd number of diagrams/triads in 
$\{t_0,t_1\dots,t_n\}$.

\item 
The diagram factorial of $g$: 
\[
	g! \; = \; n ! t_0 ! t_1 ! \dots t_n ! 
\]
\end{enumerate}

\subsection{Propagation of the Bounds \eqref{eq: lower bound norm diagram}, \eqref{eq: relation order bare order}, and \eqref{eq: bound length of I}}
With the above definitions, the bounds \eqref{eq: lower bound norm diagram}, \eqref{eq: relation order bare order}, and \eqref{eq: bound length of I} propagate from scale $k$ to $k+1$, for all $k\ge 0$. 
Let $g$ as defined by \eqref{eq: concatenation}. 
If $n=0$, these three bounds hold indeed. 
If $n\ge 1$, 
\[
    |g| 
    \; \ge \; 
    |t_0| + |t_1| 
    \; \ge \; 
    |t_0| + |\mathsf c(t_1)|_{(\mathrm r)}
    \; \ge \; 
    L_k + \beta L_k \; = \; L_{k+1} 
\]
where $|\mathsf c(t_1)|_{(\mathrm r)}$ is equal to $|\mathsf c(t_1)|$ if $\mathsf c(t_1)$ is a non-crowded diagram and is equal to $|\mathsf c(t_1)|_{\mathrm r}$ if $\mathsf c(t_1)$ is a crowded diagram. 
This shows \eqref{eq: lower bound norm diagram}.
The propagation of \eqref{eq: relation order bare order} follows since the bare order $\|\cdot\|$ propagates in an additive way, whereas the order $|\cdot|$ is subadditive, because of the replacement of $|\cdot|$ by $|\cdot|_{\mathrm r}$.
Next, \eqref{eq: bound length of I} follows from the fact that 
\[
    |g| \; \ge \; |I(t_0)| + 
    \sum_{j=1}^n |I(\mathsf c(t_i))| + |I(\mathsf l(t_i))| + |I(\mathsf r(t_i))| \; \ge \; |I(g)|. 
\]

\subsection{Tree Structure of Diagrams and Triads}\label{sec: tree structure}

\begin{figure}[h!]
\centering

\begin{tikzpicture}[scale=1,every node/.style={circle, draw, minimum size=1cm}, level distance=2cm, sibling distance=3.5cm]

\draw[dashed] (-4,3.5) -- (4,3.5);  
\draw[dashed] (-4,1.5) -- (4,1.5);  
\draw[dashed] (-4,-2.5) -- (4,-2.5); 

\node[draw=none] at (5.5,2.5) {Scale $k=1$};
\node[draw=none] at (5.5,-0.5) {Scale $k=0$};

\node[ultra thick, font=\bfseries] (g) at (0,2.5) {$g$}
  child {node[ultra thick, font=\bfseries] (g0) {$g_0$} }
  child {node[fill=gray!40,ultra thick, font=\bfseries] (t1) {$t_1$}
  child {node[ultra thick, font=\bfseries] (g1) {$g_1$}}}
  child {node[fill=gray!40,ultra thick, font=\bfseries] (t2) {$t_2$}
  child {node[ultra thick, font=\bfseries] (g2) {$g_2$}}};

\end{tikzpicture}

\caption{The tree $\tree(g)$ for a diagram $g$ at scale $1$, with $g = (g_0,t_1,t_2)$, $t_1 = (g_1,\varnothing,\varnothing)$ and $t_2 = (g_2,\varnothing,\varnothing)$.
As descendants of $g$, the diagrams $g$ and $g_0$ are $V$-diagrams and the diagrams $g_1$ and $g_2$ are $A$-diagrams.
The diagram $g_0$ and the triads $t_1$ and $t_2$ are colleagues. 
The diagrams $g_1$, $g_1$ and $g_2$ are also colleagues. }

\label{fig: diagram at scale 1}

\end{figure}

\begin{figure}[h!]
\centering

\begin{tikzpicture}[
  scale=0.9,
  level distance=2cm,
  sibling distance=1.8cm,
  every node/.style={circle, draw, minimum size=1cm, align=center}
]

\draw[dashed] (-7.5,6) -- (6.5,6);
\draw[dashed] (-7.5,4) -- (6.5,4);
\draw[dashed] (-7.5,0) -- (6.5,0);
\draw[dashed] (-7.5,-2) -- (6.5,-2);
\draw[dashed] (-7.5,-4) -- (6.5,-4);

\node[draw=none] at (7.7,5) {Scale $k$};
\node[draw=none] at (7.7,2) {$k-1$};
\node[draw=none] at (7.7,-1) {$k-2$};
\node[draw=none] at (7.7,-3) {$k-3$};

\node[ultra thick, font=\bfseries]  (g) at (0,5.2) {$g$}
  child[grow=down] {
    node[ultra thick, font=\bfseries]  (g0) [xshift=-4cm] {$g_0$}
      child[level distance=4cm] {node[ultra thick, font=\bfseries] (g00) {$g_{0,0}$}}
      child[level distance=4cm] {node[fill=gray!40,ultra thick, font=\bfseries] (t01) {$t_{0,1}$}}
      child[level distance=4cm] {node[fill=gray!40,ultra thick, font=\bfseries] (t02) {$t_{0,2}$}}
      child[level distance=4cm] {node[fill=gray!40,ultra thick, font=\bfseries] (t03) {$t_{0,3}$}}
  }
  child[grow=down] {
    node[fill=gray!40,ultra thick, font=\bfseries]  (t1) [xshift=4cm] {$t_1$}
      child[level distance=4cm, xshift=-2.3cm] {node (g1p) {$g_1'$}}
      child[xshift=-0.9cm] {
        node[ultra thick, font=\bfseries]  (g1) {$g_1$}
          child {node[ultra thick, font=\bfseries]  (g10) {$g_{1,0}$}}
          child {node[fill=gray!40,ultra thick, font=\bfseries]  (t11) {$t_{1,1}$}}
      }
      child[level distance=6cm, xshift=-0.5cm] {node (g1pp) {$g_1''$}}
  };

\foreach \x in {g00, g10, g1p, g1pp} {
  \draw[dashed] (\x) -- ++(-60:1);
  \draw[dashed] (\x) -- ++(-120:1);
}

\foreach \x in { t01, t02, t03, t11} {
  \draw[dashed] (\x) -- ++(-60:1);
  \draw[dashed] (\x) -- ++(-90:1);
  \draw[dashed] (\x) -- ++(-120:1);
}
\end{tikzpicture}

\caption{The tree $\tree(g)$ for a diagram $g$ at some scale $k$, with $g = (g_0,t_1)$, $g_0 = (g_{0,0},t_{0,1},t_{0,2},t_{0,3})$, $t_1 = (g_1,g_1',g_1'')$ and $g_1 = (g_{1,0},t_{1,1})$.
As descendants of $g$, the diagrams $g,g_0,g_{0,0},g_{1,0}$ are $V$-diagrams, $g_1$ is an $A$-diagram, $g_1'$ is a left gap-diagram and $g_1''$ is a right gap-diagram.
The diagram $g_0$ and the triad $t_1$ are colleagues; 
the diagram $g_0$ and the diagram $g_1$ are colleagues; 
the diagram $g_{0,0}$ and the triads $t_{0,1}$, $\dots$, $t_{0,3}$ are colleagues; 
the diagram $g_{1,0}$ and the triad $t_{1,1}$ are colleagues.}

\label{fig: more general diagram}

\end{figure}

\begin{figure}[h!]
\centering

\begin{tikzpicture}[
  level distance=2cm,
  sibling distance=1.8cm,
  every node/.style={circle, draw, minimum size=1cm, align=center}
]

\draw[dashed] (-5,6) -- (5,6);   
\draw[dashed] (-5,4) -- (5,4);
\draw[dashed] (-5,2) -- (5,2);
\draw[dashed] (-5,0) -- (5,0);
\draw[dashed] (-5,-2) -- (5,-2);   

\node[draw=none] at (6,5) {Scale $k$};
\node[draw=none] at (6,3) {$k-1$};
\node[draw=none] at (6,1) {$k-2$};
\node[draw=none] at (6,-1) {$k-3$};

\node[ultra thick, font=\bfseries]  (g) at (0,5.25) {$g$}
  child {node[ultra thick, font=\bfseries]  {$g_0$}
    child {node[ultra thick, font=\bfseries]  {$g_{0,0}$}
      child {node[ultra thick, font=\bfseries]  (g000) {$g_{0,0,0}$}}
      child {node[fill=gray!40,ultra thick, font=\bfseries]  (t1) {$t_{0,0,1}$}}
      child {node[fill=gray!40,ultra thick, font=\bfseries]  (t2) {$t_{0,0,2}$}}
      child {node[fill=gray!40,ultra thick, font=\bfseries]  (t3) {$t_{0,0,3}$}}
      child {node[fill=gray!40,ultra thick, font=\bfseries]  (t4) {$t_{0,0,4}$}}
    }
  };
  
 \foreach \x in {g000} {
  \draw[dashed] (\x) -- ++(-60:1);
  \draw[dashed] (\x) -- ++(-120:1);
}

\foreach \x in { t1, t2, t3, t4} {
  \draw[dashed] (\x) -- ++(-60:1);
  \draw[dashed] (\x) -- ++(-90:1);
  \draw[dashed] (\x) -- ++(-120:1);
}

\end{tikzpicture}

\caption{The tree $\tree(g)$ for a diagram $g$ at some scale $k$ that has been ``taken over'' from scale $k-2$: $g = (g_0)$, $g_0 = (g_{0,0})$ and $g_{0,0} = (g_{0,0,0},t_{0,0,1},\dots,t_{0,0,4})$.
As descendants of $g$, the diagrams $g,g_0,g_{0,0},g_{0,0,0}$ are $V$-diagrams.
The diagram $g_{0,0,0}$ and the triads $t_{0,0,1}$, $\dots$, $t_{0,0,4}$ are colleagues.}

\label{fig: taken over diagram}

\end{figure}

Diagrams and triads come with a natural hierarchical structure, which is best represented by a rooted tree.  
We now explain how to associate to any diagram \( g \) and any triad \( t \) a rooted tree, denoted respectively by \( \tree(g) \) and \( \tree(t) \).  
Each vertex of the tree will be labeled by a diagram or triad that appears in the construction of \( g \) or \( t \).  
See Figures~\ref{fig: diagram at scale 1}--\ref{fig: taken over diagram} for graphical representations of these trees in specific examples.

Given \( k \ge 0 \) and a diagram \( g \in \mathcal{G}^{(k)} \), we define the labeled rooted tree \( \tree(g) \).  
We begin with the root, labeled by \( g \).  
If \( k = 0 \), the construction stops here.  
Otherwise, we decompose \( g \) as \( g = (t_0, t_1, \dots, t_n) \), as defined above.  
In this decomposition, recall that \( t_0 \) is a diagram at scale \( k-1 \), also denoted by \( g_0 \), and that \( t_1, \dots, t_n \) are triads at scale \( k-1 \).  
The root is then connected to \( n+1 \) vertices labeled by \( t_0, t_1, \dots, t_n \).  
Each triad \( t_i \) for \( 1 \le i \le n \) can be written as \( t_i = (g_i, g_i', g_i'') \), where \( g_i \) is a diagram at scale \( k-1 \), and \( g_i', g_i'' \) are either diagrams at scale \( k-2 \) or lower, or the empty set.  
Each vertex labeled by such a triad \( t_i \) is connected to up to three new vertices, labeled by \( g_i \), \( g_i' \), and \( g_i'' \), with the convention that a vertex is omitted if it corresponds to \( g_i' = \varnothing \) or \( g_i'' = \varnothing \).  
At this stage, we obtain a rooted tree with root labeled by a diagram at scale \( k \), and leaves labeled by diagrams at scale \( k-1 \) or lower (as long as \( k \ge 1 \)).  
The construction is then iterated from each leaf until all leaves are labeled by diagrams at scale \( 0 \).
The tree \( \tree(t) \) associated to a triad \( t \in \mathcal{T}^{(k)} \) is defined in the same way, except that the root of \( \tree(t) \) is labeled by \( t \).

To describe the trees, we adopt standard terminology from graph theory, along with some specific terms adapted to our setting.  
We begin with the standard part.  
Given a diagram \( g \), the \emph{descendants} of \( g \) are all the triads and diagrams that label the vertices of the tree \( \tree(g) \), including \( g \) itself.  
The same notion applies when \( g \) is replaced by a triad.  
If \( g' \) is a descendant of \( g \), we also say that \( g \) is an \emph{ancestor} of \( g' \).  
Two diagrams or triads \( g \) and \( g' \) are said to be in \emph{hierarchical relation} if one is a descendant of the other.  
We also use terms such as \emph{child}, \emph{grandchild}, \emph{sibling}, \emph{parent}, and \emph{grandparent} to describe hierarchical relations involving one or two edges.  
Note that the whole terminology may appear counterintuitive in our context, since a diagram is constructed from its descendants.  
However, these terms are natural when describing the structure of a tree.

We now introduce terminology that is specific to our setting.
Let \( g \) be a diagram or a triad, and let \( g' \) be a descendant of \( g \).  
Let us stress that the names we assign to \( g' \) below apply to it as a descendant of \( g \), and are not intrinsic.
First, \( g' \) is said to be an \emph{\( A \)-diagram} if \( g' = \mathsf{c}(t') \) for some triad \( t' \), which is itself a descendant of \( g \).  
Next, \( g' \) is said to be a \emph{left gap-diagram} if \( g' = \mathsf{l}(t') \) and a \emph{right gap-diagram} if \( g' = \mathsf{r}(t') \) for some triad \( t' \) that is a descendant of \( g \).  
If \( g' \) is neither an \( A \)-diagram nor a gap-diagram, it is said to be a \emph{\( V \)-diagram}.  
In particular, if \( g \) is a diagram, then \( g \) itself is a \( V \)-diagram.  
Both \( A \)-diagrams and \( V \)-diagrams are also referred to as \emph{non-gap-diagrams}.

Finally, let \( g \) be a diagram at some scale greater than or equal to 1, and let \( g = (t_0, t_1, \dots, t_n) \) be its decomposition at the previous scale.  
The diagram and triads \( t_0, t_1, \dots, t_n \) are called \emph{colleagues}, and they are also siblings in the usual graph-theoretic sense.  
The non-gap diagrams \( \mathsf{c}(t_0), \mathsf{c}(t_1), \dots, \mathsf{c}(t_n) \) are also referred to as \emph{colleagues}, where we use the convention \( \mathsf{c}(t_0) = t_0 \) (since \( t_0 \) is a diagram, not a triad).  
In this case, these colleagues are not siblings.

As a final remark, and as is clear from Figures~\ref{fig: diagram at scale 1}--\ref{fig: taken over diagram}, there is no direct correspondence between the scale of a descendant of a diagram or triad and the graph distance of the corresponding vertex to the root or the leaves.

\section{Hamiltonians}\label{sec: Hamiltonian k}
Given $k\ge 0$, we now show how the Hamiltonians $H^{(k)}$ can be represented using diagrams. In order to show that our representation carries over from scale $k$ to $k+1$, we will need to define the generator $A^{(k+1)}$ and state some of its properties. They will be shown in Section~\ref{sec: A k+1}.

In this section and the next, we focus on the algebraic construction of the scheme. We defer the definition of the non-resonance sets, on which all our expansions are well-defined and convergent for appropriate $L$-independent choices of the parameters $\varepsilon,\delta,\gamma$, to Section~\ref{sec: inductive bounds}.
However, even before Section~\ref{sec: inductive bounds}, straightforward arguments can be invoked  to conclude that all the expressions in Sections~\ref{sec: Hamiltonian k} and \ref{sec: A k+1} are well-defined almost surely, i.e.\@ outside of the disorder set where some denominator may vanish, and that the expansions are convergent for any given $L$.

\subsection{The Hamiltonian $H^{(0)}$}\label{subsec: Hamiltonian H 0}
The Hamiltonian $H^{(0)}$ is equal to the Hamiltonian $H$ defined in \eqref{eq: main Hamiltonian}.
The Hamiltonian $H^{(0)}$ can be decomposed as 
\[
    H^{(0)} \; = \; E^{(0)} + V^{(0)}
\]
where
\begin{equation}\label{eq: bare energy}
    E^{(0)} 
    \; = \; 
    \sum_{x=1}^L \theta_x Z_x + \sum_{x=1}^{L-1} J_x Z_x Z_{x+1}
\end{equation}
and where $V^{(0)}$ can be written as a expansion over $X$-monomials, cf.~\eqref{eq: X monomial definition}: 
\begin{equation}\label{eq: V 0 initial}
    V^{(0)}
    \; = \; 
    \sum_{S\subset I \subset \Lambda_L} X_S f_{S,I}((Z_x)_{x\in I}).
\end{equation}
In this expression, $X_S =\prod_{x\in S}X_x$, 
$I$ is a non-empty interval, 
and $f_{S,I}$ is a function on $\{\pm 1\}^{I}$ such that $X_S f_{S,I}((Z_x)_{x\in I})$ is hermitian, cf.~\eqref{eq: self adjoint X monomial}.
Moreover, the function $f_{S,I}$ satisfies the bound $\|f_{S,I}\|_{\infty} \le \gamma^{|I|}$, 
where $\| \cdot \|_\infty$ denotes the usual $\mathrm{sup}$-norm. 

We now verify that the expression $\sum_{I} (\gamma/2)^{|I|}W_I$ in \eqref{eq: main Hamiltonian} can indeed be written as $V^{(0)}$ in \eqref{eq: V 0 initial}. 
Given a non-empty interval $I\subset\Lambda_L$, 
the space of operators supported in $I$ is made into a Hilbert space by equipping it with the Hilbert-Schmidt norm $\|\cdot\|_{\mathrm{HS}}$. 
The operator $W_I$ admits the orthogonal decomposition
\[
    (\gamma/2)^{|I|} W_I \; = \; \sum_{S\subset I} X_S f_{S,I}((Z_x)_{x\in I}). 
\]
Hence 
\[
    (\gamma/2)^{2|I|}
    \| W_I\|_{\mathrm{HS}}^2 \; = \; 
    \sum_{S\subset I} \| X_S f_{S,I}((Z_x)_{x\in I}) \|_{\mathrm{HS}}^2.
\]
Since
$\| W_I\|_{\mathrm{HS}}^2\le 2^{|I|} \| W_I\|^2  \le  2^{|I|}$,
we find that 
\[
    \|f_{S,I}\|_\infty
    \; = \; 
    \| X_S f_{S,I}((Z_x)_{x\in I}) \|
    \; \le \; 
    \| X_S f_{S,I}((Z_x)_{x\in I}) \|_{\mathrm{HS}}
    \;\le\; (\gamma/2)^{|I|}2^{|I|/2}
    \; \le \; \gamma^{|I|}.
\]

Finally, we cast the operator $V^{(0)}$ in \eqref{eq: V 0 initial} as a sum of terms corresponding to scale $0$-diagrams: 
\begin{equation}\label{eq: V0}
    V^{(0)} \; = \; \sum_{g\in\mathcal G^{(0)}} V^{(0)}(g), \qquad 
    V^{(0)}(g) \; = \;  X_g f_{I(g),\mathcal A(g)}((Z_x)_{x\in I(g)})
\end{equation} 
with the notation 
\begin{equation}\label{eq: nottion X g}
    X_g \; = \;  \prod_{x\in\mathcal A(g)}X_x,
\end{equation}
which will be used throughout the paper. 
The bound $\|f_{I(g),\mathcal A(g)}\| \le \gamma^{|g|}$ holds since $|I(g)|=|g|$.

\subsection{The Hamiltonian $H^{(k)}$}\label{subsec: Hamiltonian H k}
Let $k\ge 0$.
The Hamiltonian $H^{(k)}$ is decomposed as 
\begin{equation}\label{eq: Hk = Ek + Vk}
	H^{(k)} \; = \; E^{(k)} + V^{(k)}.
\end{equation}
We \emph{assume} that the operator $V^{(k)}$ can be represented as 
\begin{equation}\label{eq: generic form V k}
	V^{(k)} \; = \; \sum_{g\in\mathcal G^{(k)}} V^{(k)}(g)
\end{equation} 
and, for $k\ge 1$, that the operator $E^{(k)}$ takes the form
\begin{equation}\label{eq: form of E k}
    E^{(k)} 
    \; = \; 
    E^{(k-1)} + \sum_{\substack{g\in \mathcal G^{(k-1)} :\\ g \text{ diag and } |g|< L_{k}}} V^{(k-1)}(g).
\end{equation}
It will also be convenient to rewrite this as
\[
    E^{(k)} \; = \; E^{(0)} + \sum_{g \in \mathcal D^{(k)}} E^{(k)}(g)
\]
with $\mathcal D^{(k)}$ as defined in \eqref{eq: participating diagonal diagrams} and with $E^{(k)}(g) = V^{(j-1)}(g)$ for $j\in\{1,\dots,k\}$ such that $g\in \mathcal G^{(j-1)}$.

In addition, we \emph{assume} that the following property holds for any $g\in\mathcal G^{(k)}$:
the operator $V^{(k)}(g)$ is an $X$-monomial of the type 
\begin{equation}\label{eq: form V operator}
    V^{(k)}(g) \; = \;  X_g f((Z_x)_{x\in\overline{I}(g)},(\theta_x)_{x\in I(g)})
\end{equation}
for some function $f$ that depends on $g$, and with $X_g$ defined in \eqref{eq: nottion X g}. 
In particular, a diagonal diagram $g$ yields a diagonal operator $V^{(k)}(g)$.
We notice that, since $\mathcal A(g)\subset I(g)$, the operator  
$V^{(k)}(g)$ is supported in $\overline I(g)$, and depends only on the disorder in $I(g)$. 



By the definitions in Section~\ref{subsec: Hamiltonian H 0}, we already know that all the properties that have been assumed here, 
hold at the scale $k=0$. 
Later in this section, we will show inductively that they hold at all scales.  

\subsection{The Generator $A^{(k+1)}$}
To proceed, we need some information on the changes of basis. 
Given $k\ge 0$, we \emph{define} $A^{(k+1)}$ to be such that 
\begin{equation}\label{eq: definition of A k+1}
    V_{\mathrm{per}}^{(k)} + [A^{(k+1)},E^{(k)}] \; = \; 0
\end{equation}
where 
\begin{equation}\label{eq: definition of V per}
    V_{\mathrm{per}}^{(k)} 
    \; = \;
    \sum_{\substack{g\in\mathcal G^{(k)}:\\ g \text{ off-diag}, |g| < L_{k+1}}} V^{(k)}(g).
\end{equation}
In addition, we \emph{assume} that the operator $A^{(k+1)}$ can be represented as 
\begin{equation}\label{eq: A k+1 sum over triads}
    A^{(k+1)} 
    \; = \; 
    \sum_{t \in \mathcal T^{(k)}} A^{(k+1)}(t)
\end{equation}
and that a property analogous to \eqref{eq: form V operator} holds: 
For all $t \in \mathcal{T}^{(k)}$, 
the operator $A^{(k+1)}(t)$ is an $X$-monomial of the type
\begin{equation}\label{eq: form of A k+1}
    A^{(k+1)}(t) \; = \; 
    X_t  f((Z_x)_{x\in\overline I(t)},(\theta_x)_{x\in I(t)}),
\end{equation}
where $X_t = X_{\mathsf c(t)}$ by definition, for some function $f$ that depends on $t$.
In particular again, the operator $A^{(k+1)}(t)$ is supported on $\overline I(t)$ and depends on the disorder in $I(t)$.
%
%
In Section~\ref{sec: A k+1}, we will provide an explicit representation for $A^{(k+1)}$ and we will show inductively that these properties hold. 

\begin{remark}
    The proof of Properties~\eqref{eq: form V operator} and \eqref{eq: form of A k+1} will be carried out as follows.
    We will prove in Section~\ref{sec: hamiltonian h k+1} that \eqref{eq: form V operator}  and \eqref{eq: form of A k+1} at scale $k$ yield \eqref{eq: form V operator} at scale $k+1$, for $k\ge 0$.
    We will then prove in Section~\ref{sec: second representation A} that  \eqref{eq: form of A k+1} holds at scale $0$ and that \eqref{eq: form of A k+1} at scale $k-1$ together with \eqref{eq: form V operator} at scale $k$ yield \eqref{eq: form of A k+1} at scale $k$, for $k\ge 1$.
\end{remark}

\subsection{The Hamiltonian $H^{(k+1)}$}\label{sec: hamiltonian h k+1}
Knowing the expression for $H^{(k)}$ and $A^{(k+1)}$, 
we can write down an expansion for $H^{(k+1)}$ as defined by \eqref{eq: fundamental relation of the scheme}, the fundamental relation of the scheme. 

Let us introduce the temporary decomposition $V^{(k)} = V^{(k)}_{\mathrm{per}} + V^{(k)}_{\mathrm{non}}$ with $ V^{(k)}_{\mathrm{per}}$ defined in \eqref{eq: definition of V per}. Exploiting the cancellation stemming from the definition \eqref{eq: definition of A k+1}, we compute
\begin{align}
	H^{(k+1)} \; &= \; \e^{\mathrm{ad}_{A^{(k+1)}}}H^{(k)} \; = \; \sum_{n\ge 0} \frac{\mathrm{ad}_{A^{(k+1)}}^n H^{(k)}}{n!}\nonumber\\
	&= \; E^{(k)} 
	+ \sum_{n\ge 1} \frac{n}{(n+1)!}\mathrm{ad}^{n}_{A^{(k+1)}}V_{\mathrm{per}}^{(k)} 
	+ \sum_{n\ge 0} \frac{1}{n!}\mathrm{ad}^{n}_{A^{(k+1)}}V_{\mathrm{non}}^{(k)}. \label{eq: natural expansion for H k+1}
\end{align}
Both $V_{\mathrm{per}}^{(k)}$ and $V_{\mathrm{non}}^{(k)}$ can be expanded as a sum over diagrams and, given $g_0 \in \mathcal G^{(k)}$ and $n\ge1$, we expand
\begin{equation}\label{eq: commutator expansion diagram}
    \mathrm{ad}^{n}_{A^{(k+1)}}V^{(k)}(g_0)
    \; = \; 
    \sum_{t_1,\dots,t_n\in \mathcal T^{(k)}} 
    [A^{(k+1)}(t_n),[\dots,[A^{(k+1)}(t_1),V^{(k)}(g_0)]\dots]]
\end{equation}
where the triads involved in this expansion satisfy the adjacency constraint
\eqref{eq: locality constraint diagrams}. 
This last claim relies on \eqref{eq: form V operator} and \eqref{eq: form of A k+1}, which imply that two operators featuring in this expansion will commute unless an active spin of one of them lies in the support of the other.

Thanks to \eqref{eq: natural expansion for H k+1} and \eqref{eq: commutator expansion diagram} and the above remarks on locality, we are now ready to show that $H^{(k+1)}$ takes the form introduced in Section~\ref{subsec: Hamiltonian H k} and to propagate Property~\eqref{eq: form V operator} from scale $k$ to $k+1$, assuming also that Property~\eqref{eq: form of A k+1} holds at scale $k$, i.e.\@ for $t\in\mathcal T^{(k)}$. 
We consider two cases: 
\begin{enumerate}
    \item 
    First, let $g_0 \in \mathcal G^{(k)}$ be such that $g_0$ is diagonal and $|g_0|<L_{k+1}$. We set 
    \begin{equation}\label{eq: def E k+1}
        E^{(k+1)}(g_0) \; := \; V^{(k)}(g_0). 
    \end{equation}
    \item 
    Second let us consider a diagram $g\in\mathcal G^{(k+1)}$. This diagram takes the form \eqref{eq: concatenation}, and we contemplate two sub-cases.
    If $g_0$ is off-diagonal and satisfies $|g_0|<L_{k+1}$, then we set
    \begin{equation}\label{eq: 1st expression V k+1}
	V^{(k+1)}(g) \; := \; \frac{n}{(n+1)!}  [A^{(k+1)}(t_n),[\dots,[A^{(k+1)}(t_1),V^{(k)}(g_0)]\dots]].
    \end{equation}
    If instead $g_0$ is such that $|g_0|\ge L_{k+1}$, then we set 
    \begin{equation}\label{eq: 2d equation V k+1}
	V^{(k+1)}(g) \; := \; \frac{1}{n!}  [A^{(k+1)}(t_n),[\dots,[A^{(k+1)}(t_1),V^{(k)}(g_0)]\dots]].
\end{equation}
\end{enumerate}
With these definitions, we see that $H^{(k+1)}$ defined by \eqref{eq: natural expansion for H k+1} can be recast as 
\begin{equation}\label{eq: recast expansion H k+1}
    H^{(k+1)} \; = \; 
    E^{(k)} + 
    \sum_{\substack{g\in\mathcal G^{(k)}:\\ g \text{ diag and }|g|<L_{k+1}}}
    E^{(k+1)}(g)+\sum_{g\in\mathcal G^{(k+1)}}V^{(k+1)}(g), 
\end{equation}
i.e.\@ it takes the form introduced in Section~\ref{subsec: Hamiltonian H k}. 

Moreover, the representation \eqref{eq: form V operator} at scale $k+1$ follows from \eqref{eq: form V operator} and \eqref{eq: form of A k+1} at scale $k$.
Indeed, this is obtained by expanding the commutators in \eqref{eq: 1st expression V k+1} or \eqref{eq: 2d equation V k+1} into products, using the (anti-)commutation relations in \eqref{eq: anti commutation generalized} to move all $X$-operators to the left, and applying the recursive definitions of active spins and diagram domains from Section~\ref{subsec: diagrams iteratively}.


\section{Generators}\label{sec: A k+1}

To complete the description of our scheme, let us finally provide a concrete expression for $A^{(k+1)}$ that solves \eqref{eq: definition of A k+1}, for any $k\ge 0$, and check that Property \eqref{eq: A k+1 sum over triads} holds at all scales. 

Before starting, let us introduce a new notation. 
Given a diagonal operator $F$ and a diagram $g$, we let 
\begin{equation}\label{eq: diagram derivative}
    \partial_g F \; = \; X_g F X_g - F.
\end{equation}
We note that this implies the bound 
\begin{equation}\label{eq: bound diagram derivative}
    \|\partial_g F \| \; \le \; 2 \| F \|
\end{equation}
that we will use later on without further mention.
For $F= E^{(k)}$ for some scale $k\ge0$, we also define the expansion
\[
    \partial_g E^{(k)} \; = \; 
    \partial_g E^{(0)} + \sum_{g'\in\mathcal D^{(k)}} (\partial_g E^{(k)})(g')
    \qquad\text{with}\qquad
    (\partial_g E^{(k)})(g') \; := \; \partial_g (E^{(k)}(g'))
\]
and we will drop the outer parenthesis in the sequel.

\subsection{A First Representation of $A^{(k+1)}$}	
As a first way to describe the operator $A^{(k+1)}$, let us write 
\begin{equation}\label{eq: first decomposition A}
    A^{(k+1)} 
    \; = \; 
    \sum_{\substack{g\in\mathcal G^{(k)}:\\ g \text{ off-diag, }|g|<L_{k+1}}} 
    A^{(k+1)}(g)
\end{equation}
where $A^{(k+1)}(g)$ solves the commutator equation
\begin{equation}\label{eq: commutator equation}
	V^{(k)}(g) + [A^{(k+1)}(g),E^{(k)}] \; = \; 0.
\end{equation}
On the one hand, summing \eqref{eq: commutator equation} over diagrams shows that the defining relation \eqref{eq: definition of A k+1} is satisfied. 
On the other hand, 
since $V^{(k)}(g)$ is an $X$-monomial, the relation \eqref{eq: commutator equation} is satisfied if we set 
\begin{equation}\label{eq: general definition for A}
	A^{(k+1)}(g) \; = \; V^{(k)}(g) \frac{1}{\partial_g E^{(k)}}.
\end{equation}

Due to the presence of the denominator $\partial_g E^{(k)}$, we do not expect the operator $A^{(k+1)}(g)$ to be supported in $\overline{I}(g)$; hence, it will not take the form~\eqref{eq: form of A k+1} if we replace $t$ with $g$. This is why we will expand the denominators and obtain a second representation in terms of triads.

\subsection{Expanding Denominators}
Let $k\ge 0$ and let $g\in\mathcal G^{(k)}$ be off-diagonal and such that $|g|< L_{k+1}$. 
To keep notations as light as possible, we will not explicitly write the dependence on $g$ in the expressions below. 
Given integers $u,v \ge 0$, let us consider the following subsets of $\mathcal D^{(k)}$: 
\begin{multline}\label{eq: definition of D(r,s)}
    \mathcal D(u,v) 
    \; = \; 
    \big\{
    g'\in\mathcal D^{(k)} : \mathcal A(g)\cap\overline I(g')\ne\varnothing, \\
    \min I(g') \ge \min I(g) - u, 
    \max I(g') \le \max I(g) + v
    \big\}, 
\end{multline}
as well as the associated truncated denominators 
\begin{equation}\label{eq: definition of Drs}
    D_{u,v} \; = \; \partial_g E^{(0)} + \sum_{g'\in \mathcal D(u,v)} \partial_g E^{(k)}(g'). 
\end{equation}
We define also the sets $\mathcal D(\underline u,v)$ by replacing the constraint $\min I(g')\ge \min I(g) - u$ by $\min I(g')= \min I(g) - u$ in \eqref{eq: definition of D(r,s)}, and the sets $\mathcal D(u,\underline v)$ by replacing the constraint $\max I(g')\le \max I(g) +v$ by $\max I(g') = \max I(g) +v$ in \eqref{eq: definition of D(r,s)}.
We define also 
\begin{equation}\label{eq: expressions for numerators in resolvent}
    D_{\underline u,v} \; = \; 
    \sum_{g'\in \mathcal D(\underline u,v)} \partial_g E^{(k)}(g'), 
    \qquad 
    D_{u,\underline v} \; = \; \sum_{g'\in \mathcal D(u,\underline v)} \partial_g E^{(k)}(g'),
\end{equation}
where a sum over the empty set is taken to be $0$.

With these definitions, we obtain the following expansion: 
\begin{lemma}\label{lem: expansion resolvent}
    For any $R \ge 1$,
    \begin{multline}
    \frac{1}{D_{R,R}}
    \; = \; 
    \frac{1}{D_{0,0}}
    -\sum_{v=1}^R \frac{D_{0,\underline v}}{D_{0,v}D_{0,v-1}}
    -\sum_{u=1}^R \frac{D_{\underline u,R}}{D_{u,0}D_{u-1,0}}\nonumber\\
    +\sum_{1\le u,v\le R} \left(\frac{D_{\underline u,R}D_{u,\underline v}}{D_{u,v}D_{u,v-1}D_{u-1,v}}
    +
    \frac{D_{\underline u,R} D_{u-1,\underline v}}{D_{u,v-1}D_{u-1,v}D_{u-1,v-1}}\right).
    \label{eq: expansion resolvent}
    \end{multline}
\end{lemma}
\begin{proof}
    This follows from successive applications of the
    algebraic identity $\frac{1}{a+b}=\frac1a - \frac{b}{a(a+b)}$ for real numbers $a,b$ such that $a\ne0$ and $a+b\ne0$, 
    a simple form of the resolvent identity.

    We first prove by induction that 
    \[
    \frac{1}{D_{R,R}}
    \; = \; 
    \frac{1}{D_{0,R}} - \sum_{u=1}^R \frac{D_{\underline u,R}}{D_{u,R}D_{u-1,R}}
    \]
    and similarly that 
    \[
    \frac{1}{D_{0,R}}
    \; = \; 
    \frac{1}{D_{0,0}} - \sum_{v=1}^R\frac{D_{0,\underline v}}{D_{0,v}D_{0,v-1}}. 
    \]
    This already yields the two first terms in the claim. 
    To get the three remaining ones, we prove again by {induction} over $R\ge 1$ that 
    \[
    \frac{1}{D_{u,R}D_{u-1,R}}
    \; = \; 
    \frac{1}{D_{u,0}D_{u-1,0}}
    - \sum_{s=1}^R \left(
    \frac{D_{u,\underline v}}{D_{u,v}D_{u,v-1}D_{u-1,v}} 
    + \frac{D_{u-1,\underline v}}{D_{u,v-1}D_{u-1,v}D_{u-1,v-1}} \right)
    \]
    which concludes the proof. 
\end{proof}

\subsection{A Second Representation of $A^{(k+1)}$}\label{sec: second representation A}

We will now start from the representation~\eqref{eq: first decomposition A} of $A^{(k+1)}$ and re-expand $A^{(k+1)}(g)$ for each off-diagonal $g\in\mathcal G^{(k)}$ such that $|g|<L_{k+1}$, as 
\begin{equation}\label{eq: re expansion of A(g)}
    A^{(k+1)}(g) \; = \; 
    \sum_{\substack{g'\in\mathcal L^{(k)}(g),\\g''\in\mathcal R^{(k)}(g,g')}} 
    A^{(k+1)}((g,g',g'')).
\end{equation}
where $\mathcal L^{(k)}(g)$ and $\mathcal R^{(k)}(g,g')$ are defined in \eqref{eq: left gap diagrams} and \eqref{eq: right gap diagrams} respectively.
In this last expression, $(g,g',g'')$ are triads at scale $k$, and the representation~\eqref{eq: first decomposition A} will thus become 
\begin{equation}\label{eq: third decomposition A}
    A^{(k+1)} 
    \; = \; 
    \sum_{t\in \mathcal{T}^{(k)}} A^{(k+1)}(t).
\end{equation}

Let us now define $A^{(k+1)}(t)$ for a triad $t=(g,g',g'')\in\mathcal T^{(k)}$.
Recalling definitions \eqref{eq: definition of D(r,s)} and \eqref{eq: definition of Drs}, we can write 
\[
    \partial_g E^{(k)} = D_{\lfloor L_k \rfloor,\lfloor L_k \rfloor}.
\]
Indeed, if $g'\in\mathcal D^{(k)}$ is such that $\partial_g E^{(k)}(g')\ne 0$, we must have $\mathcal A(g)\cap \overline I(g')\ne \varnothing$ on the one hand, and $|I(g')| \le |g'| < L_k$ on the other hand; hence $g'\in \mathcal D_{\lfloor L_k \rfloor,\lfloor L_k \rfloor}$. 
We may now expand the factor $1/\partial_g E^{(k)}$ featuring in \eqref{eq: general definition for A} as in Lemma~\ref{lem: expansion resolvent} with $R=\lfloor L_k \rfloor$. 
Using the definitions~\eqref{eq: expressions for numerators in resolvent}, we find that the expressions \eqref{eq: general definition for A} and \eqref{eq: re expansion of A(g)} will be equivalent if we set 
\begin{multline}\label{eq: explicit form A k+1 triad}
    A^{(k+1)}(t) 
    \; = \; 
    (-1)^{\delta_{r,0}+\delta_{s,0}}
    V^{(k)}(g) \partial_g E^{(k)}(g') \partial_g E^{(k)}(g'')\\
    \left( 
    \frac{1}{(\partial_g E^{(k)})_{r,s}(\partial_g E^{(k)})_{r,s-1}(\partial_g E^{(k)})_{r-1,s}} + 
    \frac{(1 - \delta_{s,0})
    (1 - \delta_{r,0})
    (1 - \delta_{\min I(g'),\min I(g'')})}
    {(\partial_g E^{(k)})_{r,s-1}(\partial_g E^{(k)})_{r-1,s}(\partial_g E^{(k)})_{r-1,s-1}} \right).
\end{multline}
%
Here,  we have used the notation $(\partial_g E^{(k)})_{u,v}$ instead of $D_{u,v}$ in the denominators for clarity in further uses, i.e.\@ we have rewritten \eqref{eq: definition of Drs} as 
\begin{equation}\label{eq: rewriting denominator}
(\partial_g E^{(k)})_{u,v} 
\; = \; 
\partial_g E^{(0)} + \sum_{g'\in \mathcal D(u,v)} \partial_g E^{(k)}(g'), 
\end{equation}
for $u,v\ge 0$.
In addition,  $r$ and $s$ in \eqref{eq: explicit form A k+1 triad} depend on the triad $t$ and are defined as 
\begin{equation}\label{eq: offset indices}
    r(t) \; = \; \min I(g) - \min I(g'), 
    \qquad 
    s(t) \; = \; \max I(g'') - \max I(g),
\end{equation}
with the convention $\min I(g') = \min I(g)$ if $g'=\varnothing$ 
and $\max I(g'') = \max I(g)$ if $g''=\varnothing$,  
as well as the conventions 
$\partial_g E^{(k)}(\varnothing) = 1$ and $(\partial_gE^{(k)})_{u,-1} = (\partial_gE^{(k)})_{-1,u}=1$ for any integer $u$. 
{We refer to $r(t)$ and $s(t)$ as the left and right \emph{offset indices} of the triad $t$, respectively, and we will often omit the explicit dependence on $t$ for simplicity.}

Let us come to the proof of Property~\eqref{eq: form of A k+1}.
At scale $k=0$, all triads are of the type $t=(g,\varnothing,\varnothing)$ and the expression~\eqref{eq: explicit form A k+1 triad} boils down to $A^{(1)}(t) = V^{(0)}(g)(1/\partial_g E^{(0)})$.   
In this case, Property~\eqref{eq: form of A k+1} follows from the fact that $\partial_g E^{(0)}$ is diagonal, and from the specific form of the energy $E^{(0)}$ in \eqref{eq: bare energy}, which features a deterministic nearest-neighbour coupling.

Let us now consider a scale $k \ge 1$ and show that Property~\eqref{eq: form of A k+1} propagates from scale $k-1$, that is, from $t \in \mathcal{T}^{k-1}$, to scale $k$.  
We also assume that Property~\eqref{eq: form V operator} holds at scale $k$.
Examining the right-hand side of \eqref{eq: explicit form A k+1 triad}, we observe the following:
\begin{enumerate}
    \item
    The only off-diagonal operator that appears is $V^{(k)}(g)$. Therefore, $A^{(k+1)}(t)$ is an $X$-monomial with $\mathcal A(g)$ as set of active spins.

    \item
    The operator $V^{(k)}(g)$ is supported on $\overline{I}(g)$,  
    the operator $\partial_g E^{(k)}(g')$ is supported on $\overline{I}(g')$,  
    and $\partial_g E^{(k)}(g'')$ is supported on $\overline{I}(g'')$.  
    Furthermore, any operator of the form $(\partial_g E^{(k)})_{r,s}$, possibly with shifted indices such as $r-1$ or $s-1$, is supported on the interval $[\min I(g) - r - 1,\, \max I(g) + s + 1]$.  
    Hence, all of these operators are supported on $\overline{I}(t)$, and it follows that $A^{(k+1)}(t)$ depends only on $Z$-operators in $\overline{I}(t)$.

    \item
    By similar reasoning, $A^{(k+1)}(t)$ depends only on the disorder in $I(t)$.  
    Indeed, the denominator $(\partial_g E^{(k)})_{r,s}$ depends only on the disorder in the interval $[\min I(g) - r,\, \max I(g) + s]$ (rather than $[\min I(g) - r - 1,\, \max I(g) + s + 1]$ for its support).
\end{enumerate}
This establishes Property~\eqref{eq: form of A k+1} at scale $k+1$.


\section{Counting Diagrams}\label{sec: counting diagrams}

We develop the tools to control sums over diagrams. 
Let $x\in \Lambda_L$, $k\ge 0$ and $w \ge 1$ be integers. 
We define
    \[
    N(x,k,w) \;=\;
    \sum_{\substack{g \in \mathcal G^{(k)}: \|g\|=w \\ \min I(g)=x}} \frac{1}{g!} 
    \] 
    and the analogue for triads 
    \[
    N_{\caT}(x,k,w) 
    \; =\;  \sum_{\substack{t \in {\caT}^{(k)}: \| t \| =w \\  \min I(t)=x}} \frac{1}{t!}.
    \]
The main  result of this section is 
\begin{proposition}\label{thm: counting diagrams}
    There exists a constant $\Const$ such that, for all $x,k,w$ as above, 
    \[
  N(x,k,w) \leq C^{w}, \qquad   N_{\caT}(x,k,w) \leq w^8 C^{w}  
    \]
\end{proposition}
The remainder of this section is devoted to the proof of proposition \ref{thm: counting diagrams}. 
In what follows, we regularly say that $x,k,w$ are parameters of a diagram $g$, meaning that $k$ is its scale, $w=\|g\|$ and $x=\min I(g)$.

\subsection{From Diagram Counting to Triad Counting}\label{sec: triad diagram counting}

The following lemma shows how the desired bound on $N(\cdot,\cdot,\cdot)$ implies the corresponding bound on $N_{\mathcal{T}}(\cdot,\cdot,\cdot)$.
Let $(C_k)_{k\ge 0}$ be an increasing sequence; see Section~\ref{sec: running constant} below for the specific choice used throughout.
\begin{lemma}\label{lem: combi from diagrams to triads}
Let $k_0 \ge 0$. If
\[
     N(x,k,w) \;\leq\; C_k^{w} 
\]
for any $x$, $w$ and $0\le k \le k_0$, then 
\[
    N_{\caT}(x,k,w) \;\leq\; w^8 C_k^{w}
\]
for any $x$, $w$ and $0\le k \le k_0$ as well. 
\end{lemma}

\begin{proof}
Let $k\le k_0$. 
A triad $t\in {\caT}^{(k)}$ is a triple, consisting of a central diagram ${\mathsf c}(t)$, with parameters $(x_{\mathsf c},k,w_{\mathsf c})$, and gap diagrams $\mathsf l(t)$ and $\mathsf r(t)$ with parameters $(x_{\mathsf l},k_{\mathsf l},w_{\mathsf l}),(x_{\mathsf r},k_{\mathsf r},w_{\mathsf r})$ respectively. 
To fix ideas, we assume that neither $\mathsf{l}(t)$ nor $\mathsf{r}(t)$ is empty.  
The proof can be straightforwardly adapted when one or both of them are empty, by replacing $N(x_{\mathsf{l}}, k_{\mathsf{l}}, w_{\mathsf{l}})$ or $N(x_{\mathsf{r}}, k_{\mathsf{r}}, w_{\mathsf{r}})$, respectively, by $1$ in \eqref{eq: combi triads} below. 
The definitions of $N_\caT(.,.,.)$, $N(.,.,.)$ and $t!$ yield directly
    \begin{equation}\label{eq: combi triads}
    N_{\caT}(x,k,w) 
    \;\leq\; 
    \sum_{\substack{w_{\mathsf c}+w_{\mathsf l}+w_{\mathsf r}=w \\[1mm]  k_{\mathsf l},k_{\mathsf r} <k \\[1mm]  |x_i-x|\leq w, i={\mathsf c},{\mathsf l},{\mathsf r}   }}  
    N(x_{\mathsf c},k,w_{\mathsf c}) N(x_{\mathsf l},k_{\mathsf l},w_{\mathsf l})  N(x_{\mathsf r},k_{\mathsf r},w_{\mathsf r})   .
\end{equation}
The restriction on the $x$-coordinates originates from the fact that $w=\|t\|$ is an upper bound for the cardinality of the domain $I(t)$ of triad $t$.  Since also $k\leq \|t\|=w$, we see that the number of possible values of each of the eight parameters in the sum (three $x$-and three $w$-parameters and two $k$-parameters) are bounded by $w$. 
Since the sequence $(C_k)_{k\ge 0}$ is increasing, we conclude that \eqref{eq: combi triads} is bounded by $w^8 C_k^w$.
 \end{proof}

\subsection{Setup of a Running Constant}\label{sec: running constant}
Thanks to Lemma~\ref{lem: combi from diagrams to triads}, it suffices to prove the bound on $N(\cdot,\cdot,\cdot)$ to establish Proposition~\ref{thm: counting diagrams}.
The proof is by induction on the scale $k$.
To carry it out, we introduce an increasing and bounded sequence $(C_k)_{k \ge 0}$ and prove that
\begin{equation}\label{eq: bound counting diagrams with running constant}
    N(x,k,w) \; \le \; C_k^w
\end{equation}
for all $x$, $k$ and $w$.
Since $(C_k)_{k\ge 0}$ is bounded, the proposition will follow.

For a diagram $g$ at scale $k=0$, we have that $\|g\|$ is the size of the interval $I(g)$, see Section~\ref{sec: diagrams at scale zero}. 
The different scale-$0$ diagrams $g$ with fixed $\|g\|$ and $x=\min I(g)$ correspond to different choices of the set $\mathcal A(g) \subset I(g)$. Therefore, we have $N(x,0,w)\leq 2^{w}$ and proposition \ref{thm: counting diagrams} holds with $C_0=2$. 
We henceforth assume that the claim is true up to scale $k$ and we show that there is a $a>0$ (not depending on $k$) such that the claim is true at scale $k+1$  with
\begin{equation}\label{eq: definition running constant}
    C_{k+1} \;=\; L_{k+1}^{\frac{a}{L_{k+1}}} C_{k}
\end{equation}
Since the map $k \mapsto L_k$ grows exponentially, the infinite product $\prod_{k=1}^{\infty} L_k^{\frac{a}{L_k}}$ converges, and therefore the sequence $(C_k)$ is bounded.

\subsection{Preliminary Estimates for Diagram Counting}
We need some additional notation. 
Let us fix some scale $k\ge 0$.
Given a diagram $g = (g_0,t_1,\dots,t_n) \in \mathcal{G}^{(k+1)}$, let  $(x_i,k, w_i)$ for $i=0,\dots,n$ be the parameters of the diagram $g_0$  and the triads $t_{1}, \dots, t_n$, respectively. 
Since $w=w_0+w_1+\ldots+w_n$ and since 
\[
    w_i \; = \; \| t_i \| \;\ge \; \|\mathsf c(t_i) \| \; \ge \; |\mathsf c (t_i)| \; \ge \; L_k,
    \qquad 1\le i \le n,
\]
we obtain the crucial bound
\begin{equation}\label{eq: bound that should be named}
  n \;\leq\; (w-w_0)/L_{k}.
\end{equation}
The diagram $g_0$ must be treated differently from the triads $t_1, \dots, t_n$ in what follows.
This is because we can directly bound $|I(t_i)|$ in terms of the scale: 
by \eqref{eq: bound on extended support triads}, we have $|I(t_i)| \le 3L_{k+1}$, a bound that will be used repeatedly below.
No such bound is available for $g_0$, since $|g_0|$ and $|I(g_0)|$ can be arbitrarily large relative to the scale.
This motivates the definition of a bipartition $(\mathcal{N}_0, \mathcal{N}_1)$ of ${1, \ldots, n}$:
we set $i \in \mathcal{N}_0$ if $t_i$ is adjacent to $g_0$, as defined in~\eqref{eq: locality constraint diagrams}, and $i \in \mathcal{N}_1$ otherwise.


\begin{lemma}\label{lem: sum x} 
Let us consider a set of all diagrams $g = (g_0,t_1,\dots,t_n)\in\mathcal G^{(k+1)}$
such that the diagram $g_0$, the number $n$, the bipartition $(\caN_0,\caN_1)$ and the bare order $w=\|g\|$, are all fixed. 
\begin{enumerate}
    \item 
The number of possible values for $(x_{i})_{ i \in \caN_0}$ is bounded by 
\[
(w_0+4L_{k+1})^{n_0}, \qquad n_0\; = \; |\caN_0|.
\]
\item The number of possible values for $(x_{i})_{ i \in \caN_1}$ is bounded by 
\[
n_1! (12 L_{k+1})^{n_1}, \qquad n_1\; = \; |\caN_1|
\] 
\item The number of possible values for $(w_i)_{1\le i \le n}$ with $w_i=\|t_i\|$, is bounded by 
\[
\frac{(2w)^n}{n!}.
\]
\end{enumerate}
\end{lemma}
\begin{proof}
1.  If a triad $t$ is adjacent to $g_0$, then $\min I(t)$ is not smaller than $\min I(g_0)-3L_{k+1}$ and not larger than $\max I(g_0)+1$.  
The number of possible values for each $x_i$ with $i\in \mathcal N_0$ is hence bounded by $|I(g_0)| + 3 L_{k+1} + 1 \le w_0 + 4 L_{k+1}$, which yields the claim.

2.  
Let first $j_1 = \min \mathcal{N}_1$.  
The triad $t_{j_1}$ is adjacent to a triad $t_j$ with $j \in \mathcal{N}_0$. 
By definition, $t_j$ is adjacent to $g_0$, while $t_{j_1}$ is not adjacent to $g_0$.  
This implies in particular that $I(t_{j_1}) \not\subset I(g_0)$, and we will use this constraint below.
Since $I(t_j)$ can extend at most $3L_{k+1}$ sites to the left or right of $I(g_0)$, and since only the rightmost $3L_{k+1} - 1$ sites in $I(g_0)$ can host $x_{j_1}$, we find that there are at most $12L_{k+1}$ possible locations for $x_{j_1}$.
Next, consider $x_{j_2}$ with $j_2 = \min(\mathcal{N}_1 \setminus \{j_1\})$.  
The presence of $I(t_{j_1})$ adds at most $6L_{k+1}$ possible sites for the location of $x_{j_2}$, so the total number of locations for $x_{j_2}$ is at most $18L_{k+1}$.
Iterating this argument, we find that the number of possible values for $(x_i)_{i \in \mathcal{N}_1}$ is bounded by
\[
    \prod_{j=1}^{n_1} (6 + 6j)L_{k+1} = n_1! \, (12L_{k+1})^{n_1}.
\]




3. We have $w_0+\sum_i w_i=w$. If we keep $w,w_0$ fixed, the number of possible values for $(w_i)_{1\le i \le n}$ is hence  the number of ways the number $w-w_0$ can be written as a sum of $n$ non-zero natural numbers, which is bounded by 
\[
    {w-w_0 + n \choose n} \; = \; \frac{(w-w_0+n)!}{(w-w_0)!n!} \;\leq\;  \frac{(2w)^n}{n!}, 
\]
since $n\leq w-w_0$ by \eqref{eq: bound that should be named}.
\end{proof}

\subsection{Induction Step for Diagram Counting}\label{subsubsec: induction step}

Let us show that if \eqref{eq: bound counting diagrams with running constant} holds up to some scale $k\ge 0$, then it also holds at scale $k+1$, provided $a$ is taken large enough in \eqref{eq: definition running constant}. 
From the definition of the factorial of a diagram in Section~\ref{subsec: diagrams iteratively}, we find 
\[
  N(x,k+1,w) \;\leq\; \sum_{g_0} \frac{1}{g_0!} \, \sum_{0\leq n \leq n_\star}  
  \frac{1}{n!} \sum_{(\caN_0,\caN_1)}  \,  \sum_{(x_i, w_i)_{1\le i \le n}}  \, \prod_{i=1}^n N_{\caT}(x_i,k,w_i) 
\]
where the summations are constrained as follows: 
\begin{enumerate}
    \item 
    $g_0$ has parameters $(x_0,k,w_0)$ with $x \le x_0 \le x + w$. 
    
    \item Using~\eqref{eq: bound that should be named}, we can take
    \[
    n_\star = \frac{w}{L_k}. 
    \]
    
    \item $w_0 + \dots + w_n = w$.

    \item The number of possible values for $(x_i)_{1 \le i \le n}$ and $(w_i)_{1 \le i \le n}$ is limited by Lemma~\ref{lem: sum x}.
\end{enumerate}
Using our recursive hypothesis together with Lemma~\ref{lem: combi from diagrams to triads}, we find the bound 
\begin{multline}\label{eq: main bound at step k+1 for diagram counting}
    N(x,k+1,w)\\
    \; \le \;
    \sum_{g_0} \frac{C_{k}^{w-w_0}}{g_0!}   \,  \sum_{0\leq n \leq n_*}  
  \sum_{(\caN_0,\caN_1)} 
  (w_1 \dots w_n)^8 
  \frac{(2w)^n}{n!}
  \frac{(w_0+4L_{k+1})^{n_0}  n_1! (12L_{k+1})^{n_1}}{n!}.
\end{multline}
To further bound this expression, we observe that, for $A>0$, the function $x \mapsto (eA/x)^x$, defined for $x>0$, is increasing as long as $x < A$.
We derive the following inequalities: 
\begin{enumerate}
     \item 
    Since $n_\star \le w$, 
    \[
    (w_1 \dots w_n)^8
    \; \le \; 
    \left(\frac{w_1 + \dots + w_n}{n}\right)^{8n}
    \; \le \; 
    \left(\frac{w}{n} \right)^{8n}
    \; \le \; 
    \left(\frac{ew}{n_\star} \right)^{8n_\star}.
    \]

    \item Similarly, using Stirling's bound $1/m! \le (e/m)^m$, valid for any integer $m$, we find
    \[
    \frac{(2w)^n}{n!} \; \le \; 
    \left( \frac{2e w }{n} \right)^n
    \; \le \; 
    \left( \frac{2e w }{n_\star} \right)^{n_\star}.
    \]

    \item 
    $1/n! \le (1/n_0!)(1/n_1!)$.

    \item Since $u+v \le 2uv$ as soon as $u,v\ge 1$, we estimate 
    \[
    \frac{(w_0 + 4L_{k+1})^{n_0}}{n_0!}
    \; \le \; 
    2 \frac{w_0^{n_0}}{n_0!} (4L_{k+1})^{n_0}
    \; \le \; 
    2 \left(\frac{e w}{n_\star}\right)^{n_\star} (4L_{k+1})^{n_\star}.
    \]

    \item $(12L_{k+1})^{n_1} \le (12 L_{k+1})^{n_\star}$. 
\end{enumerate}
The summations over $0\le n \le n_\star$ and over pairs $(\mathcal N_0,\mathcal N_1)$ in \eqref{eq: main bound at step k+1 for diagram counting} can now be bounded by $(n_\star+1)2^{n_\star}  \le 2^{2n_\star}$. 
Finally the summation over $g_0$ is bounded as 
\[
    \sum_{g_0} \frac{C_k^{w-w_0}}{g_0!}
    \; \le \; 
    \sum_{x_0,w_0} C_k^w 
    \; \le \; 
    w^2 C_k^w
    \; \le \; C_k^w \left(\frac{e w}{n_\star}\right)^{2n_\star}. 
\]

Gathering all these estimates and inserting them into~\eqref{eq: main bound at step k+1 for diagram counting}, we find that there exist universal constants $C,a_1,a_2$ such that 
\[
    N(x,k+1,w)
    \; \le \; 
    C_k^w C^{n_\star} \left(\frac{w}{n_\star}\right)^{n_\star}
    \; \le \; 
    C_k^w L_{k+1}^{a_1 w/L_{k+1}} L_{k+1}^{a_2 w /L_{k+1}},
\]
which yields the claim with $a = a_1 + a_2$.

\section{Inductive Bounds}\label{sec: inductive bounds}

We state here the main inductive bounds of the scheme for the operators $V^{(k)}(g)$ and $A^{(k+1)}(t)$. These bounds hold provided suitable non-resonance conditions are satisfied.
To simplify some of the writings below, we will often omit the superscript $k$ in some expressions: 
we will use $V(g)$ instead of $V^{(k)}(g)$ for $g\in\mathcal G^{(k)}$, $A(t)$ instead of $A^{(k+1)}(t)$ for $t\in\mathcal T^{(k)}$, and $\partial_g E$ instead of $\partial_g E^{(k)}$ for $g\in\mathcal G^{(k)}$.



\subsection{Main Results}\label{subsec: main results on inductive bounds}
Let $k\ge 0$ and let $t=(g,g',g'')\in\mathcal T^{(k)}$. 
It is now time to adopt a more compact notation to describe the operator $A^{(k+1)}(t)$ defined in \eqref{eq: explicit form A k+1 triad}. 
We define the operators
\begin{equation}\label{eq: denominator operators 1 and 2}
\begin{cases}
    D_1(t) \; = \; 
        (\partial_g E)_{r,s}\times(\partial_g E)_{r,s-1}\times(\partial_g E)_{r-1,s}, \\
	D_2(t) \; = \; 
        (\partial_g E)_{r,s-1}\times(\partial_g E)_{r-1,s}\times(\partial_g E)_{r-1,s-1}, 
\end{cases}
\end{equation}
as well as 
\begin{equation}\label{eq: denominator operator}
    R(t) \; = \; \frac1{D_1(t)} + (1-\delta_{s,0})(1-\delta_{r,0})
    (1 - \delta_{\min I(g'),\min I(g'')})
    \frac1{D_2 (t)}
\end{equation}
where $r=r(t)$ and $s=s(t)$ are, respectively, the left and right offset indices of the triad $t$.

Let us define some classes of events, i.e.\@ subsets of the sample space $\Omega=[0,1]^L$, corresponding to \emph{non-resonance conditions}.  
First,
\begin{equation}\label{eq: NR I}
    \mathrm{NR}_\mathrm{I}(t)
    \; = \; 
    \left\{ \left\|\frac{1}{D_1(t)}\right\|,\left\|\frac{1}{D_2(t)}\right\| \;\le\; \varepsilon^{-|g|} \right\}.
\end{equation}
Second, if $g$ is non-crowded,
\begin{equation}\label{eq: NR II}
    \mathrm{NR}_{\mathrm{II}}(t)
    \; = \;
    \left\{
    \| A^{(k)}(t)\|
    \; \le \; 
    B_{\mathrm{II}}(t)
    \right\} 
    \quad \text{with} \quad
   	B_{\mathrm{II}}(t) = \frac{1}{4t!} \delta^{\|t\| - |t|} \left(\frac{\gamma}{\varepsilon}\right)^{|t|} 
\end{equation}
while, if $g$ is crowded, we set $\mathrm{NR}_{\mathrm{II}}(t)=\Omega$.

Finally, given a scale $k\ge 1$ and an interval $J\subset \Lambda_L$, we define 
\begin{equation}\label{eq: full NR triads}
    \mathbf{NR}_{<k}(J)
    \; = \; 
    \left\{
    \mathrm{NR}_{\mathrm{I}}(t) \text{ and }\mathrm{NR}_{\mathrm{II}}(t) \; 
    \forall t\in \mathcal T^{(k')} \text{ with }
    k'< k \text{ and } I(t) \subset J
    \right\}.
\end{equation}

We notice that the events $\mathrm{NR}_{\mathrm{I}}(t)$ and $\mathrm{NR}_{\mathrm{II}}(t)$ only depend on the disorder inside $I(t)$,
and that the event $\mathbf{NR}_{<k}(J)$ only depends on the disorder inside $J$. 
Given a diagram $g\in\mathcal G^{(k)}$ or a triad $t\in\mathcal T^{(k)}$, we define $\mathbf{NR}(g)$ and $\mathbf{NR}(t)$ as $\mathbf{NR}_{<k}(I(g))$ and $\mathbf{NR}_{<k}(I(t))$ respectively, for $k\ge 1$, 
and we set $\mathbf{NR}(g)=\mathbf{NR}(t)=\Omega$ if $k=0$. 


Recalling that the bounds \eqref{eq: epsilon delta 1st} to \eqref{eq: epsilon delta 3rd} are assumed to hold, we state

\begin{proposition}\label{pro: main inductive bounds}
    Let $k\ge 0$ and let $g\in\mathcal G^{(k)}$. 
    There exists an operator $\widetilde V^{(k)}(g)$ that coincides with $V^{(k)}(g)$ on $\mathbf{NR}(g)$, 
    and which is smooth as a function of the disorder $\theta$ and satisfies the following bounds on the whole sample space $\Omega$: 
    \begin{align}
	&\left\|\widetilde V^{(k)}(g) \right\| 
	\; \le \; \frac{1}{g!} \delta^{\|g\|-|g|}\left(\frac{\gamma}{\varepsilon}\right)^{|g|}, \label{eq: inductive bound V tilde}\\
	&\left\|\frac{\partial \widetilde V^{(k)}(g)}{\partial \theta_x} \right\| 
	\; \le \; \frac{1}{g!} \delta^{\|g\|-|g|}\left(\frac{\gamma}{\varepsilon^{1+b}}\right)^{|g|}, 
	\qquad \forall x \in \Lambda_L, \label{eq: inductive bound derivative V tilde}
    \end{align}
    with 
    \[
    b \; = \; 9.
    \]
    Similarly, if $k\ge 0$ and if $t=(g,g',g'')\in\mathcal T^{(k)}$, 
    there exists an operator $\widetilde A^{(k+1)}(t)$ that coincides with $A^{(k+1)}(t)$ on $\mathrm{NR}_{\mathrm{I}}(t) \cap \mathrm{NR}_{\mathrm{II}}(t) \cap \mathbf{NR}(t)$, 
    and which is smooth and satisfies the following bounds on the whole sample space $\Omega$: 
	\begin{align}
		 &\left\|\widetilde A^{(k+1)}(t)\right\| \; \le \; \frac1{2 t!}\delta^{\|t\|-|t|}\gamma^{|t|} 
		 \qquad  \text{if $g$ is crowded}, \label{eq: inductive bound A tilde crowded}\\
		 &\left\|\widetilde A^{(k+1)}(t)\right\| \; \le \; \frac1{2 t!}\delta^{\|t\|-|t|}\left(\frac{\gamma}{\varepsilon}\right)^{|t|} 
		 \qquad  \text{if $g$ is non-crowded.} \label{eq: inductive bound A tilde non crowded}
	\end{align}
    In addition, Properties~\eqref{eq: form V operator} and \eqref{eq: form of A k+1} continue to hold for the tilded operators. That is, $\widetilde V^{(k)}(g)$ and $\widetilde A^{(k+1)}(t)$ are $X$-monomials with the same active spins, the same support, and the same dependence on the disorder as $V^{(k)}(g)$ and $A^{(k+1)}(t)$, respectively.
\end{proposition}

\begin{remark}\label{rem: scale zero v bounds}
    The bound \eqref{eq: inductive bound V tilde} can be improved at scale $k=0$.
    As we will see, $\widetilde V^{(0)}(g) = V^{(0)}(g)$ for $g\in\mathcal G^{(0)}$. 
    Hence, the claims in Section~\ref{subsec: Hamiltonian H 0} imply the stronger bound  
    \begin{equation}\label{eq: special k=0 bound}
        \left\| \widetilde V^{(0)} (g) \right\|
        \; \le \; \gamma^{ |g| } \; = \; \gamma^{\|g\|}.
    \end{equation}
\end{remark}


Given the above proposition, we may define the tilded version of other observables. 
In particular, we define $\widetilde E^{(k)}$ for $k\ge 1$ by 
\begin{equation}\label{eq: def tilded energies}
    \widetilde E^{(k)} \; = \; E^{(0)} + \sum_{g\in \mathcal D^{(k)}} \widetilde{E}^{(k)}(g)
\end{equation}
with $\widetilde E^{(k)}(g) = \widetilde V^{(j-1)}(g)$ for $j\in\{1,\dots,k\}$ such that $g\in \mathcal G^{(j-1)}$. 
{Note that since $E^{(0)}$ is already a smooth function of the disorder, there is no need of introducing $\widetilde E^{(0)}$.}
Similarly, given a triad $t=(g,g',g'')$, we define $(\partial_g \widetilde E)_{r,s}$  by replacing $E$ with $\widetilde E$, and then $\widetilde D_1(t)$, $\widetilde D_2(t)$ and $\widetilde R(t)$ by using this smooth variable in \eqref{eq: denominator operators 1 and 2} and \eqref{eq: denominator operator}.
With these definitions, we can state

\begin{corollary}\label{cor: derivative tilded energy}
    Let $k\ge 0$, let $t=(g,g',g'')\in\mathcal T^{(k)}$ and let $x\in\Lambda_L$.
    There exists a constant $\Const$ such that the following estimates hold on the whole sample space $\Omega$: 
    \begin{align}
        &\left\|\frac{\partial}{\partial\theta_x}(\partial_g \widetilde E)_{r,s} 
        - \frac{\partial}{\partial\theta_x} \partial_g E^{(0)}\right\|
        \; \le \; \Const \delta, 
        \label{eq: corollary disorder dependence}\\
        &{\left\| \frac{\partial}{\partial\theta_x} (\partial_g \widetilde E)_{r,s} \right\|
        \; \le \; 3 (\Const \delta)^{d(x,I(g))},}
        \label{eq: corollary disorder dependence distance}
    \end{align}
    with the offset indices $r,s$ determined by the triad $t$, and with $d(x,I(g))=\min\{|x-y|,y\in I(g)\}$.
\end{corollary}
\begin{proof}[Proof of Corollary~\ref{cor: derivative tilded energy}]
    Let us first prove~\eqref{eq: corollary disorder dependence}.
    Using \eqref{eq: def tilded energies},
    we write 
    \[
        \left\|\frac{\partial}{\partial\theta_x}(\partial_g \widetilde E)_{r,s} 
        - \frac{\partial}{\partial\theta_x} \partial_g E^{(0)}\right\|
        \; \le \; 
        2\sum_{{h\in\mathcal D^{(k)}: x\in{I}(h)}} \left\|\frac{\partial \widetilde V(h)}{\partial \theta_x}\right\|.
    \]
    Thanks to Proposition~\ref{pro: main inductive bounds} and to the relation~\eqref{eq: epsilon delta 1st}, 
    this last sum is upper-bounded as
    \[
        \sum_{h\in\mathcal D^{(k)}: x\in{I}(h)} \frac{\delta^{\|h\|}}{h!} 
        \; \le \; 
        \sum_{j=0}^\infty \sum_{w=L_j}^\infty 
        \sum_{\substack{y\in\Lambda_L :\\|y-x|\le w}}
        \sum_{\substack{h\in\mathcal G^{(j)},\\\|h\|=w,\\ \min I(h)=y}} \frac{\delta^w}{h!}
        \;\le\;
        C \delta, 
    \]
    where the last bound follows from Proposition~\ref{thm: counting diagrams} in Section \ref{sec: counting diagrams},
    provided $\delta$ has been taken small enough. 

    Let us come to \eqref{eq: corollary disorder dependence distance}. 
    If $x\in I(g)$, and thus $d(x,I(g))=0$, this is a consequence of \eqref{eq: corollary disorder dependence} since 
    \begin{equation} \label{eq: partial g of e zero}
        \frac{\partial}{\partial\theta_x} \partial_g E^{(0)} =  -2 Z_x 1_{\{x\in\mathcal A(g)\}}.
    \end{equation}
    Otherwise, we find 
    \[
        \left\| \frac{\partial}{\partial\theta_x} (\partial_g \widetilde E)_{r,s} \right\|
        \;\le\;
        2 \sum_{\substack{h\in\mathcal D^{(k)}:\\ x\in I(h),\overline I(h)\cap I(g)\ne \varnothing}}
        \left\|\frac{\partial \widetilde V(h)}{\partial \theta_x}\right\|
        \; \le \;
        2 \sum_{\substack{h\in\mathcal D^{(k)}:\\ x\in I(h),\overline I(h)\cap I(g)\ne \varnothing}}
        \frac{\delta^{\|h\|}}{h!}.
    \]
    and the constraints on $h$ impose now $\| h\| \ge |I(h)| \ge d(x,I(g))$.
    From here, the remainder of the proof is completed as in the first part. 
\end{proof}

\subsection{Construction of the Tilded Operators}\label{sec: construction tilded operators}

We construct the tilded operators inductively on the scale. 
The induction hypothesis on scale $k\ge 0$ is: 
For all $g \in \mathcal G^{(k)}$,  $\widetilde V(g)$ is {smooth on $\Omega$ and} defined such that $V(g)=\widetilde V(g)$ on $\mathbf{NR}(g)$. 
{The fact that Properties~\eqref{eq: form V operator}~and~\eqref{eq: form of A k+1} continue to hold for the tilded operators will be a byproduct of our construction.}

For scale $k=0$, the operator $V(g)$ is deterministic and the induction hypothesis is hence trivially satisfied by setting $\widetilde V(g) = V(g)$.  Below, we assume that the induction hypothesis is satisfied up to some scale $k\ge 0$, and we extend the definition to the scale $k+1$ such that again $V(g)=\widetilde V(g)$ on $\mathbf{NR}(g)$ for $g \in \mathcal G^{(k+1)}$.  Moreover, we will construct $\widetilde A(t)$ for $t\in \mathcal T^{(k)}$ such that $A(t)=\widetilde A(t)$ on $\mathrm{NR}_{\mathrm{I}}(t) \cap \mathrm{NR}_{\mathrm{II}}(t) \cap \mathbf{NR}(t)$.  

\begin{remark}
    The induction hypothesis does not require any information on $\widetilde A(t')$ for triads $t'$ on lower scales. 
    This is ultimately due to the fact that $A(t)$ for $t \in \mathcal T^{(k)}$ is constructed entirely out of $V(g')$ for diagrams $g'$ at scales smaller or equal to $k$.
    Likewise, we construct $\widetilde A(t)$ on scale $k$ by using $\widetilde V(g')$ for $g'$ at scales smaller or equal to $k$. 
\end{remark}

Let $g = (t_0, t_1, \dots, t_n) \in \mathcal{G}^{(k+1)}$, and let $t_i = (g_i, g_i', g_i'')$ for $1 \le i \le n$.  
For each $1 \le i \le n$, we define an operator $A'(t_i)$, which can be thought of as a smooth extrapolation in disorder space of the operator $A(t_i)$ defined in~\eqref{eq: explicit form A k+1 triad}.  
The operator $A'(t_i)$ is given by
\begin{equation}\label{eq: def A i prime}
    \langle A'(t_i)|\sigma\rangle
    \; = \; 
    (-1)^{\delta_{r,0} + \delta_{s,0}}
    \langle \widetilde V(g_i) \, \partial_{g_i} \widetilde E(g_i') \, \partial_{g_i} \widetilde E(g_i'') \, \widetilde R(t_i) \, | \sigma \rangle \, \mathrm{S}_{\mathrm{I}}(t_i, \sigma).
\end{equation}
Here, $\mathrm{S}_{\mathrm{I}}(t_i, \sigma)$ can be viewed as a smooth approximation of the indicator function of $\mathrm{NR}_{\mathrm{I}}(t_i)$, tailored to the configuration $\sigma$.  
More precisely, it satisfies $0 \le \mathrm{S}_{\mathrm{I}}(t_i, \sigma) \le 1$, along with the following properties:
\begin{enumerate}
	
	\item\label{item: property 1 NRI}
	$\mathrm S_\mathrm{I}(t_i,\sigma) = 1$ on $\mathrm{NR}_{\mathrm{I}}(t_i)\cap \mathbf{NR}(t_i)$. 
	
	\item\label{item: property 2 NRI}
	If $\mathrm S_\mathrm{I}(t_i,\sigma) > 0$, 
    then $|\langle\widetilde D_1(t_i)|\sigma\rangle|,|\langle\widetilde D_2(t_i)|\sigma\rangle| \ge \frac12 \varepsilon^{|g_i|}$.
	
	\item \label{item: property 3 NRI}
	$\mathrm S_\mathrm{I}(t_i,\sigma)$ is smooth on $\Omega$ and for every $x\in\Lambda_L$, 
	\[
		\left|\frac{\partial \mathrm S_\mathrm{I}(t_i,\sigma)}{\partial \theta_x}\right| \; \le \; \Const \frac{B_0(t_i) |g_i|}{\varepsilon^{|g_i|}}
	\]
	with $B_0(t_i) = \max_{y\in\Lambda_L}\{\|\partial \widetilde D_1(t_i)/\partial{\theta_y}\|,\|\partial\widetilde D_2(t_i)/\partial {\theta_y}\|\}$.
\end{enumerate}
The existence of the function $\mathrm S_\mathrm{I}(t_i,\sigma)$ is guaranteed by Lemma~\ref{lem: abstract lemma smoothing} in Appendix~\ref{sec: smoothing lemma}, taking $p=2$, $f_j = \langle\widetilde D_j(t_i)|\sigma\rangle$ for $j=1,2$, $\eta=\varepsilon^{|g_i|}$ and letting $\mathrm S_{\mathrm I}(t_i,\sigma)$ being given by the function called $S$ there.
Property~\ref{item: property 1 NRI} above follows from the fact that $\langle\widetilde D_{j}(t_i)|\sigma\rangle = \langle D(t_i)|\sigma\rangle$ for $j=1,2$ on $\mathbf{NR}(t_i)$ by the induction hypothesis, that $|\langle D_j(t_i)|\sigma\rangle|\ge \varepsilon^{|g_i|}$ on $\mathrm{NR}_{\mathrm I}(t_i)$, and from the third item in Lemma~\ref{lem: abstract lemma smoothing}.
The other properties are direct consequences of Lemma~\ref{lem: abstract lemma smoothing}.
We remark that, thanks to our inductive hypothesis and Property~\ref{item: property 2 NRI} above, $A'(t_i)$ is well defined and smooth on the whole sample space $\Omega$. 

Second, for all $1\le i \le n$, we define the operator $\widetilde A (t_i)$, that takes the non-resonance condition $\mathrm{NR}_{\mathrm{II}}(t_i)$ into account. 
If $g_i$ is crowded, we set $\widetilde A (t_i) = A'(t_i)$, while if $g_i$ is non-crowded, we set  
\begin{equation}\label{eq: def A t i tilde}
	\langle \widetilde A (t_i)|\sigma\rangle
	\; = \; 
	\langle A'(t_i)|\sigma\rangle \, \mathrm S_{\mathrm{II}}(t_i,\sigma).
\end{equation}
Here $\mathrm S_\mathrm{II}(t_i,\sigma)$ can be seen as a smooth approximation of the indicator of $\mathrm{NR}_{\mathrm{II}}(t_i)$ tailored to the configuration $\sigma$.  It satisfies $0 \le \mathrm S_\mathrm{II}(t_i,\sigma) \le 1$ as well as the following properties: 
\begin{enumerate}
	
	\item\label{item: property 1 NRII}
	$\mathrm S_\mathrm{II}(t_i,\sigma) = 1$ on $\mathrm{NR}_{\mathrm{II}}(t_i){\cap\mathrm{NR}_{\mathrm{I}}(t_i)\cap \mathbf{NR}(t_i)}$. 
	
	\item\label{item: property 2 NRII}
	If $\mathrm S_\mathrm{II}(t_i,\sigma) > 0$, then $|\langle A'_i|\sigma\rangle | \le 2 B_{\mathrm{II}}(t_i)$ with $B_{\mathrm{II}}(t_i)$ defined in \eqref{eq: NR II}.
	
	\item\label{item: property 3 NRII}
	$\mathrm S_\mathrm{II}(t_i,\sigma)$ is smooth on $\Omega$ and for every $x\in\Lambda_L$, 
	\[
		\left|\frac{\partial \mathrm S_\mathrm{II}(t_i,\sigma)}{\partial \theta_x}\right| \; \le \; \Const \frac{B'(t_i) |t_i|}{B_{\mathrm{II}}(t_i)}
	\]
	with $B'(t_i)=\max_{y\in\Lambda_L}\{\|\partial A'(t_i)/\partial {\theta_y}\|\}$.
\end{enumerate}
The existence of the function $S_\mathrm{II}(t_i,\sigma)$ is guaranteed by Lemma~\ref{lem: abstract lemma smoothing} in Appendix~\ref{sec: smoothing lemma},
{
taking $p=1$, $f_1 = \langle A'(t_i)|\sigma\rangle$, $\eta=B_{\mathrm{II}}(t_i)$ and letting $\mathrm S_{\mathrm{II}}(t_i,\sigma)$ being given by the function called $Q$ there.
Property~\ref{item: property 1 NRII} above follows from the fact that $\langle A'(t_i)|\sigma\rangle = \langle A(t_i)|\sigma\rangle$ on
$\mathbf{NR}(t_i)\cap \mathrm{NR}_{\mathrm I}(t_i)$ by the induction hypothesis and Property~\ref{item: property 1 NRI} of $\mathrm{S_{\mathrm I}}(t_i,\sigma)$ above, 
that $|\langle A'(t_i)|\sigma\rangle|\le B_{\mathrm{II}}(t_i)$ on 
$\mathrm{NR}_{\mathrm{II}}(t_i)$, and from the third item in Lemma~\ref{lem: abstract lemma smoothing}.
The other properties are direct consequences of Lemma~\ref{lem: abstract lemma smoothing}.}

Finally,
paralleling \eqref{eq: 1st expression V k+1} or \eqref{eq: 2d equation V k+1}, we define 
\begin{equation}\label{eq: 1st expression V k+1 tilde}
	\widetilde V(g)
	\; = \; 
	\frac{n}{(n+1)!}[\widetilde A(t_n),[\dots,[\widetilde A(t_1),\widetilde V(g_0)]]
\end{equation}
if $g_0$ is off-diagonal and satisfies $|g_0|<L_{k+1}$, and 
\begin{equation}\label{eq: 2d expression V k+1 tilde}
	\widetilde V(g)
	\; = \; 
	\frac{1}{n!}[\widetilde A(t_n),[\dots,[\widetilde A(t_1),\widetilde V(g_0)]]
\end{equation}
if $g_0$ is such that $|g_0|\ge L_{k+1}$. 
Our construction implies that $\widetilde A(t_1),\dots,\widetilde A(t_n)$ and $\widetilde V(g)$ are well-defined and smooth on the whole sample space $\Omega$. 
Moreover, thanks to our inductive hypothesis and Property~\ref{item: property 1 NRI} of the functions $\mathrm S_\mathrm{I}$ and $\mathrm S_\mathrm{II}$, it guarantees that 
$\widetilde V(g)$ coincides with $V(g)$ on $\mathbf{NR}(g)$
and that $\widetilde A(t_i)$ coincide with $A(t_i)$ on $\mathrm{NR}_{\mathrm{I}}(t_i) \cap \mathrm{NR}_{\mathrm{II}}(t_i) \cap \mathbf{NR}(t_i)$ for $1\le i \le n$.

\subsection{Bounds on the Tilded Operators}\label{sec: bound on tilded operators}

Having constructed the tilded operators involved in Proposition~\ref{pro: main inductive bounds}, we are now ready for the proof of this proposition. We start with the

\begin{proof}[Proof of \eqref{eq: inductive bound A tilde non crowded}]
    Starting from the representation \eqref{eq: def A t i tilde} and using Property~\ref{item: property 2 NRII} of the smoothed indicator $\mathrm{S}_{\mathrm{II}}(t)$, we find 
    \[
    	|\langle  \widetilde A(t)|\sigma\rangle| 
		\; = \;
		|\langle  A'|\sigma\rangle| \mathrm S_{\mathrm{II}}(t,\sigma)
		\; \le \; 
		2 B_{\mathrm{II}}(t) 
    \]
    which yields the required bound, from the definition of $B_{\mathrm{II}}(t)$ in \eqref{eq: NR II}.
\end{proof}

We next come to the 

\begin{proof}[Proof of \eqref{eq: inductive bound V tilde} and \eqref{eq: inductive bound A tilde crowded}]
The proof goes by induction on $k\ge 0$.
For $k=0$, the bound \eqref{eq: inductive bound V tilde} follows from the definition \eqref{eq: V0} of $V^{(0)}$ and the fact that $\widetilde V^{(0)} = V^{(0)}$. 
All diagrams in $\mathcal G^{(0)}$ are non-crowded and \eqref{eq: inductive bound A tilde non crowded} follows from Property~\ref{item: property 2 NRII} of the smoothed indicator $\mathrm S_{\mathrm{II}}$.

Next, we show that the validity of \eqref{eq: inductive bound V tilde} up to scale $k\ge 0$, together with \eqref{eq: inductive bound A tilde non crowded} proved above, implies \eqref{eq: inductive bound A tilde crowded} at the same scale $k$.
Let $t = (g,g',g'')$.
We start from the explicit representation \eqref{eq: def A i prime} and \eqref{eq: def A t i tilde} and derive the bound 
    \begin{equation}\label{eq: bound on A k+1 inside proof}
        |\langle\widetilde A^{(k+1)}(t)|\sigma\rangle|
        \; \le \; 
        \|\widetilde V^{(k)}(g)\| \|\partial_g \widetilde E^{(k)}(g')\| \|\partial_g \widetilde E^{(k)}(g'')\|
        |\langle\widetilde R(t)|\sigma\rangle| \mathrm S_{\mathrm I}(t,\sigma). 
    \end{equation}
    The first three factors on the right-hand side are estimated using the inductive assumption \eqref{eq: inductive bound V tilde} that is now valid up to scale $k$. The derivatives $\partial_g$ give rise to two factors of $2$ by the estimate \eqref{eq: bound diagram derivative}. 
    The product of the fourth and fifth factor is upper-bounded by $4\varepsilon^{-|g|}$ thanks to Property~\ref{item: property 2 NRI} of the smoothed indicator $\mathrm S_{\mathrm I}$. 
    This yields
    \[
        \|\widetilde A^{(k+1)}(t)\|
        \; \le \; 
        \frac{16}{\varepsilon^{|g|}}
        \frac1{g!g'!g''!}
        \delta^{\|g\|+\|g'\|+\|g''\| - (|g| + |g'| + |g''|)}
        \left(\frac{\gamma}{\varepsilon}\right)^{|g|+|g'|+|g''|}.
    \]
    In this expression, we notice that the exponent of $\delta$ is $\|t\| - |t| - (|g| - |g|_{\mathrm r})$, that the exponent of $\gamma$ is $|t| + (|g| - |g|_{\mathrm r})$ and that the exponent of $1/\varepsilon$ is bounded by $4|g|$, since $|g'|,|g''|\le |g|$. Hence
    \[
        \|\widetilde A^{(k+1)}(t)\|
        \; \le \;
        \frac{1}{2t!} \delta^{\|t\|-|t|} \gamma^{|t|}
        \frac{32 \gamma^{|g|-|g|_{\mathrm r}}}{\delta^{|g|-|g|_{\mathrm r}}\varepsilon^{4|g|}}.
    \]
    To conclude, let us show that the last factor in this expression is bounded by $1$.
    By \eqref{eq: bound g prime and g}, we find that $|g|-|g|_{\mathrm r}\ge (1-\beta)|g|$, hence it is bounded by 
    \begin{equation}\label{eq: first smallness gamma}
        32\left(\frac{\gamma}{\delta}\right)^{(1-\beta)|g|}\frac{1}{\varepsilon^{4|g|}}
        \; = \; 
        32\left( \frac{\gamma^{1-\beta}}{\delta^{1-\beta}\varepsilon^4}\right)^{|g|} \; \le \; 1
    \end{equation}
    which holds thanks to the relation \eqref{eq: epsilon delta 2nd}.

Finally, let us show that \eqref{eq: inductive bound V tilde} propagates at scale $k+1$. 
Starting from the explicit representations \eqref{eq: 1st expression V k+1 tilde} or \eqref{eq: 2d expression V k+1 tilde}, 
expanding the nested commutators of $n+1$ operators as a sum of $2^n$ products of operators, 
and using the bounds \eqref{eq: inductive bound A tilde crowded} or \eqref{eq: inductive bound A tilde non crowded} that hold at scale $k$, we get 
\[
	\|\widetilde V^{(k+1)}(g)\| \;\le\;
        \frac{2^n}{n!}
        \frac{1}{2^n g_0! t_1 ! \dots t_n!}
        \delta^{\|g_0\| + \|t_1\| + \dots +\|t_n\|
        - (|g_0|+|t_1|+\dots +|t_n|)}
        \left(\frac{\gamma}{\varepsilon}\right)^{|g_0|+|t_1|+\dots +|t_n|}
\]
where we have used the (crude) bound $\gamma^{|t|}\le (\gamma/\varepsilon)^{|t|}$ whenever we used \eqref{eq: inductive bound A tilde crowded}. 
This yields \eqref{eq: inductive bound V tilde} at scale $k+1$. 
\end{proof}

\subsection{Bounds on the Derivatives of Tilded Operators}\label{sec: bound derivatives operators}

To complete the proof of Proposition~\ref{pro: main inductive bounds}, we are left with the
\begin{proof}[Proof of \eqref{eq: inductive bound derivative V tilde}]
	The proof goes by induction on $k\ge 0$. 
	The claim holds for $k=0$ since $V^{(0)}$ is deterministic. 
	Let us assume that it holds up to scale $k\ge 0$ and let us prove it at scale $k+1$. 
	So, let $g\in\mathcal G^{(k+1)}$ with $g=(t_0,t_1,\dots,t_n)$ for some $n\ge 0$. 
	If $n=0$ the claim follows directly by induction. Thus, assume $n\ge 1$. 
	Starting from the explicit definition 
        \eqref{eq: 1st expression V k+1 tilde} or \eqref{eq: 2d expression V k+1 tilde} yields 
	\begin{equation}\label{eq: norm derivative to start with}
    \left\|\frac{\partial \widetilde V(g)}{\partial \theta_x}\right\|
		\; \le \; 
		\frac{2^n}{n!}\left\|\frac{\partial \widetilde V(g_0)}{\partial \theta_x}\right\|
    \left\|\widetilde A(t_1)\right\|\dots\left\|\widetilde A(t_n)\right\|
  + \frac{2^n}{n!}\sum_{i=1}^n \left\|\widetilde A(t_1)\right\|\dots\left\|\frac{\partial \widetilde A(t_i)}{\partial \theta_x}\right\|\dots\left\|\widetilde A(t_n)\right\|.
    \end{equation}
    
 Let us bound each of the terms on the right-hand side of this expression.
	For the first one, we use our inductive assumption together with the inductive bound \eqref{eq: inductive bound A tilde non crowded} on the $\widetilde A(t_i)$ for $1\le i \le n$, which yields the bound 
	\begin{equation}\label{eq: bound 1st term in proof derivative wrt disorder}
		\frac{1}{\varepsilon^{b|g_0|}} \frac{1}{g!} \delta^{\|g\|-|g|}\left(\frac{\gamma}{\varepsilon}\right)^{|g|}. 
	\end{equation}
	Let us then consider the other terms. 
	Given $1\le i\le n$, let us write $t_i = (g_i,g_i',g_i'')$ and let us start from the explicit representation stemming from \eqref{eq: def A i prime} and \eqref{eq: def A t i tilde}: 
	\begin{equation}\label{eq: decomposition of A in 6 factors}
	\langle \widetilde A(t_i)|\sigma\rangle 
	\; = \;  
	-\langle\widetilde V(g_i) \partial_{g_i} \widetilde E(g'_i)\partial_{g_i} \widetilde E(g_i'') \widetilde R(t_i)|\sigma\rangle \, 
 \mathrm S_{\mathrm{I}}(t_i,\sigma) \mathrm S_{\mathrm{II}}(t_i,\sigma).
	\end{equation} 
	By Leibniz product rule, the derivative of this expression with respect to $\theta_x$ writes as a sum of six terms, 
	and we set 
 \[
 \frac{\partial \langle \widetilde A(t_i)|\sigma\rangle}{\partial\theta_x}  
 \; = \; 
 d_1 + \dots + d_6. 
 \]

 We now estimate these six terms separately. 
 In doing so, we may assume $\mathrm S_{\mathrm{I}}(t_i,\sigma) > 0$ in \eqref{eq: decomposition of A in 6 factors}, which implies $|\langle\widetilde R(t_i)|\sigma\rangle| \le \varepsilon^{-|g_i|}$ by Property~\ref{item: property 2 NRI} of $\mathrm S_{\mathrm{I}}(t_i,\sigma)$. 
	For $d_1,d_2,d_3$, we can use directly our inductive assumption and get 
	\[
		d_1 + d_2 + d_3 
		\; \le \;
		\left(\frac{1}{\varepsilon^{b|g_i|}} + \frac{1}{\varepsilon^{b|g_i'|}} + \frac{1}{\varepsilon^{b|g_i''|}} \right) 
		\frac{C}{\varepsilon^{|g_i|}}
		\frac{1}{t_i!} \delta^{\|t_i\| - |g_i| - |g_i'| - |g_i''|} \left(\frac\gamma\varepsilon\right)^{|g_i|+|g_i'|+|g_i''|}.
	\]
	For $d_4$, we start from the definition~\eqref{eq: denominator operator} and we compute 
    \[
    \left|\frac{\partial \langle\widetilde R(t_i)|\sigma\rangle}{\partial\theta_x}\right|
    \; \le \; 
    \frac{1}{\left|\langle\widetilde D_1(t_i)|\sigma\rangle\right|^2} 
    \left| \frac{\partial \langle\widetilde D_1(t_i)|\sigma\rangle}{\partial\theta_x} \right|
    + 
    \frac{1}{\left|\langle\widetilde D_2(t_i)|\sigma\rangle\right|^2} \left| \frac{\partial \langle\widetilde D_2(t_i)|\sigma\rangle}{\partial\theta_x} \right|.
    \]
    In this expression, the denominators are bounded by $4\varepsilon^{-2|g_i|}$ 
    and the numerators are bounded by a constant, as can be derived from the tilded version of definition~\eqref{eq: definition of D(r,s)} and the bound \eqref{eq: corollary disorder dependence distance} from Corollary~\ref{cor: derivative tilded energy}, valid at the previous scale, assuming the worse possible case $d(x,I(g_i))=0$.
    This yields 
	\[
		d_4 \; \le \; \frac{\Const}{\varepsilon^{2|g_i|}}
		\frac{1}{t_i!} \delta^{\|t_i\| - |g_i| - |g_i'| - |g_i''|} \left(\frac\gamma\varepsilon\right)^{|g_i|+|g_i'|+|g_i''|}.
	\]
    For $d_5$, we use Property~\ref{item: property 3 NRI} of  $\mathrm S_{\mathrm{I}}(t_i,\sigma)$, where we find that $B_0(t_i) \le \Const$ using again \eqref{eq: corollary disorder dependence distance} from Corollary~\ref{cor: derivative tilded energy}. Hence 
	\[
		d_5 \; \le \; \frac{\Const |g_i|}{\varepsilon^{2|g_i|}}
		\frac{1}{t_i!} \delta^{\|t_i\| - |g_i| - |g_i'| - |g_i''|} \left(\frac\gamma\varepsilon\right)^{|g_i|+|g_i'|+|g_i''|}.
	\]
	Finally, for $d_6$, we use Property~\ref{item: property 3 NRII} of $\mathrm S_{\mathrm{II}}(t_i,\sigma)$. 
	Here we notice that $B'(t_i)$ can be estimated by $d_1+\dots+d_5$, and we conclude that $B'(t_i)|g_i|/B_{\mathrm{II}}(t_i) \le \Const |g_i|\varepsilon^{-(b+2)|g_i|}$, hence 
	\[
		d_6 \; \le \; \frac{\Const |g_i||t_j|}{\varepsilon^{(b+3)|g_i|}}
		\frac{1}{t_i!} \delta^{\|t_i\| - |g_i| - |g_i'| - |g_i''|} \left(\frac\gamma\varepsilon\right)^{|g_i|+|g_i'|+|g_i''|}.
	\]

        Collecting the above bounds for $d_1,\dots,d_6$, and using that $|g_i|,|t_i| \le 3 L_{k+1}$, we find 
	\begin{equation}\label{eq: derivative A raw collection of terms}
		\left\|\frac{\partial \widetilde A(t_i)}{\partial \theta_x}\right\| \; \le \; \frac{\Const L_k^2}{\varepsilon^{(b+3)|g_i|}}
		\frac{1}{t_i!} \delta^{\|t_i\| - |g_i| - |g_i'| - |g_i''|} \left(\frac\gamma\varepsilon\right)^{|g_i|+|g_i'|+|g_i''|}.
	\end{equation}
    If $g_i$ is non-crowded, we simply rewrite $|g_i|+|g_i'|+|g_i''|$ as $|t_i|$ in this expression. 
	Instead, if $g_i$ is crowded, we can reason as in the proof of \eqref{eq: inductive bound A tilde crowded} in Section~\ref{sec: bound on tilded operators}, and upper-bound it as
	\begin{equation}\label{eq: bound other terms in proof derivative wrt disorder crowded}
	\left\|\frac{\partial \widetilde A(t_i)}{\partial \theta_x}\right\|
	\; \le \; 
	\frac{\Const L_k^2}{\varepsilon^{b|g_i|}} \frac{1}{t_i!}\delta^{\|t_i\|-|t_i|} \gamma^{|t_i|}.
	\end{equation}
    provided
     \begin{equation}\label{eq: next non trivial constrain on eps and delt}
    \left(\frac{\gamma}{\delta}\right)^{1-\beta}\frac1{\varepsilon^{7}} \;\le\; 1,   
 \end{equation}
 which is \eqref{eq: epsilon delta 3rd}.
 
	Inserting now the bounds \eqref{eq: bound 1st term in proof derivative wrt disorder} as well as \eqref{eq: derivative A raw collection of terms} or \eqref{eq: bound other terms in proof derivative wrt disorder crowded} into \eqref{eq: norm derivative to start with}  yields
	\begin{align*}
	&\left\|\frac{\partial \widetilde V(g)}{\partial \theta_x}\right\|
	\; \le \; 
	\Const L_k^2 \left(\frac{1}{\varepsilon^{b|g_0|}} +  \sum_{\substack{1 \le i \le n:\\g_i \text{ is non-crowded}}} \frac{1}{\varepsilon^{(4+b)|g_i|}} 
	+ \sum_{\substack{1 \le i \le n:\\g_i \text{ is crowded}}} \frac{\varepsilon^{|t_i|}}{\varepsilon^{b|g_i|}} \right) 
	\frac{1}{g!} \delta^{\|g\|-|g|}\left(\frac{\gamma}{\varepsilon}\right)^{|g|}\\
	\; &\le \; 
	\Const L_k^2\left(\varepsilon^{b(|g|-|g_0|)} + \sum_{\substack{1 \le i \le n:\\g_i \text{ is non-crowded}}} \varepsilon^{b|g| - (4+b)|g_i|}
	+ \sum_{\substack{1 \le i \le n:\\g_i \text{ is crowded}}} \varepsilon^{|t_i|+b(|g|-|g_i|)}\right) \\
	&\phantom{\; \le \;}
    \frac{1}{g!} \delta^{\|g\|-|g|}\left(\frac{\gamma}{\varepsilon^{1+b}}\right)^{|g|}.
	\end{align*}
    We thus need to prove that the factor $\Const L_k^2 (\ldots)$ on the first line after the second inequality above is bounded by $1$. 
 Let us derive a lower bound on the three exponents featuring there. 
 First, 
 \[
    b\left(|g| - |g_0| \right) \;\ge\; b |g_1|_{(\mathrm{r})} \;\ge\; b\beta L_k.
 \]
 Second, if $g_i$ is non-crowded for some $1\le i \le n$, 
 \[
    b(|g|-|g_i|) - 4|g_i|
    \;\ge\;
    b(L_k + (n-1)\beta L_k) - 4L_{k+1} 
    \;\ge\; nL_k 
 \]
 provided $b\ge 9$. 
 Finally, if $g_i$ is crowded for some $1 \le i \le n$, 
 \begin{multline*}
    |t_i| + b(|g| - |g_i|)
    \;=\; 
    |t_i| + b(|g| - |g_i|_r) - b(|g_i| - |g_i|_r)\\
    \;\ge\;
    \beta L_k + b (L_k + (n-1)\beta L_k) - b (L_{k+1} - \beta L_k)
    \; \ge \; 
    n\beta L_k.
 \end{multline*}
 We find thus that the pre-factor that we need to estimate is upper-bounded by 
 \[
    C L_k^2 
    \left( \varepsilon^{b\beta L_k} + n\varepsilon^{n L_k} + n\varepsilon^{\beta nL_k}\right)
    \; \le \; 1
 \]
for any $n\ge 1$, provided $\varepsilon$ is taken small enough. 
\end{proof}

\section{Probabilistic Estimates}\label{sec: probabilistic estimates}

We now begin the treatment of the probabilistic estimates, outlined in Section~\ref{sec: probabilistic estimates outline}.
We start with a definition that will be used throughout.
Given a diagram $g$, we say that a diagram descendant $h$ of $g$ is \emph{fully overlapping} if, for each $x\in I(h)$, there exists another descendant $h'$ of $g$ such that $x\in I(h')$ and $h$ and $h'$ are not in hierarchical relation.

\subsection{An Equivalence Relation} \label{sec: equivalence of diagrams}

We introduce an equivalence relation among diagrams and triads whose aim was outlined in Section~\ref{sec: probabilistic estimates outline}, namely grouping diagrams and triads for which the non-resonance condition $ \mathrm{NR_{II}}(t)$ corresponds to the same event.

Let $k\ge 0$. 
We construct the equivalence relation between diagrams in $\mathcal G^{(k)}$ in a recursive way.
First, we consider the decomposition of $g \in \mathcal G^{(k)}$ on the previous scale:
$$
g=(t_0,t_1,\ldots,t_n)
$$
We divide the colleagues $t_0,\ldots,t_n$ into relevant and irrelevant colleagues.  
The irrelevant colleagues are the triads $t_i$ with $1 \le i \le n$ such that one of the following three conditions is satisfied: 
\begin{enumerate}
    \item 
    The diagram $\mathsf c(t_i)$ is crowded.

    \item 
    The diagram $\mathsf c(t_i)$ is fully overlapping as a descendant of $g$.

    \item 
    {The triad $t_i$ is adjacent, cf.~\eqref{eq: locality constraint diagrams}, to a colleague $t_j$, for some $0\le j\ne i \le n$, such that $\mathsf c(t_j)$ is fully overlapping as a descendant of $g$.
    Here, we adopt the convention $\mathsf c(t_j) = t_j$ if $j=0$.}
\end{enumerate}
The relevant colleagues are $t_0$ and any $t_i$ that are not irrelevant for $1 \le i \le n$.
Now we define an auxiliary notion:  We define $O(g)$ as the set of points $x$ such that $x \in I(t_i)\cap \overline I(t_j)$, for some irrelevant colleague $t_i$ and some relevant colleague $t_j$.

Now we construct the equivalence relation: 
If $k=0$, two diagrams are equivalent if and only if they are equal. 
For $k\ge 1$, the equivalence relation is defined recursively, i.e.\@ we assume that the equivalence relation is defined on $\mathcal G^{(j)}$ for $j<k$ and we define it on $\mathcal G^{(k)}$.
Let $g,g' \in \mathcal G^{(k)}$ and let us decompose them into diagram/triads at the previous scale, 
i.e.\@ $g = (t_0,\dots,t_n)$ for some $n\ge 0$ and $g' = (t_0',\dots,t'_{n'})$ for some $n'\ge 0$. 
Furthermore, let $0 = i_0< i_1<\dots<i_m \le n$ be the indices such that $t_{i_j}$ are the relevant colleagues of $g$ and let  $0 = i'_0< i'_1<\dots<i'_{m'} \le n'$ be the indices such that $t'_{i'_j}$ are the relevant colleagues of $g'$.
The diagrams $g$ and $g'$ are equivalent, and we write $g\sim g'$, if and only if 
\begin{enumerate}
	\item 
	$|g| = |g'|$ and $I(g)=I(g')$, 
	
	\item 
	$\mathcal A(g) = \mathcal A(g')$ and $O(g) = O(g')$,
	
	\item
	$n=n'$ and $m=m'$, 
	
	\item 
	$i_l = i'_l$ and $t_{i_l} \sim t'_{i'_l}$ for $0\le l \le m$.
        See \eqref{eq: equivalence triads} below for the meaning of $t_{i_l} \sim t'_{i_l}$.
\end{enumerate}
The equivalence class of $g$ is denoted by $[g]$. 
This equivalence relation can be extended to triads: Given $k\ge 0$, two triads $t,t'\in \mathcal T^{(k)}$ are equivalent, and we write $t\sim t'$, if 
\begin{equation}\label{eq: equivalence triads}
    \mathsf l(t)\sim \mathsf l(t'), \qquad
    \mathsf c(t)\sim \mathsf c(t'), \qquad
    \mathsf r(t)\sim \mathsf r(t').
\end{equation}
The equivalence class of a triad $t$ is denoted by $[t]$.

Since $g\sim g'$ implies that $g$ and $g'$ are at the same scale and that $I(g)=I(g')$, 
the events $\mathbf{NR}(g)$ that were introduced in Section \ref{subsec: main results on inductive bounds} 
depend only on the class $[g]$ for a diagram $g$, and we will write $\mathbf{NR}([g])$. 
Similarly, $\mathbf{NR}(t)$ depend only on the class $[t]$ for a triad $t$, and we will write $\mathbf{NR}([t])$.

\subsection{Probability of Resonances for Equivalence Classes}\label{subsec: bounds on probability of resonances for equivalence classes}
The main claim of this section is 
\begin{proposition}\label{prop: prob of resonances}
    There exists a constant $c>0$ such that 
    \begin{align}
    \mathbb P
    \left(\bigcup_{t\in[t]}\left(\mathrm{NR}_{\mathrm I}(t)\right)^c
    \cap \mathbf{NR}([t])\right)
    \: \le \; \varepsilon^{c |t|}, 
    \label{eq: main probabilistic bound NRI}\\
    \mathbb P
    \left(\bigcup_{t\in[t]}\left(\mathrm{NR}_{\mathrm{II}}(t)\right)^c
    \cap \mathbf{NR}([t])\right)
    \: \le \; \varepsilon^{c |t|}.
    \label{eq: main probabilistic bound NRII}
    \end{align}
\end{proposition}

The main difficulty is to prove~\eqref{eq: main probabilistic bound NRII}. 
A much simpler definition of equivalence class would have sufficed to establish \eqref{eq: main probabilistic bound NRI}, and we can already provide the 

\begin{proof}[Proof of \eqref{eq: main probabilistic bound NRI}.]

We first observe that $\mathrm{NR}_{\mathrm I}(t)$, as defined in \eqref{eq: NR I}, depends on the triad $t$ only through its scale, its offset indices $r,s$, and its set of active spins $\mathcal A(t)$.
As a result, $\bigcup_{t\in[t]}(\mathrm{NR}_{\mathrm I}(t))^c = (\mathrm{NR}_{\mathrm I}(t))^c$ where $t$ is any element of $[t]$.
Given such an element $t=(g,g',g'')$, 
we find thus that the left-hand side of \eqref{eq: main probabilistic bound NRI} is bounded by 
\begin{align}
    &\mathbb P
    \left((\mathrm{NR}_{\mathrm{I}}(t))^c \cap \mathbf{NR}(t) \right)
    \nonumber\\
    \; &\le \;
    \mathbb P \left(\|1/D_1(t)\|> \varepsilon^{-|g|} \cap \mathbf{NR}(t)   \right) 
    + 
    \mathbb P \left(\|1/D_2(t)\|> \varepsilon^{-|g|} \cap \mathbf{NR}(t)   \right) \nonumber\\
    \; &\le \;
    \varepsilon^{\alpha |g|} \mathbb E\left( \left(\|1/ D_1(t)\|^{\alpha} + \|1/D_2(t)\|^{\alpha} \right) 1_{\mathbf{NR}(t)}\right)
    %
    \label{eq: bound markov NR I}
\end{align}
where we have used Markov inequality with some fractional power $0<\alpha<1$ to get the second bound.
The fractional power $\alpha$ will need to be taken small enough below. 
Further, exploiting that $D_i(t)$ is supported on $\overline I(t)$, 
we bound 
\begin{equation}\label{eq: bound Expectation D i}
\mathbb E\left(\|1/D_i(t)\|^{\alpha}1_{\mathbf{NR}(t)}\right) 
\; \le \; 
2^{|\overline I(t)|}
\max_{\sigma\in\{\pm 1\}^L}
\mathbb E\left(|\langle D_i(t)|\sigma\rangle|^{-\alpha} 1_{\mathbf{NR}(t)}\right), \qquad i=1,2.
\end{equation}

To bound this last expression, 
let us fix a spin configuration $\sigma$ as well as some point $x\in\mathcal A(t)$ (we know that this set is non-empty). 
Let us define $\widetilde D_i(t)$ as in the discussion preceding Corollary~\ref{cor: derivative tilded energy}, that is, by replacing $E$ with $\widetilde E$ in the definition given in~\eqref{eq: denominator operators 1 and 2}.
Integrating over the variable $\theta_x$ while keeping all the other variables as parameters, we find 
\[
    \int_0^1 d\theta_x |\langle D_i(t)|\sigma\rangle|^{-\alpha} 1_{\mathbf{NR}(t)}
    \; = \; 
    \int_0^1 d\theta_x |\langle \widetilde D_i(t)|\sigma\rangle|^{-\alpha} 1_{\mathbf{NR}(t)}
    \; \le \; 
    \int_0^1 d\theta_x |\langle \widetilde D_i(t)|\sigma\rangle|^{-\alpha}, 
\]
where we have used that $D_i(t)$ and $\widetilde D_i(t)$ coincide on $\mathbf{NR}(t)$ by Proposition~\ref{pro: main inductive bounds}.
Next, using definition of $\tilde D_i(t)$ and changing integration variables, we find 
\[
    \int_0^1 d\theta_x |\langle \widetilde D_i(t)|\sigma\rangle|^{-\alpha}
    \; \le \; 
    \max_{v,w}\int_0^1
    \frac{d \theta_x}{|\langle(\partial_g \widetilde E)_{v,w}|\sigma\rangle|^{3\alpha}}
    \; = \;
    \max_{v,w}
    \int_{U_{v,w}} \frac{d u}{|u|^{3\alpha}} \left|\frac{(\partial\langle\partial_g \widetilde E)_{v,w}|\sigma\rangle}{\partial \theta_x}\right|^{-1}.
\]
By the bound~\eqref{eq: corollary disorder dependence} in Corollary~\ref{cor: derivative tilded energy} and the relation in \eqref{eq: partial g of e zero}, we find that the Jacobian in the above integral can be bounded by a constant for $\delta$ small enough. 
The whole integral can then be bounded by a constant provided $\alpha < 1/3$. 

Therefore, we find that the right-hand side of \eqref{eq: bound Expectation D i} is bounded by  $C 2^{|\overline I (t)|} \le 4C 2^{|t|}$ for some constant $C$. 
Hence, the right-hand side of \eqref{eq: bound markov NR I} is in turn bounded by $8C \varepsilon^{\alpha |g|}2^{|t|}$. 
This yields the desired result, provided $\varepsilon$ is taken small enough, since $|g|\ge \frac13 (|g|+|g'|+|g''|)\ge |t|/3$.
%
\end{proof}

\subsection{Outline of proof of \eqref{eq: main probabilistic bound NRII}} \label{subsec: outline proba}

The proof of \eqref{eq: main probabilistic bound NRII} will be completed in Sections~\ref{sec: selecting integration variables} to \ref{sec: conclusion proof probability resonance} below, and here we outline the main steps. 


\paragraph{Impossibility of an Inductive Scheme.}
To streamline the discussion, 
we first disregard the union over $t\in [t]$ in \eqref{eq: main probabilistic bound NRII} and we focus on controlling the probability of $\mathrm{NR}^c_{\mathrm{II}}(t)\cap \mathbf{NR}(t)$. This means that we need to bound the matrix element
\begin{equation}\label{eq: central in outline}
    \left|\langle A(t)|\sigma\rangle \right| \, 1_{\mathbf{NR}(t)} 
\end{equation}
for a single triad $t$, with $\mathsf c(t)$ non-crowded, 
and the bound must be uniform in the configuration $\sigma$. 
This expression equals, with $t = (g,g',g'')$,
$$
\left|\langle A(t)|\sigma\rangle \right| \, 1_{\mathbf{NR}(t)} 
\; = \;   
\left|\langle V(g)|\sigma \rangle  \langle\partial_g V(g')|\sigma\rangle  \langle\partial_g V(g'')|\sigma\rangle  \langle R(t)|\sigma\rangle \right|\,
1_{\mathbf{NR}(t)}  
$$
since all operators appearing on the right-hands side, except $V(g)$, are diagonal. 
The most straightforward way to bound this expression, is to use the inductive bounds of Proposition \ref{pro: main inductive bounds} in Section \ref{sec: inductive bounds}, which are available thanks to the presence of the indicator $1_{\mathbf{NR}(t)}$. If we use those bounds for the three $V$-operators, together with
 $|t|=|g|+|g'|+|g''|$ and  $\| t\|=\|g\|+\|g'\|+\|g''\|$, 
 we get
\begin{equation}\label{eq: intermediate loose}
 \left|\langle A(t)|\sigma\rangle\right| \,1_{\mathbf{NR}(t)} 
 \;\leq\; \frac{4 }{t!} \delta^{\| t\|-|t|} \gamma^{|t|}(1/\varepsilon)^{|t|} \left|\langle R(t)|\sigma\rangle \right|\, 1_{\mathbf{NR}(t)} .
\end{equation}
%

To deduce from this bound that the probability of the event \( (\mathrm{NR}_{\mathrm{II}}(t))^c \cap \mathbf{NR}([t])\) is as small as \(\varepsilon^{c|t|}\), we would need to show that \(\max_\sigma |\langle R(t) | \sigma \rangle| = \|R(t)\| \leq \frac{1}{16}\) with probability at least \(1 - \varepsilon^{c|t|}\).
However, examining the denominators in \(R(t)\), this appears implausible; see equations~\eqref{eq: denominator operators 1 and 2} and~\eqref{eq: denominator operator}. Indeed, it's hard to fathom how \(\|R(t)\|\) could generally be significantly smaller than \(2^{|\mathcal A(g)|}\), due to the maximization over all configurations that differ within the set \(\mathcal{A}(g) \subset I(g)\).
{Moreover, as $|\mathcal A (g)|$ can typically be proportional to $|g|$, we should expect $\| R(t) \|$ to grow exponentially with $|g|$.}
In fact, an upper bound for \(\|R(t)\|\) that appears to hold with large probability arises from the non-resonance condition \(\mathrm{NR}_{\mathrm{I}}(t)\), see equation~\eqref{eq: NR I}, which gives \(\|R(t)\| \leq 2 \varepsilon^{-|g|}\).



This difficulty is what was referred to in subsection~\ref{subsec: formal setup and inductive control} as the ``impossibility of an inductive scheme.''  
As we now see, such a procedure introduces too many factors of \( 1/\varepsilon \).

\paragraph{Set-up of a Different Strategy.}
We outline now how to do better. 
Ideally, we would like to use inductive bounds only when they do not introduce a factor of \(1/\varepsilon\), that is, for triads with crowded central diagrams or for \(0\)-scale \(V\)-diagrams; see \eqref{eq: inductive bound A tilde crowded} and \eqref{eq: special k=0 bound}, respectively.
To that end, consider a subtree \(\Upsilon_{\mathrm{sub}}\) of \(\Upsilon(g)\) with root vertex \(g\).  
Starting from the root and moving downward through the subtree, we unpack each operator \(V(g')\) and \(A(t')\), for vertices \(g'\) and \(t'\) in \(\Upsilon_{\mathrm{sub}}\), into a product of the denominator associated with that vertex, if it corresponds to a triad, and lower-scale \(A\) and \(V\) operators. This continues until we reach the leaves.
The subtree \(\Upsilon_{\mathrm{sub}}\) will be chosen to closely implement our strategy, namely to ensure that most leaves do not carry a factor $1/\varepsilon$.



The operation of unpacking operators, or as we call it later, expanding vertices into their children, is set up in Section \ref{sec: expansion of vertices into children}. 
In a certain sense, this unpacking or expansion is simply using the hierarchical construction of diagrams and triads.
In particular, in Section \ref{sec: expansion of diagram into children}, we discuss the expansion of diagrams and in Section \ref{sec: expansion of triads into children} we deal with the expansion of triads.  


With the ingredients of Section \ref{sec: expansion of vertices into children}, we then carry out the unpacking procedure in Section \ref{sec: main bound matrix elements}  for a specific choice of $\Upsilon_{\mathrm{sub}}$, namely 
the subtree $\Upsilon_\caP$ defined in Section~\ref{sec: pruning trees}, where $\caP$ is the set of leaves of $\Upsilon_\caP$.   
The upshot of the procedure is stated in Proposition \ref{pro: main bound RV} in Section \ref{sec: main bound matrix elements}, and it reads 
\begin{equation}\label{eq: prop 9 repeat}
    |\langle V(g)|\sigma\rangle| 1_{\mathbf{NR}(t)} \; \le \;  \frac{1}{g!} B(g) \sum_{\overline{\sigma}\in\mathcal C(g,\sigma)} Y(g,\overline\sigma).
\end{equation}
Here, the quantity \(Y(g, \overline{\sigma})\) is a product of diagonal elements of the denominators \(R(t')\), taken over non-leaf triad vertices \(t'\) in \(\Upsilon_\mathcal{P}\). 
This corresponds to the product of denominators in equation~\eqref{eq: product of denominators} from Section~\ref{sec: probabilistic estimates outline}.  
The factor \(B(g)\), on the other hand, is the product of inductive bounds associated with the leaves of the tree, that is, the elements of \(\mathcal{P}\).

We see that the bound \eqref{eq: prop 9 repeat} also involves a sum over \emph{multiconfigurations} ${\overline{\sigma}}$, that will be properly introduced in Section~\ref{sec: sets of multiconfigs}. 
In short, multiconfigurations are finite sequences of configurations $\overline\sigma(z)$, i.e.\ elements of $\{\pm 1\}^L$,  indexed by the vertices $z$ of yet another subtree $\Upsilon_{\mathrm{rel}}$ of $\Upsilon(g)$, such that $\Upsilon_{\mathrm{rel}} \subset \Upsilon_\caP \setminus \caP$.  
These configurations indicate which diagonal elements of $R(t')$ are to be chosen in the product $Y(\cdot,\cdot)$, namely
\begin{equation}\label{eq: y definition repeat}
          Y(g,\overline\sigma)= \prod_{t': \mathsf c(t') \in \caS_{\mathrm{pro}}(g)} \langle R(t')| \overline{\sigma}(t')\rangle, 
\end{equation}
cf.~\eqref{eq: def of Y RV}.
{Here, \(\mathcal{S}_{\mathrm{pro}}(g)\) refers to the set introduced informally in Section~\ref{sec: probabilistic estimates outline} as the collection of denominators to be estimated probabilistically. A precise definition will be given in Section~\ref{subsec: denominators estimated by probabilistic or inductive bounds}.}

The sum over multiconfigurations essentially originates from the fact that a commutator consists of two products.
Indeed, let us consider an expression of the type
\[
    \langle \sigma' | [O_1,O_2] | \sigma \rangle 
    \; = \;  
    \sum_{\sigma''}    \langle \sigma' | O_1| \sigma''\rangle \langle \sigma''| O_2| \sigma \rangle 
    -   
    \sum_{\sigma''}    \langle \sigma' | O_2| \sigma''\rangle \langle \sigma''| O_1| \sigma \rangle .
\]
In our cases, all operators $O_1,O_2$ that appear, and also their commutator $[O_1,O_2] $, are X-monomials, which in particular implies that in all of the above expressions, the ket $|\cdot\rangle$ determines uniquely the bra $\langle \cdot |$. In particular, for every term on the right-hand side,  there is single value of $\sigma''$ for which the term is non-zero, say $\sigma''=\sigma(1)$ in the first term, and  $\sigma''=\sigma(2)$ in the first term. Therefore, the above equation simplifies to 
\[
    \langle [O_1,O_2] | \sigma \rangle 
    \; = \;    
    \langle  O_1| \sigma(1)\rangle \langle  O_2| \sigma \rangle -     
    \langle O_2| \sigma(2)\rangle \langle O_1| \sigma \rangle.
\]    

In fact, we extend the sum over multiconfigurations beyond what is required by the above considerations.
The rationale for this extension is that we strive for the bound \eqref{eq: prop 9 repeat} to be class-invariant, for the equivalence relation introduced in Section~\ref{sec: equivalence of diagrams}. 
Ideally, we would like the right-hand side to depend only on $[g]$, rather than on the diagram $g$. 
Indeed, so far, we disregarded the union over $t\in [t]$ in the bound~\eqref{eq: main probabilistic bound NRII}. 
Recall, though, that we cannot afford a straightforward union bound, as discussed in Section~\ref{sec: probabilistic estimates outline}. 
Hence, class-invariance would allow us to prove the full bound \eqref{eq: main probabilistic bound NRII} with no need for a union bound. 
This goal will not be fully achieved, but we will come close enough and we do not comment further here on this particular issue (see Section~\ref{sec: class dependence C(g,sigma)} for the needed adjustments). 
The inductive construction and control of $\caC(g,\sigma)$ is carried over in Sections~\ref{sec: expansion of vertices into children}~and~\ref{sec: main bound matrix elements}.


\paragraph{The Choice of the Subtree $\Upsilon_\mathcal{P}$.}
We conclude by commenting on the choice of the subtree \( \Upsilon_\mathcal{P}\). 
As mentioned earlier, inductive bounds should ideally be applied only to triads with crowded central diagrams or to \(V\)-diagrams at scale $0$, so as to avoid the proliferation of factors of $1/\varepsilon$. 
This would suggest that the tree $\Upsilon_\caP$ be defined such that all its leaves are either triads with crowded central diagrams, or diagrams at scale zero. If we were to choose such a definition, the product over denominators that we have to bound, would however include more factors than the product $Y(g,\sigma)$ in \eqref{eq: y definition repeat}, {since the definition of $\mathcal S_{\mathrm{pro}}(g)$ in Section~\ref{subsec: denominators estimated by probabilistic or inductive bounds} excludes more diagrams than crowded $A$-diagrams}. This would  pose a problem if we consider how such a product of denominators  
 will be estimated. Let us therefore describe now how the product in \eqref{eq: y definition repeat} is estimated, which will make it clear that the above definition is untenable.

This estimate will be done via direct integration, similar to the treatment of the triple denominator in the proof of \eqref{eq: main probabilistic bound NRI} in Section~\ref{subsec: bounds on probability of resonances for equivalence classes}.
In the current setting, however, we face a much larger product of denominators, and to carry out the integration efficiently, some degree of independence among factors is needed.  
To make this possible, {it is necessary to omit more denominators than the ones associated with crowded diagrams},  
thereby reducing the number of factors appearing in \(Y(g, \overline{\sigma})\).
A satisfactory compromise is achieved by defining \(\Upsilon_{\mathcal{P}}\) so that all triads \(t'\) with \(c(t') \in \mathcal{S}_{\mathrm{pro}}(g)\) are non-leaf vertices of \(\Upsilon_\mathcal{P}\); see Section~\ref{sec: pruning trees}.

On the one hand, we then need to control the factors of $1/\varepsilon$ that arise due to the leaves of $\Upsilon_\caP$ that are neither triads with crowded central diagrams, nor \(V\)-diagrams at scale $0$. 
These elements are gathered in a subset $\overline{\caP}_{\mathrm{fo}}\subset \caP$  defined in Section~\ref{sec: subset of fully overlapping vertices}, where $\mathrm{fo}$ stands for ``fully overlapping''. The crucial fact about the elements of $\overline{\caP}_{\mathrm{fo}}$ is that they do not contribute a too large power of $1/\varepsilon$, which means that the sum of their orders is not too high. A bound on this sum is achieved in Proposition \ref{prop: norm loss} in Section \ref{sec: reduction number of variables}.

On the other hand, the set \(\mathcal{S}_{\mathrm{pro}}(g)\) is constructed specifically in Section~\ref{subsec: denominators estimated by probabilistic or inductive bounds} to ensure the required form of independence for denominators associated with diagrams in \(\mathcal{S}_{\mathrm{pro}}(g)\).  
This is established in Proposition~\ref{prop: summary of erasure} in Section~\ref{sec: selecting integration variables}.  
More precisely, we assign an integration variable to each denominator in such a way that, after a change of variables, carried out in the proof of Lemma~\ref{lem: prob bound product denominators} in Section~\ref{sec: conclusion proof probability resonance}, these variables become independent.
This is eventually made possible by the fact that diagrams in \( \mathcal{S}_{\mathrm{pro}}(g)\) that are not in hierarchical relation, do not overlap too much.

\section{Determining Integration Variables}\label{sec: selecting integration variables}

Let $g \in \mathcal G^{(k)}$ be a non-crowded diagram at some scale $k\ge 0$. This diagram $g$ is fixed throughout the present section.

\subsection{Probabilistic or Inductive Bounds for Denominators}\label{subsec: denominators estimated by probabilistic or inductive bounds}

We regroup the denominators that $g$ carries into those estimated inductively and those estimated using a probabilistic method.
We define $\mathcal S(g)$ as the set of all diagram descendants of $g$, that are $A$-diagrams, together with $g$ itself.  
We partition the set $\mathcal S(g)$ as $\mathcal S(g) = \mathcal S_{\mathrm{ind}}(g)\cup\mathcal S_{\mathrm{pro}}(g)$, 
where ``$\mathrm{ind}$'' stands for \emph{inductive} and ``$\mathrm{pro}$'' for \emph{probabilistic}.

To define $\mathcal S_{\mathrm{ind}}(g)$, recall the notion of fully overlapping descendant diagrams of $g$ introduced in the beginning of Section~\ref{sec: probabilistic estimates}. 
A diagram $h \in \mathcal S(g)$ belongs to $\mathcal S_{\mathrm{ind}}(g)$ if 
 at least one of the following conditions is satisfied: 
\begin{enumerate}
	\item\label{item: overlapping}
	The diagram $h$ is fully overlapping.  
	
	\item\label{item: close to overlapping}
        The diagram $h$ is adjacent to a fully overlapping diagram $h'$ 
       {(see~\eqref{eq: locality constraint diagrams} and the  lines below~\eqref{eq: locality constraint diagrams} for the definition of adjacency between diagrams)}.
	\item\label{item: crowded}
	The diagram $h$ is crowded. 
	
	\item\label{item: sub constituent}
	The diagram $h$ is a descendant of a diagram in $\mathcal S_{\mathrm{ind}}(g)$. 
	
	\item\label{item: gap}
	There exists a triad $t$ such that $\mathsf c(t)\in \mathcal S_{\mathrm{ind}}(g)$ and $h$ is a descendant of one of the gap-diagrams $\mathsf l(t)$ or $\mathsf r(t)$.
	
\end{enumerate}

We define $\mathcal S_{\mathrm{pro}}(g)$ as 
$\mathcal S_{\mathrm{pro}}(g)=\mathcal S(g)\backslash \mathcal S_{\mathrm{ind}}(g)$.

\begin{remark}
	The above definition applies only to $A$-diagrams, but its formulation involves other diagrams as well. 
	In particular, in case~\ref{item: overlapping}, the diagram $h$ may be overlapping with diagrams that are not $A$-diagrams. 
	Similarly, in case~\ref{item: close to overlapping}, the fully overlapping diagram $h'$ does not need to be itself an $A$-diagram. 
\end{remark}

\begin{remark}
	By cases~\ref{item: crowded} and \ref{item: sub constituent}, every $A$-diagram that is a descendant of a crowded $A$-diagram, is in $\mathcal S_{\mathrm{ind}}(g)$. 
\end{remark}

\subsection{Main Result}
As outlined in Section~\ref{sec: probabilistic estimates outline}, we will now associate to any element in $\mathcal S_{\mathrm{pro}}(g)$ an integration variable $\theta_x$, i.e.\@ a site $x$. The main result of this section is 

\begin{proposition}\label{prop: summary of erasure}
{There exists a partition of $S_{\mathrm{pro}}(g)$ into six sets,  
\[
	\mathcal S_{\mathrm{pro}}(g)
	\; = \;
	\bigcup_{j=1}^6 \mathcal S'_j (g), 
\]
with the following properties.
For any $j=1,\dots,6$, we can order the elements of $\mathcal S_j'(g)$ as $h_1^j,\dots,h_{m_j}^j$ and define points $x_1^j,\dots,x_{m_j}^j \in \Lambda_L$ such that 
\begin{enumerate}
    \item 
    $x_i^j\in\mathcal A (h_i^j)$ for all $1\le i \le m_j$, 
    \item 
    either 
    \[
    x_i^j \; > \; \max I(h_1^j)\cup \dots\cup I(h_{i-1}^j) \qquad \forall\; 2\le i \le m_j,
    \]
    or 
    \[
    x_i^j \; < \; \min I(h_1^j)\cup \dots\cup I(h_{i-1}^j) \qquad \forall\; 2\le i \le m_j.
    \]
\end{enumerate}}
\end{proposition}

{The second item in this proposition states that the elements of each set \(\mathcal{S}_j'\) can be ordered so that each one sticks out with respect to its predecessor, a property that turns out to be sufficient to guarantee the form of independence required for our probabilistic estimates, as explained in Section~\ref{subsec: outline proba}.}


Let us introduce some notation that will be used below. We recall that diagram descendants of $g$ are either $V$-diagrams, $A$-diagrams, or gap-diagrams. We refer to  $A$-and $V$-diagrams as \emph{non-gap diagrams}. For an $A$-diagram $h$, we use the notation $t(h)$ to indicate the triad such that $h=\mathsf{c}(t)$; for a $V$-diagram $h$, we simply set $t(h)=h$.

\subsection{Diagrams that Are Not Fully-Overlapping}\label{sec: more about non-overlapping}

Recall the notion of fully overlapping descendant diagrams of $g$ introduced in the beginning of Section~\ref{sec: probabilistic estimates}.
Let us write $\caV_{\mathrm{nfo}}=\caV_{\mathrm{nfo}}(g)$ for the set of 
{non-gap} diagrams descendants of $g$ that are not fully overlapping. We will derive some useful properties of such diagrams.

\begin{lemma}\label{lem: triadsupport inherits order} Let $h$ and $h'$ be two elements of $\caV_{\mathrm{nfo}}$  that are not in hierarchical relation. Then the following are equivalent
\begin{enumerate}
    \item   $\min I(h) < \min I(h')$,
    \item   $\max I(h) < \max I(h')$,
    \item  $\min I(t(h)) < \min I(t( h'))$,
    \item  $\max I(t(h)) < \max I(t(h'))$.
\end{enumerate}
    Moreover, if any of the above inequalities does not hold, then it holds with the roles of $h,h'$ reversed.
\end{lemma}
\begin{proof}
For a pair $h,h'$ of not fully overlapping diagrams that are not in hierarchical relation, it cannot happen that  $I(h) \subset I(t(h'))$. All claims flow from this observation.  
\end{proof}

This motivates the following definition, applicable for elements $h,h' \in \caV_{\mathrm{nfo}}$ {that are not in hierarchical relation}:
We say that $h$ \emph{is to the left} of $h'$ and  $h'$ \emph{is to the right} of $h$ if either (and hence all) of the above inequalities hold.

{We recall the relation of adjacency defined for colleagues in Section~\ref{subsec: diagrams iteratively}. We need now the weaker notion of \emph{weak adjacency}, defined for non-gap diagrams descendants $h$ of $g$.    
We say that $h, h'$  are weakly adjacent if the union of $I(t(h))$ and $I(t(h'))$ is an interval (rather than two intervals). 
Note that if $h,h'$ are adjacent, then in particular they are weakly adjacent, but the converse need not be true; 
    for example, diagrams that are not colleagues can never be adjacent, but they can be weakly adjacent.  }

\begin{lemma}\label{lem: 3 diagrams general}
    Let $h,h',h''$ be three non-gap diagrams descendants of $g$ that have no hierarchical relations.  If every two of them are weakly adjacent, then at least one of the diagrams is fully overlapping. 
\end{lemma}
\begin{proof} 

If every two of the three diagrams are weakly adjacent, then for at least one of them, say $h''$, it holds that $I(t(h''))\subset  I(t(h))\cup I(t(h')) $, by elementary geometric considerations.  
However, then $h''$ is fully overlapping. 
\end{proof}
\begin{corollary}\label{cor: to lem 3 diagrams}
        If  $h,h',h''$ are three non-gap diagrams descendants of $g$ that have no hierarchical relations, then at  least one of them is fully overlapping if any of the following conditions is satisfied
        \begin{enumerate}
            \item          
            $h,h'$ are adjacent and $I(h'')$  intersects $\mathcal A(t(h),t(h'))$, as defined in \eqref{eq: locality constraint diagrams},
            \item 
            Both $h',h''$ are to the right of $h$ and weakly adjacent to $h$.
        \end{enumerate}
\end{corollary}
\begin{proof}
One checks that any of these conditions leads to the condition of Lemma \ref{lem: 3 diagrams general}.
\end{proof}

The following lemma will be needed only in Section \ref{sec: reduction number of variables}, and its proof is analogous to the reasoning above.

\begin{lemma}\label{lem: 3 diagrams adjacent to something}
    Let $h,h',h''$ be three non-gap diagrams descendants of $g$, that are adjacent to the non-gap diagram $h'''$.  Then at least one of $h,h',h''$ is fully overlapping.
\end{lemma}

\subsection{Decomposing $\mathcal S_{\mathrm{pro}}(g)$}\label{subsec: decomposition of S pro}

We now partition the set $\mathcal S_{\mathrm{pro}}(g)$ itself as $\mathcal S_{\mathrm{pro}}(g) = \mathcal S_1(g) \cup \mathcal S_2(g)$. 
A diagram $h\in \mathcal S_{\mathrm{pro}}(g)$ is in $\mathcal S_1(g)$ if it shares an active spin with a colleague, 
i.e.\@ if there exists a colleague $h'$ and a point $x$ such that {$x \in \caA(h) \cap \caA(h')$}. 
Let us point out that the diagram $h'$ cannot be fully overlapping, otherwise $h$ itself would be an element of $\mathcal S_{\mathrm{ind}}(g)$. 
All elements in $\mathcal S_{\mathrm{pro}}(g)$ that are not in $\mathcal S_1(g)$ are in $\mathcal S_2(g)$.

\subsection{Integration Variables for Diagrams in $\mathcal S_1(g)$.}\label{subsec: integration variables for S1}

We now will partition $\mathcal S_1(g)$ itself as $\mathcal S_1(g) = \mathcal S_{1,\mathrm{l}}(g)\cup\mathcal S_{1,\mathrm{r}}(g)$, 
where ``$\mathrm{l}$'' stands for \emph{left} and ``$\mathrm{r}$'' stands for \emph{right}.
Given $h\in \mathcal S_1(g)$, we consider all its colleagues that share an active spin with $h$ (there is at least one). 
If $h$ is to the left of at least one of these colleagues, we say that $h\in \mathcal S_{1,\mathrm{l}}(g)$. 
Otherwise $h\in \mathcal S_{1,\mathrm{r}}(g)$. 
We notice that there are at most two such colleagues, one on the left and one on the right, but we will not need this fact.

To any diagram $h\in\mathcal S_{1,\mathrm{l}}(g)$, we associate a point $x(h)$, where $x(h)$ is an active spin for $h$ shared with a colleague that is to the right of $h$.
There may be several possibilities for $x(h)$, in which case we simply select one. 
We now order the elements $h \in \mathcal S_{1,\mathrm{l}}(g)$ as $h_1,\ldots, h_{p}$ with $p=|\mathcal S_{1,\mathrm{l}}(g)|$ 
in such a way that
$$
x(h_i)\leq x(h_{i+1}), \qquad  1\le i \le p-1.
$$
We abbreviate $x_i=x(h_i)$.

Similarly, to any diagram $f\in\mathcal S_{1,\mathrm{r}}(g)$, we associate a point $y(f)$, where $y(f)$ is an active spin for $f$ shared with a colleague that is to the {left} of $f$, and we order the diagrams as $(f_1,\ldots, f_q)$ with $q = |\mathcal S_{1,\mathrm{r}}(g)|$ such that $y_i = y(f_i)$ satisfies $y_i \ge y_{i+1}$ for $1 \le i \le q-1$.

\begin{proposition}\label{prop: independent x for s one}
For any $2 \le i \le p$, 
\[
    x_i \; > \;  \max I(h_1) \cup \dots \cup I(h_{i-1}).
\]
For any $2 \le i \le q$, 
\[
    y_i \; < \;  \min I(f_1) \cup \dots \cup I(f_{i-1}).
\]
\end{proposition}

\begin{remark}
Since active spins for a diagram are in the support of that diagram, Proposition \ref{prop: independent x for s one} implies that $x_1 < \dots < x_p$ and  $y_1 > \dots > y_q$. 
\end{remark}

\begin{proof}[Proof of Proposition  \ref{prop: independent x for s one}]
	We only prove the first statement; the proof of the second one is analogous. 
	{To each $h_i$ with $1 \le i \le p$, we associate a colleague $h_i'$ that is to the right of $h_i$ and that shares with it the active spin $x_i$.}
	Since the sequence $(x_i)_{1\le i \le p}$ is non-decreasing, it is enough to show that $x_j \notin I(h_i)$ for any $1 \le i < j \le p$.
	So let us assume that $x_j \in I(h_i)$ for  $1 \le i < j $ and force a contradiction. 
    
Since $x_i \leq x_j$ and $h_{i}'$ is to the right of $h_i$, it follows that 
\[
    x_j \; \in \; I (h_i) \cap I (h_i') \cap I (h_j) \cap I (h_j').
\]
Recall that {none} of these four diagrams is fully overlapping. We contemplate three possibilities. 
\begin{enumerate}
    \item  If $h_i,h_i',h_j,h_j'$ are all distinct, then at least three of them are not in hierarchical relation (as there are two pairs of colleagues). Therefore,  Lemma~\ref{lem: 3 diagrams general} implies that at least one out of the four diagrams is fully overlapping, hence  a contradiction. 
    \item If three of those four diagrams are distinct, then these three diagrams are colleagues, hence they are not in hierarchical relation, so  Lemma  \ref{lem: 3 diagrams general}  again gives contradiction. 
    \item The only other possibility is that $h_i=h_j'$ and $h_j=h_i'$, but this yields the contradiction
    $\min I(h_i) < \min I(h'_{i})= \min I( h_{j}) < \min I( h'_{j}) =  \min I(h_i)$.  
\end{enumerate}
This concludes the proof. 
\end{proof}

\subsection{Integration Variables for Diagrams in $\mathcal S_2(g)$.}\label{subsec: integration variables for S2}

We now partition $\mathcal S_2(g)$ as $\mathcal S_2(g) = \mathcal S_{2,\mathrm l}(g) \cup \mathcal S_{2,\mathrm r}(g)$. 
For this, recall the definition of adjacent diagrams in \eqref{eq: locality constraint diagrams}.
Any diagram $h \in \mathcal S_{\mathrm{pro}}(g)$ with $h\ne g$ is adjacent to at least one other diagram $h'$. 
Since $h \in \mathcal S_{\mathrm{pro}}(g)$, nor $h$ nor any of its adjacent diagram $h'$ is fully-overlapping, and it is hence possible to tell whether $h$ is to the left or to the right of $h'$. 
We say that $h\in \mathcal S_2(g)$ is an element of $\mathcal S_{2,\mathrm l}(g)$ if $h=g$ (by convention) of if it is adjacent to a diagram $h'$ on its right. 
Otherwise $h\in\mathcal S_{2,\mathrm r}(g)$, and in this case we notice that it is adjacent to a diagram on its left.

\begin{lemma}\label{lem: last group}
	Let $h\in\mathcal S_{2,\mathrm l}(g)$ and assume $h\ne g$. 
	\begin{enumerate}
		\item 
		The diagram $h'$ that is adjacent to $h$ and is to the right of $h$, is unique. 
		\item Let $f$ be a diagram descendant of $g$, such that $h$ and $h'$ are descendants of $f$. 
		Then any $x\in \mathcal A(t(h),t(h'))$ is active for $f$. 
	\end{enumerate}
\end{lemma}

\begin{proof}
	Let us start with the first item.  
	If another $h''$ is adjacent to $h$ on its right, then item 2 of Corollary~\ref{cor: to lem 3 diagrams} yields the claim.
 
	Let us come to the second item. 
	Since $h$ and $h'$ are adjacent, any $x \in \mathcal A(t(h),t(h'))$, i.e.\@ $x$ is active for $h$ or $h'$, is not active for both since $h \in \mathcal S_2(g)$. 
	Assume now by contradiction that the property would be wrong. 
	Then there exists some diagram $e'$ such that $x$ is active for $e'$ and such that $e'$ is a colleague of some diagram $e$ of which $h$ and $h'$ are the descendants. Moreover, $e'$ is not fully overlapping, since otherwise $e$ and hence $h,h'$ would be in $h\in \mathcal S_{\mathrm{ind}}$. 
    By item 1 of Corollary~\ref{cor: to lem 3 diagrams} applied with $h$, $h'$ and $e'$, this yields  contradiction.
\end{proof}

\begin{corollary}
	Let $h\in\mathcal S_{\mathrm{pro}}(g)$. If $h$ is a descendant of a gap-diagram, then $h\in\mathcal S_1(g)$. 
\end{corollary}
\begin{proof}
	By the lemma above, if $h\in\mathcal S_2(g)$, the gap diagram of which it is a descendant would have an active spin, which is impossible. 
\end{proof}

To proceed, let us extend the notion of being to the left or right, previously defined for two diagrams in $\caV_{\mathrm{nfo}}$ that are not in hierarchical relation. We now also say that $h $ is to the left of $h'$ if $h$ is a descendant of $h'$. We will also write $h<h'$ as a shorthand for ``$h$ is to the left of $h'$".  The main advantage of this notion is in the next lemma.

\begin{lemma}
	The binary relation $<$ defines a strict total order on $\caV_{\mathrm{nfo}}$
\end{lemma}

\begin{proof}
Once all properties of a strict total order have been checked, except for transitivity, it suffices to argue that three diagrams $h,h',h''$ cannot satisfy the following cyclic relation
$$
h<h', \qquad  h'<h'', \qquad h''<h
$$
So let us assume this cyclic relationship and derive a contradiction. We consider three possible cases:
\begin{enumerate}
    \item  None of these 3 pairs are in hierarchical relation. Then the cyclic relation implies 
    $$
    \min I(h) <\min I(h') <\min I(h'') <\min I(h)
    $$
    which is impossible.
    \item 
    At least two of these three pairs are in hierarchical relation. Without loss of generality, we can assume that $h$ is a descendant of $h'$ and $h'$ is a descendant of $h''$. It then follows that $h$ is a descendant of $h''$ which contradicts the relation $h''<h$.
    \item 
    Exactly one of these three pairs is in hierarchical relation. Without loss of generality, we can assume that $h''$ is a descendant of $h$, which implies $I(h'') \subset I(h)$. This leads to the contradiction 
    $$
   \max I(h)  < \max I(h') < \max I(h'') \leq \max I(h) 
    $$
    where the {first and second inequality follow from Lemma \ref{lem: triadsupport inherits order}}.
\end{enumerate}
This concludes the proof. 
\end{proof}

Let us order all the diagrams in $\mathcal S_{2,\mathrm l}$ as $(h_1,\dots,h_p)$ with $p=|\mathcal S_{2,\mathrm l}|$ in such a way that $h_i<h_{i+1}$ for any $1 \le i \le p-1$.  Similarly, we order the diagrams in $\mathcal S_{2,\mathrm r}$ as $(f_1,\dots , f_q)$ with $q=|\mathcal S_{2,\mathrm r}|$ in such a way that $f_i>f_{i+1}$ for any $1 \le i \le q-1$.

\begin{proposition}\label{prop: integration variables for s2}
    For all $3\le i \le p$, there exists an active spin $x_i$ for $h_i$ such that
    \[
        x_i \; > \; \max I(h_1)\cup \dots \cup I(h_{i-2}).
    \]
    Similarly, for all $3 \le i \le q$, there exists an active spin $y_i$ for $f_i$ such that
    \[
        y_i \; < \; \min I(f_1)\cup \dots \cup I(f_{i-2}).
    \]
\end{proposition}
\begin{proof}
	We only prove the first claim, the proof of the second one being completely analogous. 
	Let us consider three consecutive diagrams $h_i$, $h_{i+1}$ and $h_{i+2}$, and let us show that there is an active spin $x$ for $h_{i+2}$ that is larger than any point in $I(h_i)$.
	Let $h_i',h_{i+1}'$ be the diagrams adjacent to and to the right of $h_i,h_{i+1}$, respectively. 
 	We distinguish four cases stemming from the two ways the relation $<$ (being to the left) can be satisfied, and that covers thus all possibilities.

	\begin{enumerate}
	\item 
	$h_i$ is not a descendant of $h_{i+1}$ and $h_{i+1}$ is not a descendant of $h_{i+2}$. 
	In this case, if an active spin $x$ for $h_{i+2}$ would be in $I(h_i)$, or even smaller than any point in $I(h_i)$, then  
 both $h_{i+1},h_{i+2}$ would be to the right of $h_i$ and weakly adjacent to $h_i$. Thus, item 2 of Corollary \ref{cor: to lem 3 diagrams} yields a contradiction.

	\item 
	$h_i$ is a descendant of $h_{i+1}$ but $h_{i+1}$ is not a descendant of $h_{i+2}$. 
	We note that $h_i$ and $h_i'$ are not descendant of $h_{i+2}$ and 
 that $h_i,h'_i,h_{i+2}$ are not in hierarchical relation.
Assume by contradiction that there exists an active spin $x$ for $h_{i+2}$ that would be in $I(h_i)$ or even smaller than any point in $I(h_i)$. 
Then $h'_i,h_{i+2}$ would both be to the right of $h_i$ and weakly adjacent to $h_i$. This yields again a contradiction by item 2 of Corollary \ref{cor: to lem 3 diagrams}.

	\item 
  	$h_i$ is not a descendant of $h_{i+1}$ but $h_{i+1}$ is  a descendant of $h_{i+2}$.
We take $x_i$ to be an element of $\mathcal A(h_{i+1},h'_{i+1})$, which is possible by Lemma \ref{lem: last group}.  Note that $h_i,h_{i+1},h'_{i+1}$ are not in hierarchical relation. 
 If $x \in I(h_i)$, then item 1 of Corollary  \ref{cor: to lem 3 diagrams} yields a contradiction.

	\item 
 	$h_i$ is a descendant of $h_{i+1}$ and $h_{i+1}$ is not a descendant of $h_{i+2}$.
We take $x_i$ to be an element of $\mathcal A(h_{i+1},h'_{i+1})$, which is possible by Lemma \ref{lem: last group}.  Note that $h_i,h'_{i},h'_{i+1}$ are not in hierarchical relation. 
 If $x_{i+2} \in I(h_i)$, then in particular $h'_{i+1}$ is weakly adjacent to $h_i$, as is of course $h'_{i}$. Since both
 are also 
 to the right of $h_i$, item 2 of Corollary  \ref{cor: to lem 3 diagrams} yields contradiction.
\end{enumerate} 
This concludes the proof. 
\end{proof}
Finally, we define a partition $S_{2,\mathrm{l}}=S_{2,\mathrm{l},\mathrm{e}}\cup S_{2,\mathrm{l},\mathrm{o}}$ by 
\[
S_{2,\mathrm{l},\mathrm{e}}=\{ h_i \in  S_{2,\mathrm{l}} \, |\,   i \text{ is even} \}, \qquad  
S_{2,\mathrm{l},\mathrm{o}}=\{ h_i \in  S_{2,\mathrm{l}} \, |\,   i \text{ is odd} \}
\]
with $h_i$ referring to ordering used above. Analogously, we define the partition 
$S_{2,\mathrm{r}}=S_{2,\mathrm{r},\mathrm{e}}\cup S_{2,\mathrm{r},\mathrm{o}}$. 
We now have all tools in hand to finish the 
\begin{proof}[Proof of Proposition \ref{prop: summary of erasure}]
We set 
\[
	\mathcal S_{\mathrm{pro}}(g)
	\; = \;
	\mathcal S_{1,\mathrm{l}} \cup
	\mathcal S_{1,\mathrm{r}} \cup
	\mathcal S_{2,\mathrm{l},\mathrm{o}} \cup
	\mathcal S_{2,\mathrm{l},\mathrm{e}} \cup
	\mathcal S_{2,\mathrm{r},\mathrm{o}} \cup
	\mathcal S_{2,\mathrm{r},\mathrm{e}} 
	\; =: \; 
	\bigcup_{j=1}^6 \mathcal S'_j (g). 
\] 
The clauses of Proposition \ref{prop: summary of erasure} are fulfilled by Proposition  \ref{prop: independent x for s one} for $\mathcal S_{1,\mathrm{l}/\mathrm{r}}$ and by Proposition \ref{prop: integration variables for s2} for $\mathcal S_{2,\mathrm{l}/\mathrm{r},\mathrm{e}/\mathrm{o}}$.
\end{proof}

\section{Contribution of Fully Overlapping Descendants}\label{sec: reduction number of variables}

Let $k\ge 0$ and let $g\in\mathcal G^{(k)}$.
We consider the diagram $g$ to be fixed throughout this section.
Our aim is to bound the sum of the orders of the diagram descendants of $g$ that are either fully overlapping or adjacent to a fully overlapping descendant, and that contribute a denominator requiring an inductive estimate. 
This is the content of Proposition~\ref{prop: norm loss} below.
As explained in Section \ref{subsec: outline proba}, the usefulness of such an estimate lies in the fact that these diagrams lead to elements of $\caP$ that are not necessarily crowded, nor of scale zero. 

{We recall the tree $\tree(g)$ introduced in Section~\ref{sec: tree structure}, whose vertices represent the descendants of $g$. 
We denote the set of vertices of a tree $\tree$ by $\mathcal{V}(\tree)$, and typically use the letters $z$ and $z'$ to refer to vertices. 
In particular, each vertex $z \in \mathcal{V}(\tree(g))$ corresponds to a descendant diagram or triad of $g$.
We say that a diagram or triad $z'$ is a \emph{proper} descendant of $z$ if $z'$ is a descendant of $z$ and $z' \ne z$.
Finally, since $g$ is fixed throughout, we will mostly omit it from the notation for simplicity, and write $\tree$ for $\tree(g)$.}

\subsection{Pruning the Tree $\tree$}\label{sec: pruning trees}

Let us prune the tree $\tree$ in several ways, obtaining tree different subtrees, 
all rooted in $g$: $\tree_{\caP},\tree_{\caP^*},\tree_{\caP^{**}}$, see Figures~\ref{fig: various prunings}~and~\ref{fig: intervals for pruning} for an illustration.
Note that such a subtree is completely determined by its set of leaves.  The respective sets of leaves are called $\caP,\caP^*,\caP^{**}$. 

\begin{enumerate}
    \item  $\tree_\caP$ is obtained by erasing from $\tree$  all proper descendants of triads $t$ such that $\mathsf c(t)\in\mathcal S_{\mathrm{ind}}(g)$.
    It follows that elements of $\caP$ are either triads or diagrams at scale $0$.
\item  $\tree_{\mathcal P^*}$ is a subtree of $\tree_\caP$ obtained by erasing from $\tree_\caP$ all proper  descendants of fully-overlapping $V$-diagrams. 

\item    $\tree_{\caP^{**}}$ is defined such that $\tree_{\caP^*}$ is a subtree of  $\tree_{\caP^{**}}$ : all diagram leaves of $\tree_{\caP^*}$ are also leaves of $\tree_{\caP^{**}}$ and all children of triad leaves of $\tree_{\caP^*}$ are leaves of $\tree_{\caP^{**}}$.  
{In loose words, $\mathcal P^{**}$ is obtained from $\mathcal P^*$ be replacing any triad $t\in\mathcal P^*$ by its children $\mathsf c(t)$, $\mathsf l(t)$ and $\mathsf r(t)$.}
Hence, $\mathcal P^{**}$ consists entirely of diagrams.
%
\end{enumerate}


\begin{figure}[p]
    \centering

    \begin{subfigure}[b]{\textwidth}
        \centering


\scalebox{0.7}{
\begin{tikzpicture}[scale=1,
    every node/.style={circle, draw, minimum size=1cm},
    level distance=2cm
]

\node[ultra thick, font=\bfseries] (g) at (0,2.5) {$g$}
  child[xshift=-3cm] {node[ultra thick, font=\bfseries] (g0) {$g_0$}
    child {node[ultra thick, font=\bfseries] (g00) {$g_{0,0}$}}
    child {node[fill=gray!40,ultra thick, font=\bfseries] (t01) {$t_{0,1}$}}
  }
  child[xshift=0cm] {node[fill=gray!40, ultra thick, font=\bfseries] (t1) {$t_1$}
    child {node[ultra thick, font=\bfseries] (g1p) {$g_1'$}}
    child {node[ultra thick, font=\bfseries] (g1) {$g_1$}}
    child {node[ultra thick, font=\bfseries] (g1pp) {$g_1''$}}
  }
  child[xshift=4cm] {node[fill=gray!40, ultra thick, font=\bfseries] (t2) {$t_2$}
    child {node[ultra thick, font=\bfseries] (g2p) {$g_2'$}}
    child {node[ultra thick, font=\bfseries] (g2) {$g_2$}}
    child {node[ultra thick, font=\bfseries] (g2pp) {$g_2''$}}
  };

\foreach \leaf in {g00, t01, g1p, g1pp, g2p, g2pp}
{
    \draw[dashed,thick] (\leaf.south west) ++(0.2,-0.1) -- ++(-0.3,-1);
    \draw[dashed,thick] (\leaf.south) -- ++(0,-1);
    \draw[dashed,thick] (\leaf.south east) ++(-0.2,-0.1) -- ++(0.3,-1);
}

\foreach \leaf in {g1}
{
    \draw[dashed,thick] (\leaf.south west) ++(0.2,-0.1) -- ++(-0.3,-1);
    \draw[dashed,thick] (\leaf.south east) ++(-0.2,-0.1) -- ++(0.3,-1);
}

\foreach \leaf in {g2}
{
    \draw[dashed,thick] (\leaf.south west) ++(0.1,-0.1) -- ++(-0.4,-1);
    \draw[dashed,thick] (\leaf.south west) ++(0.3,-0.1) -- ++(-0.15,-1);
    \draw[dashed,thick] (\leaf.south east) ++(-0.3,-0.1) -- ++(0.15,-1);
    \draw[dashed,thick] (\leaf.south east) ++(-0.1,-0.1) -- ++(0.4,-1);
}

\end{tikzpicture}
}

\caption{The tree $\tree (g)$}
\label{fig: original tree}
\end{subfigure}

\bigskip

\begin{subfigure}[b]{\textwidth}
\centering

\scalebox{0.7}{
\begin{tikzpicture}[scale=1,
    every node/.style={circle, draw, minimum size=1cm},
    level distance=2cm
]

\node[ultra thick, font=\bfseries] (g) at (0,2.5) {$g$}
  child[xshift=-3cm] {node[ultra thick, font=\bfseries] (g0) {$g_0$}
    child {node[ultra thick, font=\bfseries] (g00) {$g_{0,0}$}}
    child {node[fill=gray!40,ultra thick, font=\bfseries] (t01) {$t_{0,1}$}}
  }
  child[xshift=0cm] {node[fill=gray!40, ultra thick, font=\bfseries] (t1) {$t_1$}
  }
  child[xshift=4cm] {node[fill=gray!40, ultra thick, font=\bfseries] (t2) {$t_2$}
    child {node[ultra thick, font=\bfseries] (g2p) {$g_2'$}}
    child {node[ultra thick, font=\bfseries] (g2) {$g_2$}}
    child {node[ultra thick, font=\bfseries] (g2pp) {$g_2''$}}
  };

\foreach \leaf in {g00, t01, g2p, g2pp}
{
    \draw[dashed,thick] (\leaf.south west) ++(0.2,-0.1) -- ++(-0.3,-1);
    \draw[dashed,thick] (\leaf.south) -- ++(0,-1);
    \draw[dashed,thick] (\leaf.south east) ++(-0.2,-0.1) -- ++(0.3,-1);
}

\foreach \leaf in {g2}
{
    \draw[dashed,thick] (\leaf.south west) ++(0.1,-0.1) -- ++(-0.4,-1);
    \draw[dashed,thick] (\leaf.south west) ++(0.3,-0.1) -- ++(-0.15,-1);
    \draw[dashed,thick] (\leaf.south east) ++(-0.3,-0.1) -- ++(0.15,-1);
    \draw[dashed,thick] (\leaf.south east) ++(-0.1,-0.1) -- ++(0.4,-1);
}

\end{tikzpicture}
}

\caption{The tree $\tree_{\mathcal P}(g)$}
\label{fig: 1st pruning}
\end{subfigure}

\bigskip

\begin{subfigure}[b]{\textwidth}
\centering

\scalebox{0.7}{
\begin{tikzpicture}[scale=1,
    every node/.style={circle, draw, minimum size=1cm},
    level distance=2cm
]

\node[ultra thick, font=\bfseries] (g) at (0,2.5) {$g$}
  child[xshift=-3cm] {node[ultra thick, font=\bfseries] (g0) {$g_0$}
  }
  child[xshift=0cm] {node[fill=gray!40, ultra thick, font=\bfseries] (t1) {$t_1$}
  }
  child[xshift=4cm] {node[fill=gray!40, ultra thick, font=\bfseries] (t2) {$t_2$}
    child {node[ultra thick, font=\bfseries] (g2p) {$g_2'$}}
    child {node[ultra thick, font=\bfseries] (g2) {$g_2$}}
    child {node[ultra thick, font=\bfseries] (g2pp) {$g_2''$}}
  };

\foreach \leaf in {g2p, g2pp}
{
    \draw[dashed,thick] (\leaf.south west) ++(0.2,-0.1) -- ++(-0.3,-1);
    \draw[dashed,thick] (\leaf.south) -- ++(0,-1);
    \draw[dashed,thick] (\leaf.south east) ++(-0.2,-0.1) -- ++(0.3,-1);
}

\foreach \leaf in {g2}
{
    \draw[dashed,thick] (\leaf.south west) ++(0.1,-0.1) -- ++(-0.4,-1);
    \draw[dashed,thick] (\leaf.south west) ++(0.3,-0.1) -- ++(-0.15,-1);
    \draw[dashed,thick] (\leaf.south east) ++(-0.3,-0.1) -- ++(0.15,-1);
    \draw[dashed,thick] (\leaf.south east) ++(-0.1,-0.1) -- ++(0.4,-1);
}

\end{tikzpicture}
}

\caption{The tree $\tree_{\mathcal P^*}(g)$}
\label{fig: 2nd pruning}
\end{subfigure}

\bigskip

\begin{subfigure}[b]{\textwidth}
\centering

\scalebox{0.7}{        
\begin{tikzpicture}[scale=1,
    every node/.style={circle, draw, minimum size=1cm},
    level distance=2cm
]

\node[ultra thick, font=\bfseries] (g) at (0,2.5) {$g$}
  child[xshift=-3cm] {node[ultra thick, font=\bfseries] (g0) {$g_0$}
  }
  child[xshift=0cm] {node[fill=gray!40, ultra thick, font=\bfseries] (t1) {$t_1$}
    child {node[ultra thick, font=\bfseries] (g1p) {$g_1'$}}
    child {node[ultra thick, font=\bfseries] (g1) {$g_1$}}
    child {node[ultra thick, font=\bfseries] (g1pp) {$g_1''$}}
  }
  child[xshift=4cm] {node[fill=gray!40, ultra thick, font=\bfseries] (t2) {$t_2$}
    child {node[ultra thick, font=\bfseries] (g2p) {$g_2'$}}
    child {node[ultra thick, font=\bfseries] (g2) {$g_2$}}
    child {node[ultra thick, font=\bfseries] (g2pp) {$g_2''$}}
  };

\foreach \leaf in {g2p, g2pp}
{
    \draw[dashed,thick] (\leaf.south west) ++(0.2,-0.1) -- ++(-0.3,-1);
    \draw[dashed,thick] (\leaf.south) -- ++(0,-1);
    \draw[dashed,thick] (\leaf.south east) ++(-0.2,-0.1) -- ++(0.3,-1);
}

\foreach \leaf in {g2}
{
    \draw[dashed,thick] (\leaf.south west) ++(0.1,-0.1) -- ++(-0.4,-1);
    \draw[dashed,thick] (\leaf.south west) ++(0.3,-0.1) -- ++(-0.15,-1);
    \draw[dashed,thick] (\leaf.south east) ++(-0.3,-0.1) -- ++(0.15,-1);
    \draw[dashed,thick] (\leaf.south east) ++(-0.1,-0.1) -- ++(0.4,-1);
}

\end{tikzpicture}
}

\caption{The tree $\tree_{\mathcal P^{**}}(g)$}
\label{fig: 3rd pruning}
\end{subfigure}
    
\caption{The tree $\tree(g)$ and its prunings. The rationale behind these prunings can be understood from the positions of the intervals $I(g)$, $I(g_0)$, $I(t_1)$, and $I(t_2)$ shown in Figure~\ref{fig: intervals for pruning}.}
\label{fig: various prunings}
\end{figure}

\begin{figure}[h]
\centering

\begin{tikzpicture}

\pgfmathsetmacro{\a}{1.5}    
\pgfmathsetmacro{\b}{7.5}    

\pgfmathsetmacro{\atwo}{1.5} 
\pgfmathsetmacro{\btwo}{5.0} 

\pgfmathsetmacro{\aone}{4.5} 
\pgfmathsetmacro{\bone}{7.5} 

\pgfmathsetmacro{\aZero}{5.3}   
\pgfmathsetmacro{\bZero}{6.8}   

\pgfmathsetmacro{\topY}{1.5}
\pgfmathsetmacro{\lowY}{0}

\pgfmathsetmacro{\midg}{(\a+\b)/2}
\pgfmathsetmacro{\midtTwo}{(\atwo+\btwo)/2}
\pgfmathsetmacro{\midtOne}{(\aone+\bone)/2}
\pgfmathsetmacro{\midgZero}{(\aZero+\bZero)/2}

\draw[ultra thick] (\a,\topY) -- (\b,\topY);
\draw[thick] (\a,\topY+0.1) -- (\a,\topY-0.1);
\draw[thick] (\b,\topY+0.1) -- (\b,\topY-0.1);
\node at (\midg,\topY+0.4) {$I(g)$};

\draw[thick] (\atwo,\lowY) -- (\btwo,\lowY);
\draw[thick] (\atwo,\lowY+0.1) -- (\atwo,\lowY-0.1);
\draw[thick] (\btwo,\lowY+0.1) -- (\btwo,\lowY-0.1);
\node at (\midtTwo,\lowY+0.4) {$I(t_2)$};

\draw[thick] (\aone,\lowY-0.8) -- (\bone,\lowY-0.8);
\draw[thick] (\aone,\lowY-0.7) -- (\aone,\lowY-0.9);
\draw[thick] (\bone,\lowY-0.7) -- (\bone,\lowY-0.9);
\node at (\midtOne,\lowY-0.4) {$I(t_1)$};

\draw[thick] (\aZero,\lowY+0.6) -- (\bZero,\lowY+0.6);
\draw[thick] (\aZero,\lowY+0.5) -- (\aZero,\lowY+0.7);
\draw[thick] (\bZero,\lowY+0.5) -- (\bZero,\lowY+0.7);
\node at (\midgZero,\lowY+1) {$I(g_0)$};

\end{tikzpicture}

\caption{The reasons for pruning the tree $\tree(g)$ in Figure~\ref{fig: various prunings}: $g_0$ is a fully overlapping $V$-diagram, $t_1$ is adjacent to the fully overlapping diagram $g_0$ and hence $\mathsf c(t_1)\in\mathcal S_{\mathrm{ind}}(g)$, and we assume that $t_2$ is such that $\mathsf c(t_2)\in \mathcal S_{\mathrm{pro}}(g)$.}
\label{fig: intervals for pruning}
\end{figure}


We now introduce the following shorthand:
\begin{equation}
    |z|_{(\mathrm r)} \;=\; \begin{cases}
        |z|_{\mathrm r}  & \text{$h$ is a crowded $A$-diagram} \\
        |z| & \text{otherwise}
    \end{cases}
\end{equation}
Writing $\caR$ to mean either $\caP,\caP^*$ or $\caP^{**}$, we have the following simple properties:

\begin{lemma}\label{lem: simple about primaries}
For any $z \in \mathcal{V}(\tree_{\mathcal R})$,
    $$
    I(z)\;=\;\bigcup_{\substack{z' \in \mathcal R \\  
    z' \mathrm{\; descendant \; of \;} z}}  I(z')
    $$
and  
$$
    |z|\;=\;\sum_{\substack{z' \in \mathcal R \\  
    z' \mathrm{\; descendant \; of \;} z}}  |z'|_{(\mathrm r)}.
    $$
 Moreover, if $\mathcal R=\mathcal P$ or $\mathcal R=\mathcal P^*$, then $|z'|_{(\mathrm r)}=|z'|$ in the above equality. 
\end{lemma}

\begin{proof}
The claim about $I(\cdot)$ is obvious, and we focus on the second one.   
We consider how the order $|z|$ of a vertex $z$ is determined by the children of $z$.  {If $z$ is a diagram $h$ or a triad $t$ with $\mathsf c(t)$ non-crowded, then 
\begin{equation}\label{eq: additivity of order for children}
    |z|\; = \;\sum_{z'\, \text{children of } z} |z'|.
\end{equation} }
If  $z$ is a triad $t$ and $\mathsf c(t)$ is crowded, then this relation is modified into (in that case the children of $t$ are $\mathsf c(t),\mathsf l(t),\mathsf r(t)$):
\begin{equation}\label{eq: broken additivity of order for children}
    |t|\; = \; |\mathsf l(t)|+|\mathsf r(t)|+|\mathsf c(t)|_{\mathrm r}.
\end{equation}
Now our claim is established by going from the root $g$ of the tree $\tree_{\mathcal R}$ down to the leaves and using successively the above relations \eqref{eq: additivity of order for children} and \eqref{eq: broken additivity of order for children}. The only instance in which \eqref{eq: broken additivity of order for children}  is used instead of \eqref{eq: additivity of order for children}, is in the case of $\mathcal R=\mathcal P^{**}$.
Indeed, triads $t$ such that $\mathsf c(t)$ is crowded, can only {be vertices of $\tree_\caP$ if they are leaves of $\tree_\caP$}. Since $\tree_{\caP^*}$ is a subtree of $\tree_\caP$, the same holds for $\tree_{\caP^*}$. 
\end{proof}

\subsection{Subsets Related to Fully Overlapping Diagrams}\label{sec: subset of fully overlapping vertices}

We define two subsets of $\caP$, called $\mathcal P_{\mathrm{fo}}$ and $\overline{\mathcal P}_{\mathrm{fo}}$ respectively: 
\begin{enumerate}
    \item    
    $\mathcal P_{\mathrm{fo}}$ contains all triads $t$ in $\mathcal P$  such that $\mathsf c(t)$ is fully overlapping.  
    \item    
    $\overline{\mathcal P}_{\mathrm{fo}}$ contains the triads in $\mathcal P_{{\mathrm{fo}}}$, and the triads in $\caP$ that are adjacent to a fully overlapping diagram.
\end{enumerate}
Note that these sets consist entirely of triads. 
Analogously, we define two subsets of $\caP^*$, called $\mathcal P_{\mathrm{fo}}^*$ and $\overline{\mathcal P}_{\mathrm{fo}}^*$ respectively: 
\begin{enumerate}
    \item    $\mathcal P^*_{\mathrm{fo}}$ contains all  triads $t$ in $\mathcal P^*$  such that $\mathsf c(t)$ is fully overlapping, and $V$-diagrams in $\caP^*$ that are fully overlapping.
    \item    $\overline{\mathcal P}^*_{\mathrm{fo}}$ contains all elements in  $\mathcal P^*_{\mathrm{fo}} $ and all triads in $\caP^*$ that are adjacent to a fully overlapping diagram. 
\end{enumerate}
Finally, we define the subset $\mathcal P^{**}_{\mathrm{fo}}$ of $\caP^{**}$:
$\mathcal P^{**}_{\mathrm{fo}}$ contains all diagrams in $\mathcal P^{**}$  that are fully overlapping and that are not a gap diagram (i.e.\ they are $V$ or $A$-diagrams).

The relevant connection between $\overline{\mathcal P}_{\mathrm{fo}}$ and $\overline{\mathcal P}_{\mathrm{fo}}^*$ is established in the following lemma:
\begin{lemma}\label{lem: prim and mod prim}
    \[
    \sum_{z \in \overline{\mathcal P}_{\mathrm{fo}}} |z| 
    \; \leq \;
    \sum_{z \in \overline{\mathcal P}^*_{\mathrm{fo}}} |z|.
\]
\end{lemma}
\begin{proof}

Next, let $f:\caP\to\caP^*$ be defined as follows: If $z\in \caP$ belongs to $\caP^*$, then $f(z)=z$. If not, then there is an unique fully overlapping diagram $v$ in $\caP^*$ such that
$z$ is a descendant of $v$. We then set $f(z)=v$ and we abbreviate 
 ``fully overlapping $V$-diagram $v$'' by ``$\mathrm{fo\;}v$''.
Then we compute
\begin{align*}
    \sum_{z \in \overline{\mathcal P}_{\mathrm{fo}}} |z| 
    &\; = \;
  \sum_{z \in \overline{\mathcal P}_{\mathrm{fo}}: f(z)=z} |z| 
    + 
   \sum_{{{\mathrm{fo}}\;} v \in \caP^*} 
    \sum_{z \in \overline{\mathcal P}_{\mathrm{fo}}: f(z)=v} |z|  \\
    & 
    \; \leq \;    \sum_{z \in \overline{\mathcal P}_{\mathrm{fo}}: f(z)=z} |z| +
 \sum_{ {{\mathrm{fo}}\;} v \in \caP^*}  |v|  \\
     & 
    \; \leq \; 
    \sum_{z \in \overline{\mathcal P}^*_{\mathrm{fo}}} |z|.
\end{align*}
{The first equality uses that the clauses $f(z)=z$ and $f(z)=v$ are exclusive, because $\overline{\mathcal P}_{\mathrm{fo}}$ consists only of triads. }
The first inequality follows from Lemma~\ref{lem: simple about primaries} with $\mathcal R=\mathcal P$. 
{The last inequality follows because the two sums run over disjoint sets of vertices (as noted above, $\overline{\mathcal P}_{\mathrm{fo}}$ consists of triads) }
\end{proof}

Finally, we make two remarks that will be used in the next subsection: 
\begin{remark}\label{rem: the first remark about primary descendants}
Every element in $\overline\caP^*_{\mathrm{fo}}\setminus \caP^*_{\mathrm{fo}}$ is adjacent to an element of $\caP^*_{\mathrm{fo}}$.  In contrast, an element of $\overline\caP_{\mathrm{fo}}\setminus \caP_{\mathrm{fo}}$ can fail to be adjacent to an element of  $\caP_{\mathrm{fo}}$. This distinction is the main reason why we introduced $\caP^*$. 
\end{remark}

\begin{remark}\label{rem: relation between star star and star}
$\caP^{**}_{\mathrm{fo}}$ and $\caP^{*}_{\mathrm{fo}}$ are in one-to-one correspondence. More concretely, any $A$-diagram $h$ in $\caP^{**}_{\mathrm{fo}}$ is the central diagram $\mathsf c(t)$ of a triad $t \in \caP^{*}_{\mathrm{fo}}$ and any $V$-diagram  $h$ in $\caP^{**}_{\mathrm{fo}}$ is also an element of $\caP^{*}_{\mathrm{fo}}$. Moreover, any element of $ \caP^{*}_{\mathrm{fo}}$ is related to a diagram in $\caP^{**}_{\mathrm{fo}}$ in this way. 
\end{remark}

\subsection{Bounding the Contribution from $\overline{\mathcal P}_{\mathrm{fo}}$ }

We now state the main result of this section: 
\begin{proposition}\label{prop: norm loss}
        For any non-crowded diagram $g$, 
	\[
		\sum_{z \in \overline{\mathcal P}_{\mathrm{fo}}(g)} |z| \; \le \; 156 (1-\beta)|g|. 
	\]
\end{proposition}

{If $\beta \ge 1 - \frac1{312}$, the above proposition implies that 
\begin{equation}\label{eq: corolary norm loss lemma}
    \sum_{z \in \overline{\mathcal P}_{\mathrm{fo}}} |z| \; \le \; \frac12 |g|.
\end{equation}
This is what we will need later on, and we thus set $\beta =1 - \frac{1}{312}$ starting now.}

We now state a sequence of lemmas that will lead to the proof of Proposition \ref{prop: norm loss}. 
Below, $g$ is always assumed to be non-crowded, as in the statement of Proposition \ref{prop: norm loss}.

\begin{lemma} \label{lem: magic half fully overlapping observation}
Let $h \in \mathcal P^{**}_{\mathrm{fo}}$ and $x\in I(h)$. Then there exists a diagram $h'\in \mathcal P^{**}$ such that $x\in I(h')$ and $h'\neq h$ . 
\end{lemma}
\begin{proof}
Since $h$ is fully overlapping, we can find a diagram  $\hat h$ that is a descendant of $g$, such that $x\in I(\hat h)$, {and such that $h$ and $\hat h$ are not in hierarchical relation}.
Because $\caP^{**}$ are the leaves of a pruning of $\tree$, either of the following alternatives holds true:
\begin{enumerate}
    \item $\hat{h}$ is a descendant of a diagram $h' \in \mathcal P^{**}$, possibly $h'=
    \hat{h}$.
    \item $\hat{h}$ has proper descendants $h'$ that are in $\mathcal P^{**}$.
\end{enumerate}
If the first possibility is realized, then $ x \in I(\hat{h}) \subset  I(h')$, and the claim holds. Indeed, $h'\neq h$ since no pair of elements of $\caP^{**}$ can be in hierarchical relation (because $\caP^{**}$ are the leaves of a tree).

If the second possibility is realized, then we deduce from the first claim of Lemma \ref{lem: simple about primaries} that there is a proper descendant $h'$ {of $\hat h$} such that $x\in I(h')$. By the same argument as above, the possibility $h'=h$ is excluded and we are done as well.
\end{proof}

The following statement is then an immediate consequence of Lemma \ref{lem: magic half fully overlapping observation}: 
\begin{lemma} \label{lem: magic half fully overlapping}
\begin{equation*} 
    |I(g)| 
     \; \le \;
    \sum_{h\in \mathcal P^{**}\setminus  \mathcal P^{**}_{\mathrm{fo}} }  |I(h)| +\frac12 \sum_{h \in  \mathcal P^{**}_{\mathrm{fo}} }|I(h)|.
\end{equation*}
\end{lemma}

Our next result is already similar to Proposition \ref{prop: norm loss}, but with $\caP_{\mathrm{fo}}^*$ instead of $\overline{\caP}_{\mathrm{fo}}$:


\begin{lemma}\label{lem: eq: first real bound}
    \begin{equation*} 
  \frac{1}{12} \sum_{z \in  \mathcal P^*_{\mathrm{fo}} }|z|   
  \;\leq\; (1-\beta) |g|  .
\end{equation*}
\end{lemma}

\begin{proof}
We compute 
\begin{align}
  |I(g)| 
    &\; \le \;   \sum_{h\in \mathcal P^{**}\setminus  \mathcal P^{**}_{\mathrm{fo}}  }  |I(h)| +\frac12 \sum_{h \in  \mathcal P^{**}_{\mathrm{fo}} }|I(h)|     \nonumber\\
& \;\leq\;  \sum_{h\in \mathcal P^{**}\setminus  \mathcal P^{**}_{\mathrm{fo}}}  |h|_{(r)}
+
\frac12 \sum_{h \in  \mathcal P^{**}_{\mathrm{fo}} }|h|_{(r)}     \nonumber\\
& \; =\; \sum_{h\in \mathcal P^{**} }  |h|_{(r)}  -
\frac12 \sum_{h \in  \mathcal P^{**}_{\mathrm{fo}} }|h|_{(r)}     \nonumber\\
& \; = \;  \sum_{z\in \mathcal P^*  }  |z|- 
\frac12 \sum_{z \in  \mathcal P^*_{\mathrm{fo}}}|\mathsf c(z)|_{(r)} \nonumber\\
 &\;\leq\;  \sum_{z\in \mathcal P^*  }  |z|- 
\frac{1}{12} \sum_{z \in  \mathcal P^*_{\mathrm{fo}}}|z| .
\label{eq: long inequalities}
\end{align}
The first inequality is Lemma \ref{lem: magic half fully overlapping}. 
The second inequality uses $|I(h)| \leq |h|_{(r)}$ for diagrams $h$.  The first equality is obvious. 
The second equality follows by an equality for both sums. 
For the first sum, we use Lemma \ref{lem: simple about primaries} with $\mathcal R =  \caP^{**}$ and $z \in \Upsilon_{\caP^*}$. 
For the second sum, we use the one-to-one correspondence of Remark \ref{rem: relation between star star and star}. 
Finally, the last inequality follows from 
 $|\mathsf c(z)|_{(\mathrm r)} \ge \frac16 |z|$, since $\beta \geq 1/2$.

Next, we use that $g$ is non-crowded, hence $\beta |g| \leq |I(g)|$, and 
the equality $|g| = \sum_{z\in \mathcal P^*} |z|$ from Lemma \ref{lem: simple about primaries}. Combining these with \eqref{eq: long inequalities}, 
 we get the statement of the lemma.  
\end{proof}

Now we relate the above bound for $\caP^{*}_{\mathrm{fo}}$ to $\overline{\caP}^{*}_{\mathrm{fo}}$:
\begin{lemma} \label{lem: bound on cal p bar}
    \begin{equation*} 
        \sum_{z\in\overline{\mathcal P}^*_{\mathrm{fo}}}|z| 
        \;\le\;  
        13\sum_{z\in\mathcal P^*_{\mathrm{fo}}}|z|.
  \end{equation*}
\end{lemma}
\begin{proof}
The set $\overline{\mathcal P}^*_{\mathrm{fo}}\backslash {\mathcal P}^*_{\mathrm{fo}}$ consists entirely of triads.
For any $t\in \overline{\mathcal P}^*_{\mathrm{fo}}\backslash {\mathcal P}^*_{\mathrm{fo}}$, we denote by $f(t)$ an arbitrary element of ${\mathcal P}^*_{\mathrm{fo}}$ that is adjacent to $t$. 
There is always at least one such element by Remark~\ref{rem: the first remark about primary descendants}.
We note that $|t| \leq 6 |f(t)|$. Indeed, $t$ and $f(t)$ are at the same scale, which we call $k$. Then  $|t| \leq 3 L_{k+1}$ (since $t$ is a triad) and $|f(t)| \geq \beta L_k$. 
This provides the first inequality in the following bound: 
\begin{equation}\label{eq: bound with twelve}
\sum_{t \in  \overline{\mathcal P}^*_{\mathrm{fo}}\backslash {\mathcal P}^*_{\mathrm{fo}}}  |t| 
\;\leq\;   
6 \sum_{t \in  \overline{\mathcal P}^*_{\mathrm{fo}}\backslash {\mathcal P}^*_{\mathrm{fo}}} |f(t)| 
\;\leq\; 12 \sum_{z \in {\mathcal P}^*_{\mathrm{fo}}} |z|.
\end{equation}
To get the second inequality, we first note that any element $t\in  \overline{\mathcal P}^*_{\mathrm{fo}}\backslash {\mathcal P}^*_{\mathrm{fo}}$ is such that $\mathsf c(t)$ is not fully overlapping. Therefore, for a given  $z$, there can be at most two triads $t \in  \overline{\mathcal P}^*_{\mathrm{fo}}\backslash {\mathcal P}^*_{\mathrm{fo}}$ such that  $z=f(t)$. This follows from Lemma \ref{lem: 3 diagrams adjacent to something} in Section \ref{sec: selecting integration variables}. 
The bound \eqref{eq: bound with twelve} implies the claim of the lemma. 
%
\end{proof}

Finally, we are ready to give the

\begin{proof}[Proof of Proposition~\ref{prop: norm loss}]
    By Lemma \ref{lem: prim and mod prim}, it suffices to prove the proposition with $\overline{\mathcal P}_{\mathrm{fo}}$ replaced by $\overline{\mathcal P}^*_{\mathrm{fo}}$.
This follows then by combining Lemma~\ref{lem: eq: first real bound} and Lemma \ref{lem: bound on cal p bar}.
\end{proof}

\section{Expansion of Vertices into Their Children}\label{sec: expansion of vertices into children}

This section mainly contains preparatory work for the next. 
We show how the absolute value of a matrix element of an operator $V(g)$ corresponding to some diagram $g$ is bounded by a sum of products of the absolute values of matrix elements of its children.  
Such estimates are the ingredients of the unpacking procedure described in Section \ref{subsec: outline proba}.

The important point is that we have a good control on the number of elements in this sum, {labelled by multiconfigurations $\overline{\sigma}$}. 
This is the content of Proposition~\ref{lem: main lemma on the reduction of RV} below. 
We generalize this to triads in Section~\ref{sec: expansion of triads into children}.

\subsection{Expanding a Diagram into Its Children} \label{sec: expansion of diagram into children}

Let us start with some definitions and notation.
Let us fix a non-zero scale diagram $g$ and decompose it at the previous scale, i.e.\@ $g=(t_0,\dots,t_n)$ for some $n\ge 0$.  Recall the definition of the indices $0=i_0<i_1<\dots<i_m$ that label the relevant siblings $t_{i_0},\dots,t_{i_m}$ introduced in Section~\ref{sec: equivalence of diagrams}. 
Let us write
$$
    \Sigma_0 \; = \; \{\pm 1\}^L
$$
for the set of configurations.
%
%
Let $I_{\mathrm{pro}}$ be a subset of $\{i_1,\ldots,i_m\} \subset \{1,\ldots,n\}$ and let 
$$
    I_{\mathrm{ind}} \; = \; \{1,\ldots, n\} \setminus  I_{\mathrm{pro}}.
$$
Note in particular that all $t_i, i \in I_{\mathrm{pro}}$ are relevant siblings (together with $t_0$).  {For the time being, the partition  $\{1,\ldots,n\}=I_{\mathrm{pro}}\cup I_{\mathrm{ind}}$  is not fixed further. We will fix it later depending on the place of $g$ in the tree of a diagram at higher scale}. 
We also recall the set $O(h)$ introduced in Section \ref{sec: equivalence of diagrams}.
 
\begin{proposition}\label{lem: main lemma on the reduction of RV}
We choose a configuration $\sigma\in\Sigma_0$ and a diagram $g$ as above. 
There is a  set of \emph{multiconfigurations}
    \[
        {\caW}(g, \sigma) 
        \; \subset \;  \Sigma_0^{n+1}
    \]
such that the following bound holds for any choice of sets $I_{\mathrm{ind}},I_{\mathrm{pro}}$ as defined above: 
    \begin{equation}\label{eq: main result in lemma reduction RV}
	\left|\langle V(g)|\sigma\rangle\right| 
	\; \le \; 
	\frac{1}{n!}
        \;  \prod_{j \in I_{\mathrm{ind}}} \| A(t_j)\|
        \sum_{\overline \sigma \in {\caW}(g,\sigma)} \; 
|\langle V(t_{0})|\overline\sigma({0})\rangle|   
	\prod_{j \in I_{\mathrm{pro}}} \, | \langle A(t_{j})|\overline\sigma({j})\rangle|  \, 
    \end{equation}
    where  we wrote the elements of $ \Sigma_0^{n+1}$ as 
$$
    \overline{\sigma} \; = \; (\overline{\sigma}(0),\ldots, \overline{\sigma}(n)).
$$
Moreover,  ${\caW}( g, \sigma)=  {\caW}( g', \sigma) $ whenever $g\sim g'$ and there exists a constant $C$ such that 
\begin{equation}\label{eq: bound cardinality W}
    |{\caW}( g, \sigma)| \; \le \;  2^{n + |O(g)|}.
\end{equation}
\end{proposition}


\subsection{Proof of Proposition \ref{lem: main lemma on the reduction of RV}}
 Given $n+1$ operators $M_0,\dots,M_n$, and given an \emph{ordering} $\pi \in \{\pm 1 \}^n$, we define the product $(M_0,\dots,M_n)_\pi$ as follows. 
    First, we set $(M_0)_\pi = M_0$. 
    Next, assuming $(M_0,\dots,M_l)_\pi$ has been defined for $0 \le l < n$, we define $(M_0,\dots,M_{l+1})_\pi$ as 
    \begin{align*}
	(M_0,\dots,M_{l+1})_\pi \; = \;  M_{l+1}(M_0,\dots,M_{l})_\pi \quad \text{if} \quad \pi_l = -1, \\
	(M_0,\dots,M_{l+1})_\pi \; = \;  (M_0,\dots,M_{l})_\pi M_{l+1} \quad \text{if} \quad \pi_l = +1. 
    \end{align*}
    For convenience, we also define the permutation $p_{\pi}$ of $n+1$ elements, such that 
    $$
    (M_0,\dots,M_{n})_\pi =     M_{p_{\pi}(n)} \ldots M_{p_{\pi}(0)}.
    $$ 

We start with the expansion 
	\begin{align*}
		|\langle V(g)|\sigma\rangle|
		\;&\le\;
		\frac1{n!}\left|\langle [A(t_n),[\dots,[A(t_1),V(t_0)]]\dots]|\sigma\rangle\right|\\
		\;&\le\;
		\frac1{n!}\sum_{\pi\in\{\pm 1\}^n} \left|\langle(V(t_0),A(t_1),\dots,A(t_n))_\pi|\sigma\rangle\right|.
	\end{align*}
	Given an ordering $\pi$, we construct a sequence of configurations  $\overline{\sigma}=(\overline{\sigma}(0),\dots,\overline{\sigma}(n))$ with $\overline{\sigma}(i)\in\Sigma_0$ for $0 \le i \le n$, such that 
	\begin{equation}\label{eq: expanding as a full product}
	\langle (V(t_0),A(t_1),\dots,A(t_n))_\pi|\sigma\rangle
		\; = \; 
		\langle V(t_0)|\overline\sigma(0)\rangle \prod_{i=1}^n \langle  A(t_i)|\overline\sigma(i)\rangle, 
	\end{equation}
 This sequence of configurations $\overline{\sigma}$ is given explicitly by 
 \begin{equation}\label{eq: explicit multiconfiguration}
 |\overline\sigma(i)\rangle=  \left( \prod_{{ 0\le j  < p_{\pi}^{-1}(i) } } X_{t_{p_\pi(j)}} \right) |\sigma\rangle, \qquad i=0,\ldots,n
 \end{equation}
 where the permutation $p_\pi$ was introduced  above, and the product is defined to equal $1$
 whenever $p_{\pi}^{-1}(i)=0$.   

We now bound the factors in the product in \eqref{eq: expanding as a full product} in two different ways, depending on whether $i \in I_{\mathrm{ind}}$ or $i\in I_{\mathrm{pro}}$:
\begin{equation}\label{eq: bound of V by relevant constituents and denominators}
\left|\langle V({g})|\sigma\rangle\right| 
\; \le \; 
\left(\prod_{j \in I_{\mathrm{ind}}} \| A(t_j) \| \right) 
\sum_{\pi \in \{\pm 1\}^n}
|\langle V(t_0)| \overline\sigma(0) \rangle | \prod_{j\in I_{\mathrm{pro}}} |\langle A(t_j)|  \overline\sigma(j)\rangle|.
\end{equation}
The right-hand side has already the form claimed in the lemma: Let us write $\overline{\sigma}(\cdot)=\overline{\sigma}(\cdot,g,\pi,\sigma)$ to make all parameters explicit. If we were to define ${\caW}(g,\sigma)$ as the set
\begin{equation}\label{eq: first attemps W}
\{(\overline{\sigma}(j_0,g,\pi,\sigma),\ldots, \overline{\sigma}(j_m,g,\pi,\sigma)) \,  | \, \pi  \in \{\pm 1\}^n \},
\end{equation}
then the bound \eqref{eq: main result in lemma reduction RV} would be satisfied.   However, this choice of ${\caW}(g,\sigma)$  is in general not a class function in $g$.   
Therefore, we proceed in a different way and we enlarge the set defined in \eqref{eq: first attemps W}. 

{For this, we observe that the above reasoning can be carried over starting with any diagram  $g'$ such that $g'\sim g$, instead of $g$ itself. 
We can then consider the function $f$ that associates to every such diagram $g'$ the $m+1$-tuple of restricted configurations}
\[
    f(g') \; = \; 
    (\overline{\sigma}(j_0,g',\pi,\sigma)|_{\overline{I}(t_{j_0})}, \ldots, \overline{\sigma}(j_m,g', \pi,\sigma)|_{\overline{I}(t_{j_m})} )
\]
as obtained above,
with $\pi$ and $\sigma$ considered fixed and 
where $\tau|_S$ with $S\subset \{1,\ldots,L\}$ denotes the restriction of a configuration  $\tau \in \Sigma_0= \{\pm 1 \}^L$ to $\{\pm 1 \}^S$. 
Noting that $n,m,O(g)$ are class functions,  we establish
\begin{lemma}\label{lem: crux of wonderful lemma}
    As $g'$ ranges over all diagrams in the equivalence class $[g]$, $f(g')$ takes at most $2^{|O(g)|}$ different values.
\end{lemma}

\begin{proof}
We recall that the indices $j_i$ of relevant triads/diagrams are class invariants, as well as the active sets $\caA(t_{j_i})$. In particular, if there are no irrelevant triads, then $f$ is constant within classes.  In the general case, i.e.\@ allowing for irrelevant triads, let us inspect how $f(g')$ is affected when $g'$ varies over the class $[g]$.  Irrelevant triads $t$ can affect $f(g')$ if $\caA(t) $ intersects any of the $\overline{I}(t_{j_i}), i=0,\ldots,m$. In other words, $f(g')$ can be affected in any of the points $x\in O(g)$.   
The crucial insight is now that any $x\in O(g)$ is contained in $\overline{I}(t_{j_i})$ for at most one relevant $t_{j_i}$, which is a  
direct consequence of Lemma \ref{lem: 3 diagrams general}. This means that the number of possibilities for $f(g)$ is at most $2^{|O(g)|}$.  
\end{proof}

We can now conclude the proof of Proposition \ref{lem: main lemma on the reduction of RV}. 
Indeed, by Lemma \ref{lem: crux of wonderful lemma}, we find a collection $(g_{p})_{p=1,\ldots,N}$ of diagrams in an equivalence class $[g]$ with $N\leq 2^{|O(g)|}$ and such that all values of $f$ are attained on this collection. 
We define then the set of multiconfigurations
$$
    {\caW}(g,\sigma) 
    \;=\; 
    \{(\overline\sigma(j_0,g_p,\pi,\sigma),\ldots, \overline\sigma(j_m,g_p,\pi,\sigma)) \, |\,   p=1,\ldots,N, \pi \in \{\pm 1\}^n \}
$$
Obviously, $|{\caW}(g,\sigma)| \leq 2^{n+|O(g)|}$ and $g\mapsto {\caW}(g,\sigma)$ is constant on equivalence classes.

\subsection{Expanding a Triad into Its Children} \label{sec: expansion of triads into children}

Above, we bounded $\langle V(g)|\sigma\rangle$ in terms of the children of $g$. We will also need to bound carefully the absolute value of $\langle A(t)| \sigma\rangle$ for triads $t$ in terms of its children.  For the sake of explicitness, we first assume that the triad has three children $t=(f,f',f'')$ and we deal with other possibilities afterwards.
We bound 
\begin{equation}\label{eq: expansion of matrix element triad}
  |\langle A(t)| \sigma\rangle|\leq    |\langle R(t)|\sigma\rangle|  \sum_{\overline{\tau} \in \caK(t,\sigma) }  \prod_{h=f,f',f''} |\langle V(h)|\overline{\tau} (h)\rangle|  
   \end{equation}
with $\caK(t,\sigma) \subset \Sigma_0^3$ defined as 
$$
\caK(t,\sigma) =\{(\sigma,X^{\alpha_1}\sigma,X^{\alpha_2}\sigma)\,|\,  \alpha_1,\alpha_2 \in \{0,1\} \}.
$$
That is, $\caK(t,\sigma) $ has 4 for elements. 
For reasons that will be clear in the next section, we prefer to label the components of multiconfigurations $\overline{\tau}$ by the children of $t$ instead of $1,2,3$. Therefore we write, also in \eqref{eq: expansion of matrix element triad},  $\overline{\tau}= (\overline{\tau} (f),\overline{\tau} (f'),\overline{\tau} (f''))$.

In case $t$ has two children
 $t=(f,f')$, or $t=(f,f'')$, then the variable $h$ in the product in \eqref{eq: expansion of matrix element triad} runs of course over these two children, and $\caK(t,\sigma) $ has two elements, namely   $\caK(t,\sigma) =  \{(\sigma,X^{\alpha}\sigma)\,|\,  \alpha \in \{0,1\} \} $.
In case $t$ has one child $f$, then $\caK(t,\sigma)$ is the singleton $\{\sigma\} \subset \Sigma_0$ and product in \eqref{eq: expansion of matrix element triad} has one factor with $h=f$.


\section{Bounding Matrix Elements}\label{sec: main bound matrix elements}

In this section, we state our main bound on the matrix element $\langle V(g)|\sigma\rangle$, where $g$ is a given diagram and $\sigma$ is a configuration. 
Proposition~\ref{pro: main bound RV} below provides a bound on this quantity in terms of the matrix elements associated with certain descendants of $g$.
This bound involves two types of descendants:
those in $\mathcal{S}_{\mathrm{pro}}(g)$, which feature in the random variable $Y(g,\overline\sigma)$ defined in \eqref{eq: def of Y RV} below,  
and those in the leaves of the subtree $\tree_{\mathcal{P}}(g)$, as defined in \eqref{sec: pruning trees}, which are estimated inductively and contribute a factor $B(g)$, as defined in~\eqref{eq: def of B(g)} below.

In addition, the bound in Proposition~\ref{pro: main bound RV} also involves a sum over a set $\mathcal C(g,\sigma)$ of \emph{multiconfigurations}, introduced in Section~\ref{sec: sets of multiconfigs}.
Thanks to the results established in the previous section, we are able to obtain a good bound on the size of this set.
{While the set $\mathcal C(g,\sigma)$ is class invariant, 
it is important to note that $Y(g,\overline\sigma)$ need not be, as $\mathcal{S}_{\mathrm{pro}}(g)$ itself may not be.}
It will hence also matter to know how many different values the random variable $Y(g,\overline\sigma)$ can take as $g$ varies over the class $[g]$.
These aspects are dealt with in Section~\ref{sec: class dependence C(g,sigma)}.

\subsection{Sets of Multiconfigurations} \label{sec: sets of multiconfigs}

Let $g$ be some non-zero scale diagram to be fixed throughout. 
We will mostly omit $g$ from our notation. 
Recall the tree $\tree = \tree(g)$ introduced in Section~\ref{sec: tree structure}, as well as the pruned tree $\tree_{\mathcal P}$ introduced in Section~\ref{sec: pruning trees}. 
We also define the tree $\tree_{\mathrm{rel}}$ by removing from $\tree$ all triads that are irrelevant siblings, as defined in Section~\ref{sec: equivalence of diagrams}, as well as their descendants.
Recall also that, as in Section~\ref{sec: reduction number of variables}, we write $\caV(\cdot)$ to denote the vertex set of a tree,
and typically use the letters $z$ and $z'$ to refer to vertices.

We will use the set of multiconfigurations
\[
    \Sigma(g) \; = \; \Sigma_0^{\caV(\tree_{\mathrm{rel}}(g))}
\]
consisting of tuples $\overline{\sigma} =(\overline{\sigma}(z))_{z\in \tree_{\mathrm{rel}}}$. 
Given any vertex $z\in \tree_{\mathrm{rel}}(g)$, 
we define $\tree_{\mathrm{rel}}(z)$ as the subtree of $\tree_{\mathrm{rel}}=\tree_{\mathrm{rel}}(g)$ where we keep only the descendants of $z$.  If $z$ is a diagram, then $\tree_{\mathrm{rel}}(z)$ could as well have been defined without reference to the top diagram $g$, which makes this a natural notion.

We will now construct a set of multiconfigurations $\caC(g,\sigma)$ inductively, by moving from the leaves towards the root of the tree $\tree_{\mathrm{rel}}$. Note that all leaves of $\tree_{\mathrm{rel}}$ are diagrams. For any leaf $h$ of $\tree_{\mathrm{rel}}$, we simply set $\caC(h,\sigma)=\{\sigma\} \subset \Sigma_0$.

There are two inductive steps. 
First, given a diagram vertex $h$ {that is not a leaf}, we consider its children $t_0,t_1,\ldots,t_n$ and we need to define the set $\caC(h,\sigma)$ in 
$$
\Sigma(h)= \Sigma_0 \times  \Sigma(t_0) \times \ldots \times  \Sigma(t_n) .
$$
It is defined as 
$$
\caC(h,\sigma)= \bigcup_{\overline{\sigma} \in {\caW}(h,\sigma)}\left( \sigma \times  \caC(t_0,\overline{\sigma}(t_0)) \times \ldots \times  \caC(t_n,\overline{\sigma}(t_n))\right)
$$
where ${\caW}(h,\sigma)$ was introduced in Proposition \ref{lem: main lemma on the reduction of RV} but we slightly changed the notation in that we now label the components of   $\overline{\sigma} $ by the relevant children $t_0,t_{i_1},\ldots, t_{i_m}$ instead of by the indices $i_0,i_1\ldots i_m$.

Second, given a triad vertex $t$, we assume for explicitness that it has three children $f,f',f''$ and we define $\caC(t,\sigma)$ as a subset of 
$$
\Sigma(t)= \Sigma_0 \times  \Sigma(f) \times \Sigma(f') \times  \Sigma(f'') 
$$
by 
$$
\caC(t,\sigma)= \bigcup_{\overline{\tau} \in \caK(t,\sigma)}\left( \sigma \times \caC(f,\overline{\tau}(f)) \times \caC(f',\overline{\tau}(f'))\times \caC(f'',\overline{\tau}(f'')) )\right)
$$
where $\caK(t,\sigma)$ was defined in subsection \ref{sec: expansion of triads into children}.
In case $t$ has two children or a single child, the product over $f,f',f''$ is of course restricted to two factors or to a single factor.


\subsection{Bound on the Matrix Elements} \label{sec: bound on the matrix elements}

Let $z \in \tree_{\mathrm{rel}}(g)$, then we define, for $\overline{\sigma} \in \Sigma(z)$,
\begin{equation}\label{eq: def of Y RV}
Y(z,\overline\sigma) \; = \; \prod_{\substack{t: \mathsf{c}(t)\in\mathcal S_{\mathrm{pro}}(g) \\ \text{$t$ triad descendant of $z$ }}} \langle R(t) |\overline\sigma(t) \rangle 
\end{equation}
with the convention that $Y(t,\overline\sigma)=1$ if the product contains no factor.

  



For vertices $z \in \caP$, we define 
\begin{equation}\label{eq: factor in decomposition primaries}
b(z)\;=\; 
\begin{cases}   
\delta^{\|z\|-|z|} \gamma^{|z|}&  
\text{if $z$ is a triad with $\mathsf c(z)$ crowded,} \\
\delta^{\|z\|-|z|}   (\gamma/\varepsilon)^{|z|}&  
\text{if $z$ is a triad with $\mathsf c(z)$ non-crowded,} \\
\gamma^{|z|}&  \text{if $z$ is a diagram at scale $0$.} 
\end{cases}
\end{equation}
For any $ z\in\tree_\caP(g)$, we define also 
\begin{equation}\label{eq: def of B(g)}
	B(z) \; = \;   
    \prod_{\substack{z' \in \caP(g)\\ \text{$z'$ is descendant of $z$} }}   b(z').
\end{equation}

Our main result is 
\begin{proposition}\label{pro: main bound RV}
Let $k\ge 0$, let $g\in\mathcal G^{(k)}$, let $z \in \tree_{\caP}(g)$, and let $\sigma\in\Sigma_0$.
The following bound holds on $\mathbf{NR}(g)$:
	\[
		|\langle M(z)|\sigma\rangle| \; \le \;  \frac{1}{z!} B(z) \sum_{\overline{\sigma}\in\mathcal C(z,\sigma)} Y(z,\overline\sigma),
	\]
    where $M(z)=V(z)$ if $z$ is a diagram and $M(z)=A(z)$ if $z$ is a triad.
\end{proposition}


\begin{proof}
The proof is made out of three steps. 

We first check the case when $z$ is a leaf of $\tree_{\caP}$, namely an element of $\caP$ itself.
Then, the bound on the right hand side is simply $\frac{1}{z!} b(z)$ and the proposition holds {by Proposition~\ref{pro: main inductive bounds} and the bound~\eqref{eq: special k=0 bound} at scale $k=0$}. In particular, all vertices $z\in \tree_\caP$ at scale $0$ are necessarily leaves (and diagrams) and hence for these the proposition holds true. 
Apart from this, the bound is inductive in the scale of $z$.


Second, let us now assume that the proposition holds for all vertices $z$ with scale at most $k$, for some $k\geq 0$.  We will show that it holds also for for triads $z=t$ at scale $k+1$. 
For concreteness, we assume again that  $t$ has three children $(f,f',f'')$. 
By subsection \ref{sec: expansion of triads into children} and the induction hypothesis, we have 
\[
 \left|\langle A(t)|\sigma\rangle\right| 
 \;\leq\;   
 \left|\langle R(t)|\sigma\rangle\right|        \sum_{ \overline{\tau}\in \caK(t,\sigma)}  
 \,\prod_{h=f,f',f''} \, 
\left(  \frac{B(h)}{h!} 
\sum_{\overline{\sigma} \in \caC(h,\overline{\tau}(h))}
 Y(h,\overline{\sigma})
\right)
\]
abd we straightforwardly  check that the right-hand side  is indeed the bound for $ \left|\langle A(t)|\sigma\rangle\right|$ claimed in Proposition \ref{pro: main bound RV}. For this we use the inductive definition of $\caC(t,\sigma)$ given in \ref{sec: sets of multiconfigs} and the multiplicative definitions $t!=\prod_h h!$ and \eqref{eq: def of B(g)}.

Finally, we assume that the proposition holds for all diagrams with scale at most $k$ and all triads with scale at most $k+1$, for some $k\geq 0$.  We show now that it also holds for diagrams $g$ at scale $k+1$. 
Let us decompose $g$ as $g = (t_0,\dots,t_n)$ and let us apply Proposition~\ref{lem: main lemma on the reduction of RV} with $I_{\mathrm{ind}}$ such that $i\in I_{\mathrm{ind}}$ if $t_i\in\mathcal P(g)$, and $I_{\mathrm{pro}}$ such that $i\in I_{\mathrm{pro}}$ if $i\ne 0$ and $t_i\notin \mathcal P(g)$. 
Here, the subscripts ``$\mathrm{ind}$" and ``$\mathrm{pro}$" refer to ``inductive" and ``probabilistic".
We get
\begin{align*}
	\left|\langle V(g)|\sigma\rangle\right|
	&
\; \le \; 
	\frac{1}{n!}
        \;  \left(\prod_{j \in I_{\mathrm{ind}}} \| A(t_j)\| \right)
        \sum_{\overline \tau \in {\caW}(g,\sigma)} \; 
|\langle V(t_{0})|\overline\tau({0}\rangle| 
	\prod_{j \in I_{\mathrm{pro}}} \, | \langle A(t_{j})|\overline\tau({j})\rangle|  \\
	&\; \le \; 
	\frac{1}{n!}
        \;  \left(\prod_{j \in I_{\mathrm{ind}}} \frac{b(t_j)}{t_j!} \right)
        \sum_{\overline \tau \in {\caW}(g,\sigma)} \; 
\prod_{j \in I_{\mathrm{pro}} \cup \{0\}} 
\left(  \frac{B(t_j)}{t_j!} 
\sum_{\overline{\sigma} \in \caC(t_j,\overline{\tau}(t_j))}
 Y(t_j,\overline{\sigma})
\right)
\end{align*}
and the latter expression again equals the desired bound. 
\end{proof}

\subsection{Bound on $B(g)$}

We render the quantity $B(z)$ introduced in \eqref{eq: def of B(g)} more explicit when $z=g$: 
\begin{lemma}
\begin{equation}\label{eq: bound on bz}
    |B(g)| 
    \;\leq\; 
    \delta^{\| g\| -|g|}\frac{1}{g!}\frac{\gamma^{|g|}}{\varepsilon^{|g|/2}}.
\end{equation}
\end{lemma}
\begin{proof}
First, we note that for diagrams at scale $0$, $|g|=\|g\|$, so that the last line of \eqref{eq: factor in decomposition primaries} can also be written as $\delta^{\|z\|-|z|} \gamma^{|z|}$.
Then, if $g$ is a non-zero scale diagram, we observe that $|g|=\sum_{z\in \caP}|z|$ and $\|g\|=\sum_{z\in \caP}\|z\|$.
Indeed, the first equality is the second claim of Lemma~\ref{lem: simple about primaries}, and the second equality follows trivially since the bare order of a diagram/triad is always the sum of the bare orders of its children. 
Then, we have
$$
    B(g) \;\leq\;  
    \delta^{\|g\|-|g|}  \frac{1}{g!}\gamma^{|g|} (1/\varepsilon)^{\sum_{z \in \caP_{\mathrm{nc}}}  |z|}
$$
where $\caP_{\mathrm{nc}}\subset \mathcal P$ contains all vertices in $\mathcal P$ that are triads with a non-crowded central diagram.
In view of the definition of $\mathcal S_{\mathrm{ind}}$ in Section~\ref{subsec: denominators estimated by probabilistic or inductive bounds}, on which the definition of $\mathcal P$ in Section~\ref{sec: pruning trees} is based,
we realize that $\caP_{\mathrm{nc}} \subset \overline\caP_{\mathrm{fo}}$. 
Therefore, by the bound 
\eqref{eq: corolary norm loss lemma}, we get that $\sum_{z \in \caP_{\mathrm{nc}}}  |z| \leq |g|/2$, which implies the claim of the lemma, since $\varepsilon\leq 1$.
\end{proof}

\subsection{Cardinality of $\mathcal C(g,\sigma)$ and Class-Dependence}\label{sec: class dependence C(g,sigma)}

We address here counting aspects, which are covered in Lemma~\ref{lem: counting values of y} below.
We first need some preliminaries. 
%
%
%
If two diagrams $g$ and $g'$ are equivalent, the trees $\tree_{\mathrm{rel}}(g)$ and $\tree_{\mathrm{rel}}(g')$ are isomorphic as graphs. 
Indeed, the corresponding bijection $\iota_{g,g'}:\mathcal V(g)\to\mathcal V(g')$, where we recall that $\mathcal V(\cdot)$ denotes the set of vertices of a graph, is readily constructed following the inductive definition of equivalence classes given in subsection \ref{sec: equivalence of diagrams}. 
First, we set $\iota_{g,g'}(g)=g'$. 
Then we proceed to the children of $g$ and $g'$: 
we have $g=(t_0,\ldots,t_n), g'=(t'_0,\ldots,t'_n)$ and, for both $g$ and $g'$, the relevant children have the same labels $i_0,\ldots,i_m$. 
We then set $\iota_{g,g'}(t_{i_j})=t'_{i_j}$ for $j=0,\ldots,m$, and we continue in this way down the trees $\tree_{\mathrm{rel}}(g)$ and $\tree_{\mathrm{rel}}(g')$. 

This bijection can also pushed back to a bijection between the sets of multiconfigurations $\Sigma(g)$ and $\Sigma(g')$: 
\[
    {\iota_{g,g'} (\overline{\sigma} (z)) 
    \; = \;  
    \overline{\sigma}(\iota_{g,g'}(z)). }
\]
We also observe that the construction of $\caC(z,\sigma)$ in \ref{sec: sets of multiconfigs} satisfies the following property: 
if two diagrams $g$ and $g'$ are equivalent, and if $z \in \tree_{\mathrm{rel}}(g)$, then
\[
    \iota_{g,g'}\caC(z,\sigma) \; = \; \caC(\iota_{g,g'} z,\sigma). 
\]


We are now ready to state the result of this subsection: 
\begin{lemma}\label{lem: counting values of y}
\begin{enumerate}
    \item  For any diagram $g$ and any configuration $\sigma$, 
\[
    |\mathcal C(g,\sigma)| \; \le \; C^{|g|},
\]
\item 
{For a given multiconfiguration $\overline\sigma \in \caC(g,\sigma)$, the number of different values that $Y(g,\overline\sigma)$ can assume as $g$ varies over its equivalence class $[g]$, is bounded by $C^{|g|}$ (notice that the order of a diagram and the set $\caC(g,\sigma)$ is a class function) }
\end{enumerate}

\end{lemma}

\begin{proof}
We begin with the first item.
The proof follows similar reasoning to that used in Sections~\ref{sec: counting diagrams} and~\ref{sec: counting equivalence classes of diagrams}, and we highlight only the key points here.
By the construction of the sets $\caC(z,\sigma)$ in subsection \ref{sec: sets of multiconfigs}, for $z\in\tree$ and a configuration $\sigma$, we see that the cardinality $|\caC(z,\sigma)|$ is independent of $\sigma$ and we abbreviate
$a(z)=|\caC(z,\sigma)|$.
By the inductive construction in subsection \ref{sec: sets of multiconfigs} and by the bound~\eqref{eq: bound cardinality W} in Proposition~\ref{lem: main lemma on the reduction of RV}, we have 
\begin{equation}\label{eq: iteration ag}
    a(g) \;\leq\; 2^{n+|O(g)|} \prod_{\text{$t$ relevant}} a(t)
\end{equation}
where $t$ are the children of a diagram $g$, and where $n+1$ denotes the number of children of $g$. 
Similarly, if $t$
 is a triad and $h$ are its children, then
 \begin{equation}\label{eq: iteration at}
   a(t)\; \leq \; 4 \prod_{h} a(h). 
\end{equation}

We will now prove the claim by induction on the scale. 
More precisely, we prove that there exists a bounded, increasing sequence of constants $(C_k)_{k\ge 0}$ such that $|a(h)| \leq C_k^{|h|}$ for any diagram $h \in \caV(\tree_{\mathrm{rel}}(g))$ at scale $k$.

For diagrams $g$ at scale $0$, the claim holds with $C_0=1$. 
Let then $k\ge 1$, and let us assume that the claim holds for all diagrams in $\caV(\tree_{\mathrm{rel}}(g))$ at scale $k'$ with constant $C_{k'}$, for all $0\leq k'\leq k$.  
Then, for a triad $t \in \caV(\tree_{\mathrm{rel}}(g))$ at scale  $k$, we have $|t|=\sum_h |h|$ where the sum is over the diagram children of $h$ and we used that $\mathsf c(t)$ is not crowded. Therefore, by \eqref{eq: iteration at}, we have
\[
    a(t) \;\leq\; 4 C_k^{|t|}.
\]
Let us now consider a diagram $h \in \caV(\tree_{\mathrm{rel}}(g))$ at scale  $k+1$. By \eqref{eq: iteration ag}, we have
\begin{equation}\label{eq: iteration ag used}
    a(h) \;\leq\; 2^{n+|O(h)|} \prod_{\text{$t$ relevant}} 4 C_k^{|t|},
\end{equation}
with $t$ the children of $h$, and $n+1$ the number of children of $h$.
Moreover, by the definition of $O(h)$ in Section~\ref{sec: equivalence of diagrams}, we have
\[
|O(h)| \;\leq\; 
\sum_{\text{$t$ irrelevant}}  |I(t)|  
\;\leq\; \sum_{\text{$t$ irrelevant}}  |t| 
\; = \;   |h|- \sum_{\text{$t$ relevant}} |t|.
\]
We also observe that 
\begin{equation}\label{eq: instance of k bound}
n \;\leq\; \frac{|h|}{\beta L_k}.
\end{equation}
Inserting these two bounds in \eqref{eq: iteration ag used}, we get
\begin{equation*}
    a(h) 
    \;\leq\; 
    8^{n} 2^{|O(h)|} \prod_{\text{$t$ relevant}}  C_k^{|t|} 
    \;\leq\; 8^{\frac{1}{\beta L_k} |h|}  (\max{(2,C_k)})^{|h|}. 
\end{equation*}
Hence, the bound propagates if we set $C_{k+1}=  8^{\frac{1}{\beta L_k}} \max{(2,C_k)} $. 
Since $L_k$ increases exponentially in $k$, this sequence is bounded, uniformly in $g$ and hence item 1 follows.

Let us move to the second item of the lemma. 
Analogous to the random variable $Y(g,\overline\sigma)$ defined in \eqref{eq: def of Y RV}, let us define 
\[
    Z(g,\overline\sigma) 
    \; = \; 
    \prod_{\text{triad } t \in \mathcal V(\tree_{\mathrm{rel}}(g))}
     \langle R(t) |\overline\sigma(t) \rangle
\]
This is a class function, i.e.\@ $Z(g,\overline\sigma) = Z(g',\overline\sigma)$ whenever $g\sim g'$. 
Indeed, from the definition of $R(t)$ in Section \ref{subsec: main results on inductive bounds}, $R(t)=R(t')$ whenever $t\sim t'$; the property then follows from the definition of the bijection $\iota_{g,g'}$ at the beginning of this subsection. 
Further, we observe that 
\[
    \lbrace \text{triad } t : \mathsf c(t) \in \mathcal S_{\mathrm{pro}}(g) \rbrace
    \; \subset \;
    \lbrace \text{triad } t \in \mathcal V(\tree_{\mathrm{rel}}(g))\rbrace .
\]
Hence, for every $g'\in[g]$, $Y(g',\overline\sigma)$ can be obtained from $Z(g,\overline\sigma)$ by deciding for every triad $t \in \mathcal V(\tree_{\mathrm{rel}}(g))$ whether to keep or not the corresponding factor $\langle R(t)|\overline\sigma(t)\rangle$.
This implies that $Y(g',\overline\sigma)$ takes at most $2^{|\mathcal V(\tree_{\mathrm{rel}}(g))|}$ different values as $g'$ varies over $[g]$.
Therefore, it suffices to show that there exists a constant $C$ such that
\begin{equation}\label{eq: cardinality of Tree rel}
    |\mathcal V(\tree_{\mathrm{rel}}(g))| \; \le \; C|g|.
\end{equation}

The bound \eqref{eq: cardinality of Tree rel} is derived inductively on the scale, in a similar way as for the first item of the lemma.
For simplicity, let us write $a(g) = |\mathcal V(\tree_{\mathrm{rel}}(g))|$. 
We set up a bounded, increasing sequence of running constants $(C_k)_{k\ge 0}$. 
The claim is true at scale $k=0$ with $C_0 = 1$. 
Let us now assume that the claim has been shown up to some scale $k\ge 0$, and let $g$ be a diagram at scale $k+1$.
We decompose $g = (t_0,\dots,t_n)$ and $t_i = (g_i,g_i',g_i'')$ for $1 \le i \le n$. 
We find 
\[
    a(g) \; \le \; 
    1 + a(t_0) + \sum_{\substack{1 \le i \le n,\\ t_i \text{ relevant}}} \left( 1 + a(g_i) + a(g_i') + a(g_i'')\right). 
\]
Using our inductive hypothesis, the fact that relevant triads have a non-crowded central diagram, and observing that the bound~\eqref{eq: instance of k bound} still holds for $n+1$ instead of $n$ in the left-hand side, this is bounded as 
\[
    a(g) \; \le \; n+1 + C_{k-1} |g| \; = \; \left( \frac{n+1}{|g|} + C_{k-1} \right) |g|
    \; \le \;
    \left( \frac{1}{\beta L_k} + C_{k} \right) |g|. 
\]
This yields \eqref{eq: cardinality of Tree rel}, provided we set $C_{k+1} = C_k + 1/\beta L_k$.
\end{proof}
%
%
%

\section{Conclusion of the Proof of Proposition~\ref{prop: prob of resonances}}\label{sec: conclusion proof probability resonance}

We need to prove \eqref{eq: main probabilistic bound NRII} in Proposition~\ref{prop: prob of resonances}. For this, we fix a triad class $[t]$ and without loss of generality, we may assume that for each $t=(g,g',g'')\in[t]$, $g$ is non-crowded.  Below, we always use $t,g$ in this sense, i.e.\ $t\in [t]$ and $g=\mathsf c(t)$.

We split the proof in two lemmas. 
The first one, Lemma \ref{lem: first part of proof prop 3}, uses mainly the results from Sections~\ref{sec: inductive bounds},~\ref{sec: reduction number of variables}, \ref{sec: expansion of vertices into children}~and~\ref{sec: main bound matrix elements}. To state it, we introduce the following notation.
Let  $(\overline r(h),\overline s(h),\overline \sigma(h))_{h \in \mathcal S_{\mathrm{pro}}(g)}$ be collections labelled by $h \in \mathcal S_{\mathrm{pro}}(g)$ of configurations $\overline \sigma(h)$ and of naturals  $r(h),s(h)$ satisfying
$$
0 \;\leq\; r(h) \leq  r(t'),\qquad  0 \;\leq\; s(h) \;\leq\;  s(t'),
$$
where $t'$ is the triad such that $h=\mathsf c(t')$ and $r(t'),s(t')$ are the offset indices of the triad $t'$, introduced in Section \ref{sec: second representation A}.
{We also recall that the diagonal operators \((\partial_h \widetilde{E})_{u,v}\) are defined in~\eqref{eq: rewriting denominator}, together with the tilde operation in~\eqref{eq: def tilded energies}.}

\begin{lemma}\label{lem: first part of proof prop 3} 
    There exists a constant $C$ such that for any $0<\alpha\leq 1$, 
    \begin{multline*} 
    \mathbb P\left( \bigcup_{t\in[t]}\left(\mathrm{NR}_{\mathrm{II}}(t)\right)^c \cap \mathbf{NR}([t]) \right)   \; \leq \; \\
    C^{(1+\alpha)|g|}\varepsilon^{\alpha |g|/2}
  \max_{t\in [t],\overline\sigma,\overline{r},\overline{s}}  
    \left(\mathbb E\left( \prod_{h\in\mathcal S_{\mathrm{pro}}(g)}\frac{1}{|\langle  (\partial_h \widetilde E)_{\overline r(h),\overline s(h) }|\overline \sigma(h)\rangle |^{3\alpha}} \right)\right)^{1/3} .
\end{multline*} 
Indeed, the right hand side is determined by $[t]$ because  $|t|, |g|$ are class variables and $g=\mathsf c(t)$. 
\end{lemma}

\begin{proof}

Given a triad $t \in [t]$, the event $\left(\mathrm{NR}_{\mathrm{II}}(t)\right)^c$ implies that 
\[
    \max_{\sigma\in\Sigma_0} |\langle V(g) \partial_g V(g')\partial_g V(g'')R(t)|\sigma\rangle|
    \; > \; 
    B_{\mathrm{II}}(t)
\]
as can be seen from the definition \eqref{eq: NR II} together with \eqref{eq: explicit form A k+1 triad} and (\ref{eq: denominator operators 1 and 2},\ref{eq: denominator operator}).
On $\mathbf{NR}([t])$, we can bound $\|\partial_g V(g')\|,\|\partial_g V(g'')\|$ using \eqref{eq: bound diagram derivative} and \eqref{eq: inductive bound V tilde}. Using that $\partial_g V(g'),\partial_g V(g''),R(t)$ are diagonal operators, and that $g$  is non-crowded and hence $|t|=|g|+|g'|+|g''|$, we conclude that $\left(\mathrm{NR}_{\mathrm{II}}(t)\right)^c \cap \mathbf{NR}([t])$ implies 
\begin{equation}\label{NR II c made explicit}
\max_{\sigma\in\Sigma_0} \left|\langle V(g)R(t)|\sigma\rangle \right| 
\; > \; 
\frac{1}{g!} \delta^{\|g\|-|g|}\left(\frac{\gamma}{\varepsilon}\right)^{|g|}.
\end{equation}


Next, by Proposition~\ref{pro: main bound RV} and the definition of $B(g)$ in \eqref{eq: def of B(g)}, we conclude that \eqref{NR II c made explicit} implies 
\[
    \max_{\sigma\in\Sigma_0}  \left|\langle R(t)|\sigma\rangle \sum_{\overline\sigma \in\mathcal C(g,\sigma)}Y(g,\overline\sigma)\right| 
    \; > \; 
    \frac{1}{\varepsilon^{(1 - \frac12)|g|}}
    \; = \; 
    \frac{1}{\varepsilon^{|g|/2}}.
\]
By the first item of Lemma \ref{lem: counting values of y}, we can bound $|\sum_{\overline{\sigma}}  Y(g,\overline\sigma)| \leq C^{|g|} \max_{\overline{\sigma}} |Y(g,\overline\sigma)|  $ 
and hence 
\begin{multline*}
 \bigcup_{t\in[t]}\left(\mathrm{NR}_{\mathrm{II}}(t)\right)^c \cap \mathbf{NR}([t]) 
\; \subset \; 
\left\{\max_{\sigma \in \Sigma_0} \max_{t \in [t]} \max_{\overline\sigma \in \caC(g,\sigma)} |\langle R(t)|\sigma\rangle   |Y(g,\overline \sigma)|  \geq \frac{1}{(C\varepsilon)^{|g|/2}}\right\}.
\end{multline*}

To parse this formula, we first recall that the triad \(t\) fixes the diagram \(g\), and this should be kept in mind throughout the present proof. 
To estimate the probability of the event on the right-hand side, we bound the number of possible values that the variable \(\langle R(t) | \sigma \rangle \, Y(g, \overline{\sigma})\) can take as \(t\), \(\overline{\sigma}\), and \(\sigma\) range over \([t]\), \(\mathcal{C}(g, \sigma)\), and \(\Sigma_0\), respectively.
We start by fixing \(\sigma\). In this case, \(\langle R(t) | \sigma \rangle\) takes a single value, since \(R(t)\) depends only on the class \([t]\) of \(t\). Meanwhile, \(Y(g, \overline{\sigma})\) takes at most \(C^{|g|}\) values, as deduced from both items of Lemma~\ref{lem: counting values of y}, by first varying \(t\) for fixed \(\overline{\sigma}\), and then varying \(\overline{\sigma}\).
Finally, varying \(\sigma\) yields at most \(2^{|\overline{I}(t)|} \le 2^{|t| + 2}\) values, since only the restriction of the configuration \(\sigma\) to \(\overline{I}(t)\) influences the values of \(\langle R(t) | \sigma \rangle\) and \(Y(g, \overline{\sigma})\).



We get hence 
 \begin{multline}\label{eq: first bound on total prob}
\mathbb P\left( \bigcup_{t\in[t]}\left(\mathrm{NR}_{\mathrm{II}}(t)\right)^c \cap \mathbf{NR}([t]) \right)
\; \leq \;  \\
\max_{\sigma \in \Sigma_0} \max_{t \in [t]} \max_{\overline\sigma \in \caC(g,\sigma)}     2^{|t|+2} C^{|g|}
 \mathbb{P}\left( \{ |\langle R(t)|\sigma\rangle   |Y(g,\overline \sigma)|  \geq \frac{1}{(C\varepsilon)^{|g|/2}} \}  \cap \mathbf{NR}([t])\right).
\end{multline}
By inspecting \eqref{eq: def of Y RV}, we note that  
$$
\langle R(t)|\sigma\rangle   Y(g,\overline \sigma)  \; = \; \prod_{t': \mathsf{c}(t')\in\mathcal S_{\mathrm{pro}}(g) } \langle R(t') |\overline\sigma(t') \rangle.
$$
Two remarks are in order. First, the product on the right hand side is meant to include also the triad $t$ itself. 
{This is to be contrasted with the expression~\eqref{eq: def of Y RV} in the previous section, where the triads involved in the product are all assumed to be descendants of $g$.}
Second, on the right hand side, $\overline \sigma \in \Sigma(t)$ is obtained from $\overline\sigma \in \Sigma(g)$ on the left hand side by adding the configuration $\overline\sigma(t)=\sigma$ (and some arbitrary configurations for the children $\mathsf l(t),\mathsf r(t)$ of $t$). 

Therefore, and invoking the Markov inequality with some fractional exponent $0<\alpha\leq 1$, \eqref{eq: first bound on total prob} is bounded by
\begin{equation}\label{eq: first appearance fractional bound}
2^{|t|+2}C^{(1+\alpha)|g|}\varepsilon^{\alpha |g|/2}
\max_{t\in [t],\overline\sigma \in \Sigma(t)}\mathbb E\left( \prod_{t':\mathsf{c}(t')\in\mathcal S_{\mathrm{pro}}(g)}|\langle R(t')|\overline\sigma(t')\rangle |^\alpha 1_{\mathbf{NR}([t])} \right).
\end{equation}
 We now bound, for any configuration $\sigma$, 
\begin{align*}
|\langle R(t')|\sigma\rangle | 1_{\mathbf{NR}([t])}  
\;&\leq\;   
|\langle \widetilde R(t')|\sigma\rangle | 
\;\leq\;  \sum_{j=1,2}\prod_{i=1,2,3}   \left|\left\langle \frac{1}{(\partial_h \widetilde E)_{r_{i,j},s_{i,j}}}  \middle|\sigma\right\rangle \right| \\
\;& =\;   
\sum_{j=1,2}\prod_{i=1,2,3}   \frac{1}{|\langle  (\partial_h \widetilde E)_{r_{i,j},s_{i,j}} |\sigma\rangle |  }  
\end{align*}
for some $(r_{i,j},s_{i,j})_{i=1,2; j=1,2,3}$ satisfying $0\leq r_{i,j} \leq r(t')$ and  $0\leq s_{i,j} \leq s(t')$. 
Therefore, $\max_{\ldots} \mathbb{E} (\ldots)$ in \eqref{eq: first appearance fractional bound} is bounded by 
\begin{equation}\label{eq: first appearance of offset collections}
 2^{|\mathcal S_{\mathrm{pro}}(g)|} \max_{t\in [t],\overline\sigma,\overline{r},\overline{s}}  
    \left(\mathbb E\left( \prod_{h\in\mathcal S_{\mathrm{pro}}(g)}\frac{1}{|\langle  (\partial_h \widetilde E)_{\overline r(h),\overline s(h) }|\overline \sigma(h)\rangle |^{3\alpha}} \right)\right)^{1/3}.
\end{equation}
To get the bound \eqref{eq: first appearance of offset collections}, we used also that $ (x+y)^\alpha \leq x^\alpha+x^{\alpha}$ for nonnegative $x,y$ and $0\leq \alpha \leq 1$ and  the H{\"o}lder inequality $(\mathbb{E}(XYZ))^3\leq \mathbb{E}(X^3) \mathbb{E}(Y^3)\mathbb{E}(Z^3) $ for nonnegative $X,Y,Z$.

Proposition~\ref{prop: summary of erasure} implies that $|\mathcal S_{\mathrm{pro}}(g)|\le 6|I(g)| \le 6 |g|$, so the factor in front of $\max (\ldots)$ is bounded by $C^{|g|}$. Then, the claim of the lemma follows by combining the bounds in 
\eqref{eq: first appearance fractional bound} and \eqref{eq: first appearance of offset collections}, and using $|t| \leq 3(1+\beta) |g|$, from \eqref{eq: lower bound norm diagram} and \eqref{eq: bound on order triads}. 
\end{proof}

We now come to the second lemma in the proof of \eqref{eq: main probabilistic bound NRII}, namely Lemma \ref{lem: prob bound product denominators} below. 
This relies entirely on Section \ref{sec: selecting integration variables}. 
Let us fix some choice of $\overline\sigma,\overline r,\overline s$ (as defined above)  and  introduce the short-hand notation 
\[
	\Delta_h \; = \; \langle  (\partial_h \widetilde E)_{\overline r(h),\overline s(h) }|\overline \sigma(h)\rangle.
\]

\begin{lemma}\label{lem: prob bound product denominators}
We let the diagram $g$ be as above.
There exists a constant $C$ such that, uniformly in  $\overline\sigma,\overline r,\overline s$ and in $0<\alpha \leq 1/19$,
\begin{equation*}
    \mathbb E\left( \prod_{h\in\mathcal S_{\mathrm{pro}}(g)}\frac{1}{|\Delta_h|^{3\alpha}} \right)
	\; \le \;   C^{|g|}.
\end{equation*}
\end{lemma}

\begin{proof}
Considering the partition of $\mathcal S_{\mathrm{pro}}(g)$ in Proposition~\ref{prop: summary of erasure}, we write the bound 
\begin{equation}\label{eq: expectation that will be estimated as integral}
    \mathbb E\left( \prod_{h\in\mathcal S_{\mathrm{pro}}(g)}\frac{1}{|\Delta_h|^{3\alpha}} \right)
	\; \le \; 
	\left(\prod_{j=1}^6 \mathbb E\left( \prod_{h\in\mathcal S_j'(g)} \frac{1}{|\Delta_h|^{18\alpha}}\right)\right)^{1/6}.
\end{equation}

Let us now fix $1\le j \le 6$ in \eqref{eq: expectation that will be estimated as integral} and let us estimate the corresponding expectation as an explicit integral. 
To avoid (artificial) problems at large values of the denominators, we use the bound 
\[
\frac{1}{|\Delta_h|} \; \le \; \frac{1_{\{|\Delta_h| \le 1\}}}{|\Delta_h|} + 1. 
\]
Proposition~\ref{prop: summary of erasure} implies that that $|\mathcal S'_j(g)| \le |g|$.
Hence, our expectation is upper-bounded by a sum of at most $2^{|g|}$ terms, each being an expectation as in \eqref{eq: expectation that will be estimated as integral}, with some denominators being replaced by 1. 
To streamline the discussion, we will now focus on the term containing all of them (terms with less denominators can be dealt with in the same way).
Proposition~\ref{prop: summary of erasure} features two cases, and we will now assume that the points $x_1^j,\dots,x_{m_j}^j$ form a strictly increasing sequence such that $x_i^j>\max I(h_1^j)\cup\dots\cup I(h_{i-1}^j)$ for all $2\le i \le m_j$, with the notations introduced there; the other case is analogous. In what follows, we fix $j$ and we omit it as a subscript/superscript.

In estimating our integral, we will choose $(\theta_{x_j})_{1\le i \le m}$ as integration variables, keeping all other random fields as spectators: 
\[
    \mathbb E\left( \prod_{h\in\mathcal S_j'(g)} \frac{1_{\{|\Delta_h|\le 1\}}}{|\Delta_h|^{18\alpha}}\right)
    \; = \; 
    \int \left(\prod_{y:y\ne x_1,\dots ,x_{m}} d \theta_y\right) 
    \int \left(\prod_{i=1}^{m} d\theta_{x_i}\right)
    \frac{1_{\{|\Delta_{h_i}|\le 1\}}}{|\Delta_{h_i}|^{18\alpha}}.
\]
To evaluate this integral, we perform the change of variables 
\[
	u_i \; = \; \Delta_{h_i}(\theta_{x_1},\dots,\theta_{x_m}), \qquad 1 \le i \le m. 
\]
{To proceed, we need a lower bound on the determinant of the Jacobian matrix $J = (J_{i,l})_{1\le i \le m}$ defined as 
\[
    J_{i,l} 
    \; = \; 
    \frac{\partial u_i}{\partial\theta_{x_l} }
    \; = \;
    \frac{\partial \Delta_{h_i}}{\partial\theta_{x_l} }
\]
evaluated on any point $\theta\in[0,1]^m$. 
For this, we first check that the conditions of Lemma~\ref{lem: spectral lemma} in Appendix~\ref{sec: appendix spectrum} are satisfied for the matrix $\frac13 J$. Indeed,
\begin{enumerate}
    \item By \eqref{eq: corollary disorder dependence distance} in Corollary~\ref{cor: derivative tilded energy}, we find that all elements of $J$ satisfy $|J_{i,l}| \le 3$.
    \item  The points $x_1,\dots,x_m$ are such that $d(x_l,I(h_i))\ge l - i$ for $l>i$. 
Hence, \eqref{eq: corollary disorder dependence distance} yields the bound $|J_{i,l}|\le 3 (C\delta)^{l-i}$ for $l>i$. 
\item  The bound \eqref{eq: corollary disorder dependence} in Lemma~\ref{cor: derivative tilded energy} yields $|J_{i,i}| \ge 2 - C\delta$.
We used here in particular that $x_i\in \caA(h_i)$ and hence  {$\frac{\partial }{\partial \theta_{x_i}} \langle \partial_{h_i} E^{(0)}|\overline{\sigma}(h_i)\rangle=\pm 2 $, by \eqref{eq: partial g of e zero}}.
\end{enumerate}
 Next, 
We can thus apply Lemma~\ref{lem: spectral lemma} with $\epsilon=C\delta$ (note that this      parameter $\epsilon$ appears only in Appendix \ref{sec: appendix spectrum} and it is unrelated to $\varepsilon$). 
In addition, 
\[
    |\det J|
    \; = \; 
    3^m \left|\det \left(\frac13 J\right)\right|
    \; = \; 
    3^m \prod_{\lambda\in\mathrm{spec}(\frac13 J)}|\lambda|
    \; \ge \; 
    3^m \left(\frac23 - C\delta\right)^m
    \; \ge \; 
    1.
\]}
Therefore 
\[
    \int \left(\prod_{i=1}^{m} d\theta_{x_i}\right)
    \frac{1_{\{|\Delta_{h_i}|\le 1\}}}{|\Delta_{h_i}|^{18\alpha}}
    \; \le \;
    \int_{[-1,1]^m} \frac{d u_j}{|u_j|^{18\alpha}}
    \; \le \; C^m 
    \; \le \; C^{|g|}
\]
{provided $\alpha \leq 1/19$. }
Inserting this bound into previous estimates yields the claim, provided $\varepsilon$ is taken small enough. 
\end{proof}

Proposition \ref{prop: prob of resonances} follows now directly by combining Lemmas \ref{lem: first part of proof prop 3} and \ref{lem: prob bound product denominators}.


\section{Counting Equivalence Classes of Diagrams}\label{sec: counting equivalence classes of diagrams}

In subsection \ref{sec: equivalence of diagrams},
we introduced  diagram equivalence classes $[g]$  and triad equivalence classes $[t]$.  
From the definition, we see that the following diagram/triad functions are actually class functions 
\[
    |\cdot|, \qquad I(\cdot), \qquad  \mathcal A(\cdot).
\]
In particular, the property of being crowded is also a diagram class function. 
Moreover, the functions $\mathsf c,\mathsf l,\mathsf r$  are class functions as well, mapping a triad class to a diagram class.

We proceed now to counting these equivalence classes. 
Let $x\in\Lambda_L$, $k\ge 0$ and $w\ge 1$ be integers, and let
\begin{align*}
   \widehat N(x,k,w) & \;=\;  
   \left|\{ [g] \, |\, g \in \mathcal G^{(k)}: |g|=w,\,\min I(g)=x\} \right|, \\
   \widehat N_{\caT}(x,k,w)& \;=\;  
   \left|\{[t] \, |\, t \in \mathcal T^{(k)}: \mathsf c(t) \, \text{non-crowded},\, |t|=w,\, \min I(t)=x\} \right|,
\end{align*}
this is, the number of equivalence classes with starting point, order and scale fixed.
We will prove the following result:
\begin{proposition}\label{prop: combinatoric for type two resonances}
There exists a constant $C$ such that for every $x,k,w$ as above, 
    \begin{equation*}
        \widehat  N(x,k, w) \leq C^{w},  
        \qquad  
        \widehat N_\caT(x,k,w) \leq w^8C^{w}.
    \end{equation*}
\end{proposition}
The remainder of this section is devoted to the proof of this proposition.  
As in the proof of Proposition~\ref{thm: counting diagrams} in Section~\ref{sec: counting diagrams}, the argument proceeds by induction on the scale $k$.  
As there, we introduce an increasing and bounded sequence $(C_k)_{k \ge 0}$ and prove that
\begin{equation}\label{eq: thing to prove counting classes}
    \widehat N(x, k, w) \leq C_k^{w},  
    \qquad  
    \widehat N_\mathcal{T}(x, k, w) \leq w^8 C_k^{w}.
\end{equation}
The sequence $(C_k)$ is defined as in~\eqref{eq: definition running constant}, possibly with a different constant $a > 0$.

\subsection{From Diagrams to Triads}\label{sec: from diagrams to triads counting classes}
Similarly to what we did in Section~\ref{sec: triad diagram counting}, we first establish that a bound on the number of triad equivalence classes can be derived from a bound on diagram equivalence classes. 

\begin{lemma}  \label{lem: from combi prob diagrams to triads}
Let $k\ge 0$.
If the bound $\widehat N(x,k, w) \leq C_k^{w}$ holds for all $0 \le k'\leq k$, and all $x,w$, then the bound
\[
    \widehat N_{\caT}(x,k,w) \;\leq\; w^8C_k^{w}  
\]
holds as well for all $x,w$. 
\end{lemma}

 \begin{proof}
 The proof is very similar to the proof of Lemma \ref{lem: combi from diagrams to triads} in Section \ref{sec: triad diagram counting}. 
A triad $ t\in {\caT}^{(k)}$ is a triple, consisting of a central diagram $\mathsf c(t)$, with parameters $(x_{\mathsf c},k,w_{\mathsf c})$, and  two other diagrams, with parameters $(x_{\mathsf l},k_{\mathsf l},w_{\mathsf l}),(x_{\mathsf r},k_{\mathsf r},w_{\mathsf r})$. 
To fix ideas, we assume that neither $\mathsf{l}(t)$ nor $\mathsf{r}(t)$ is empty.  
The proof can be straightforwardly adapted when one or both of them are empty, by replacing $\widehat N(x_{\mathsf{l}}, k_{\mathsf{l}}, w_{\mathsf{l}})$ or $\widehat N(x_{\mathsf{r}}, k_{\mathsf{r}}, w_{\mathsf{r}})$, respectively, by $1$ in \eqref{eq: combi triads bis} below.
The definitions of $\widehat N_\caT(.,.,.)$ and $\widehat N(.,.,.)$ yield directly
    \begin{equation}\label{eq: combi triads bis}
    \widehat N_{\caT}(x,k,w) 
    \;\leq\; 
    \sum_{\substack{w_{\mathsf c}+w_{\mathsf l}+w_{\mathsf r}=w \\[1mm]  k_{\mathsf l},k_{\mathsf r} <k \\[1mm]  |x_i-x|\leq w, i=\mathsf c,{\mathsf l},{\mathsf r}   }}
    \widehat N(x_{\mathsf c},k,w_{\mathsf c}) \widehat N(x_{\mathsf l},k_{\mathsf l},w_{\mathsf l}) \widehat N(x_{\mathsf r},k_{\mathsf r},w_{\mathsf r}).   
    \end{equation}
    Note that $w_{\mathsf c}+w_{\mathsf l}+w_{\mathsf r}=w$ holds because $\mathsf c(t)$ is non-crowded, by the definition of $\widehat N_\mathcal T (x,k,w)$. 
The restriction on the $x$-coordinates originates from the fact that $w=|t|$ is an upper bound for the support $I(t)$ of triad $t$.  Since also $k\leq |t|=w$, we see that the number of possible values of each of the eight parameters in the sum (three $x$-and $w$-parameters and two $k$-parameters) are bounded by $w$.  
Therefore, \eqref{eq: combi triads bis} is bounded by $w^8C^w_k w^{8}$
where we also used that $k\mapsto C_k$ is non-decreasing.
 \end{proof}

We can also obtain a bound on the number of triad equivalence classes, without the requirement that $\mathsf{c}(t)$ be non-crowded, as in the definition of $\widehat N(x, k, w)$, and without any explicit constraint on the order $w$.
Assuming that Proposition~\ref{prop: combinatoric for type two resonances} holds, we derive 
\begin{corollary}\label{cor: number of triad classes}
There exists a constant $C$ such that for every $x,k$, 
\[
    \left|\{[t] \, |\, t \in \mathcal T^{(k)}:\, \min I(t)=x\} \right|
    \; \le \;
    C^{L_k}.
\]
\end{corollary}

\begin{proof}  
Remember that $\beta L_k \leq |t| \leq 3L_{k+1}$ for any $t\in\caT^{(k)}$, cf.~\eqref{eq: bound on extended support triads}.
Let us first fix $w\le 3L_{k+1}$ and add the constraint $|t|=w$.
The proof parallels that of Lemma~\ref{lem: from combi prob diagrams to triads}, but we may no longer assume that $\mathsf c(t)$ is non-crowded.
Instead, we can use the bounds $w_{\mathsf l},w_{\mathsf r},w_{\mathsf c} \leq L_{k+1}$ as well as $w\le 3 L_{k+1}$. 
This yields the bound $(3L_{k+1})^8 C^{3L_{k+1}}$. 
Finally, summing over the at most $3L_{k+1}$ possible values of $w$ yields the claim. 
%
\end{proof}

\begin{remark}
    Corollary \ref{cor: number of triad classes} appears to be quite different from the bound on $\widehat N_\caT$ given above.
    We can afford this crude bound because Corollary \ref{cor: number of triad classes} is not involved in any inductive argument. 
\end{remark}

\subsection{Preliminary Bounds for Counting Diagram Classes}

Let $k\ge 0$ and let a diagram $g \in \mathcal G^{(k+1)}$ be given for this whole subsection.
We have the decomposition in diagrams/triads colleagues at scale $k$: $g=(t_0, t_1,\ldots,t_n)$ where we recall that $t_0$ is a diagram. We recall the terminology introduced in subsection \ref{sec: equivalence of diagrams}: A subset of the $t_i$ are called relevant colleagues, and we let them be indexed by $\alpha=0,1,\ldots, m$ as  $t_{i_{\alpha}}$, with $0=i_0<i_1<\ldots< i_m\leq n$, as introduced in Section \ref{sec: equivalence of diagrams}. The index $\alpha$ is used for the sake of recognizability, and only within the present section.  Note that there are $n-m$ irrelevant colleagues.

The following lemma is the crucial insight explaining the fact that the number of equivalence classes is fundamentally smaller than the number of diagrams. Indeed,  the number of diagrams $N(k,x,w)$ in Section \ref{sec: counting diagrams} was bounded by $C_k^w$ upon discounting each diagram with a factor $1/g!$, whereas the number of equivalences classes $\widehat  N(k,x,w)$ is bounded by  $C_k^w$ without any such discounting.

The reason is that, when constructing the diagram at scale $k+1$, each relevant triad must be placed either to the left or to the right of the union of supports of the triads with lower indices. This follows from the fact that a relevant triad is non-fully overlapping with its colleagues.
We note, however, that the diagram $t_0$ is always relevant by definition, and can thus be fully overlapping with its colleagues. This explains why the conclusion of the lemma below applies only for $\alpha \ge 2$: 
\begin{lemma}\label{lem: order of relevant triads}
Let $x_\alpha=\min I(t_{i_\alpha})$ for $\alpha=0,\ldots,m$.
For any $2 \leq \alpha\leq m$, one of the two following holds: 
either
$$
x_\alpha < x_{\alpha'}, \qquad \forall \, 1\leq \alpha'<\alpha,
$$
or 
$$
x_\alpha > x_{\alpha'}, \qquad \forall \,   1\leq \alpha'<\alpha,
$$
In the first case, we will refer to the index $\alpha$ as a left-extender, in the second case we call it a right-extender. 

\end{lemma}
\begin{proof}
Consider the process of constructing a diagram by attaching triads $t_i, i=1,\ldots,n$ to the diagram $t_0$.  At every step $i$, the set
\[
    I_i \;=  \; I(t_0) \cup I(t_1)\ldots\cup I(t_{i-1})
\]
is an interval, as follows from the definition of adjacency of colleagues.  If the triad $t_i$ is relevant, then in particular its central diagram $\mathsf c(t_i)$, considered as a descendant of $g$, is non-fully overlapping. Therefore, $I(t_i)\setminus I_i$ cannot be be empty, so $I(t_i)$ has to stick out to the left, or to the right,  or both to the left and to the right, of $I_i$.  The claim now follows by Lemma~\ref{lem: triadsupport inherits order} in Section~\ref{sec: more about non-overlapping}.
\end{proof}

In the next part, up to Lemma \ref{lem: auxiliary for combi of proba}, we simply write  $t$ for colleagues from the set $\{t_0,\ldots,t_n\}$, and we refer to the dichotomy relevant versus irrelevant.  We recall the definition of $O(g)$ from Section~\ref{sec: equivalence of diagrams}:
\[
O(g) \;= \;   
\left(\mathop{\bigcup}\limits_{t \, \text{irrelevant}}I(t) \right)        
\cap   
\left(\mathop{\bigcup}\limits_{t \, \text{relevant}} \overline{I}(t) \right) 
\]
and we also introduce another set, $E(g)$, defined as
\[
E(g) \; = \;   
\mathop{\bigcup}\limits_{t \, \text{irrelevant}}I(t) 
\; = \;  \left( I(g) \setminus  \mathop{\bigcup}\limits_{t \, \text{relevant}} \overline{I}(t) \right)  \cup O(g).
\]
The last equality shows that $E(\cdot)$ is also a class function. 
Moreover, since $O(g) = E(g)\cap (\cup_{t\text{ relevant}} \overline I(t))$ and since $\overline{I}(t)$ is a class function for relevant colleagues $t$, we conclude that we may equivalently define the class by specifying $E(g)$ instead of $O(g)$.

Finally, let us list some  properties that we will use in the upcoming subsection.
\begin{lemma}\label{lem: auxiliary for combi of proba} Let $x=\min I(g)$. 
\begin{enumerate}
\item $I(g) \subset [x,x+|g|-1]$. 
    \item 
    $ I(g)= E(g)\cup (\cup_{\alpha=1}^m  I(t_{i_\alpha}))$.
    \item  $ |E(g)|  \leq  |g|-\sum_{\alpha=0}^m |t_{i_\alpha}|$.
    \item The set $\mathcal A(g) \setminus E(g)$ is determined by the relevant colleagues. More precisely, $x$ belongs to $\mathcal A(g) \setminus E(g)$ whenever it belongs to ${\mathcal A}(t_{i_\alpha})$ for an odd number of $t_{i_\alpha}$. 
    \item $n\leq n_*= \frac{|g|}{\beta L_k}$.
    \item The set $E(g)$ is a union of at most $n-m$ disjoint intervals $E_j$. 
\end{enumerate}
\end{lemma}

\begin{proof}
Items 1, 2 and 4 are obvious.
Item 3 follows from
\[
|g|\; = \;
\sum_{i=0 }^n |t_i| \;\geq\;   
\sum_{\text{relevant} \, t } |t| + \sum_{\text{irrelevant} \, t}  |I(t)| 
\;\geq\; \sum_{\alpha=0}^m|t_{i_\alpha}| +|E(g)|.
\] 
Item 5 follows from the fact that $|g|=\sum_{i=0}^n |t_i|$, and $|t_i|\geq \beta L_k$.
Finally, for Item 6, we note that $E(g)$ is a union of $n-m$  intervals $I(t)$, with $t$ running over irrelevant colleagues. These intervals $I(t)$ are not necessarily disjoint, but the claim follows by considering the connected components $E_j$ of $E(g)$.  
\end{proof}

\subsection{Proof of Proposition \ref{prop: combinatoric for type two resonances}}

Thanks to Lemma~\ref{lem: from combi prob diagrams to triads}, we only have to show the first bound in \eqref{eq: thing to prove counting classes}.
The proof is by induction on $k\ge 0$. 
For $k=0$, no two distinct diagrams are equivalent and $g!=1$ always. Therefore, the proof proceeds just as in Section~\ref{sec: running constant}. 
Next, we assume that the proposition is proven up to scale $ k$ and we consider  equivalence classes $[g]$ at scale $k+1$.
We will perform the sum over equivalence classes 
\[
 \widehat  N(x,k+1, w)
 \; = \;  
 \sum_{[g]:  |g|=w, \min I(g)=x,g\in\mathcal G^{(k+1)}} 1
\]
step by step, as outlined here in order of execution:
\begin{enumerate}
    \item the increasing $m$-tuple of indices $1\leq i_1 < \ldots <i_m\leq n$ for given $m,n$, 
    \item the set of active spins $\mathcal A$ for fixed classes $[t_{i_\alpha}]$,  $\alpha=0,\ldots,m$, 
    \item the disjoint intervals $E_j$,
    \item the triad classes $[t_{i_\alpha}]$ for fixed $x_\alpha,w_\alpha$,  $\alpha=1,\ldots,m$,
    \item the parameters $(x_{\alpha}), \alpha=1,\ldots,m$,
    \item the equivalence classes $[t_0]$ for fixed $x_0,w_0$, 
    \item the parameter $x_0$ of $t_0$,
    \item the parameters $w_{\alpha}$, $\alpha=0,\ldots,m$,
    \item the numbers $m, n$.
\end{enumerate}
At each of these steps, we keep the summation variables of later steps fixed, as well as the parameters $x,k,w$, so that all the sums are highly constrained. 
Hence, we will successively perform the following sums from right to left:
\begin{equation}\label{eq: successive sums} 
\sum_{0\leq m \leq n\leq n_*}    \sum_{\substack{(w_\alpha) \\ \alpha=0,\ldots,m}} \sum_{x_0}
\sum_{[t_0]}\sum_{\substack{(x_\alpha) \\ \alpha=1,\ldots,m}}  \sum_{\substack{[t_\alpha] \\ \alpha=1,\ldots,m}}  \sum_{(E_j)}  \sum_{\mathcal A} \sum_{1\leq i_1 < \ldots <i_m\leq n}  1.
\end{equation}
We now describe the result of the rightmost seven sums.
\begin{enumerate}
    \item  The  sum over the indices $1\leq i_1,\ldots, i_m \leq n$ is bounded by the number of ways to choose $m$ increasing numbers between $1$ and $n$:
\[ 
    \frac{n^m}{m!} \;\leq\; e^n.
\]

\item  The number of possibilities for $\mathcal A$, given the classes $[t_{i_\alpha}]$ for $\alpha=0,\ldots,m$, is bounded by
\begin{equation*}
   2^{|E|} \;\leq\;  
   2^{|g|-\sum_{\alpha=0}^m |t_{i_\alpha}|}
   \;=\; 2^{|g|-\sum_{\alpha=0}^m w_\alpha}.
\end{equation*}
Indeed, on  $I(g)\setminus E(g)$, the set $\mathcal A$ is fully determined by the relevant colleagues $t_{i_\alpha}$, cf.\ item 4 of Lemma \ref{lem: auxiliary for combi of proba}.

\item  The disjoint intervals $(E_j)_{j=1,\ldots, n_E}$ are subsets of $[x,x+|g|-1]$, with $n_E\leq n-m$.  Their choice is specified by choosing their minima and maxima. Hence, the number of possibilities is bounded by the number of ways of choosing at most $2n_E$ distinct elements of $[x,x+|g|-1]$, with $n_E\leq n$.   
This is hence bounded by 
\begin{equation}\label{eq: 1st sum of factorials like}
 \sum_{0\leq n_E\leq n}\frac{|g|^{2n_E}}{(2n_E)!}  .
\end{equation}

\item The sum over classes $[t_\alpha], \alpha=1,\ldots,m$ with parameters $(x_\alpha,w_\alpha)$ fixed, yields of course 
\[
\prod_{\alpha=1}^m  N_{\caT}(x_\alpha,k,w_\alpha) 
\;\leq\; 
C_k^{\sum_{\alpha=1}^m w_\alpha},
\]
by the induction assumption. 
\item 
The parameters $x_0,x_1$ are constrained to lie in $I(g) \subset [x,x+|g|-1]$, hence the number of possibilities for these coordinates is bounded by $|g|^2$. For $x_{\alpha}, \alpha=2,\ldots,m$, we use Lemma \ref{lem: order of relevant triads}. Either of these indices is either a left-extender or a right-extender and we first fix the partition of $\{2,\ldots,m\}$ into right-and left-extenders. Let then $\alpha_j, j=1,\ldots, n_r$ with $n_r\leq m$ and $j\mapsto\alpha_j$ increasing, be the right-extenders. Then the number of possibilities for $x_{\alpha_j}, j=1,\ldots, n_r$ is bounded by 
$$
\frac{|g|^{n_r}}{n_r!}.
$$
In the same way we bound the number of choices for left-extenders, with $n_l<m$ the number of left-extenders.
 Finally, we note there are at most $2^{m-1}$ ways to partition $\{2,\ldots,m\}$ into left-and right-extenders. Hence we get overall the bound
\begin{equation}\label{eq: 2nd sum of factorials like}
    |g|^2  2^n  \left(\sup_{1\leq n'\leq n}  \frac{|g|^{n'}}{n'!}\right)^2,
\end{equation}
where we also used $m\leq n$.

\item The sum over the classes $[t_0]$ with parameters fixed, is again given by the induction hypothesis. It is 
$$
C_k^{w_0}
$$
\item The sum over $x_0$ is bounded by $|g|$ since it has to lie in $[x,x+|g|-1]$. 
\end{enumerate}

As already observed in Section~\ref{subsubsec: induction step}, 
given $A>0$,
\begin{equation}\label{eq: bounding factorials by their max}
    \max_{q = 0,1,\dots, p} \frac{A^q}{q!}
    \; \le \; 
    \max_{q = 0,1,\dots, p} \left(\frac{eA}{q}\right)^q
    \; = \; 
    \left(\frac{eA}{p}\right)^p
\end{equation}
as long as $p\le A$.
Collecting the above bounds, we write $w = |g|$ and apply the estimate~\eqref{eq: bounding factorials by their max} with $p = n_*$ to bound~\eqref{eq: 1st sum of factorials like} and~\eqref{eq: 2nd sum of factorials like}.  
Since $n_* \le w$, this is valid, and we conclude that~\eqref{eq: successive sums} is bounded by
\begin{equation*}
\sum_{0\leq m \leq n\leq n_*}  \sum_{\substack{(w_\alpha) \\ \alpha=0,\ldots,m}}    \, w \,\times \, C_k^{w_0} \,\times 
w^2  2^n  \left(\frac{ew}{n_*}\right)^{2n_*} 
\times C_k^{\sum_{\alpha=1}^m w_\alpha}  \times  
n \left(\frac{ew}{n_*}\right)^{2n_*} 
\times  2^{w-\sum_{\alpha=0}^m w_\alpha} \times  e^n .
\end{equation*}
Since $C_0 = 2$, as defined in Section~\ref{sec: running constant},
we have $C_k \ge 2$, and we can bound this sum by 
\begin{equation*}
C_k^w w^4 (2e)^{n_*}  \left(\frac{ew}{n_*}\right)^{4n_*} \sum_{0\leq m \leq n\leq n_*}    \sum_{\substack{(w_\alpha) \\ \alpha=0,\ldots,m}}  1.
\end{equation*}
where we have used the bound $n\le w$ again. 
To sum over the parameters $(w_\alpha)$, $\alpha=1,\ldots,m$, we recall that $\beta L_k \leq w_\alpha \leq 3L_{k+1}$, cf.~\eqref{eq: bound on order triads}. 
Hence the number of possibilities for these coordinates is bounded by $(3L_{k+1})^m \leq (3L_{k+1})^n $. 
The number of possibilities for $w_0$ is simply bounded by $|g|=w$.
Finally the number of possibilities for $m,n$ is bounded by $n_*^2 \le w^2$.
This yields thus the overall bound 
\[
    C_k^w w^7 (6e)^{n_*} L_{k+1}^{n_*} \left(\frac{ew}{n_*}\right)^{4n_*}.
\]
We conclude as in Section~\ref{subsubsec: induction step}.

\section{Probability of Absence of Resonances} \label{sec: probability of absence of resonances}

In section \ref{subsec: main results on inductive bounds}, we defined the non-resonance events
$\nonres_{\mathrm I}(t)$ and $\nonres_{\mathrm{II}}(t)$ for any triad $t$. 
Recall that these events depend only on the disorder variables 
$\{\theta_{x} \, : \,  x \in  I(\triad) \}$.
The goal of this section is to provide a lower bound on the probability that all non-resonance conditions in the chain are satisfied, i.e.\ that all non-resonance events happen: 
\begin{proposition}\label{thm: total prob of resonances}
There exists a constant $c>0$ such that, 
\[
\prob\left(\bigcap_{k \in \mathbb{N}} \bigcap_{\triad\in {\caT}^{(k)} }  (\nonres_{\mathrm I}(t) \cap \nonres_{\mathrm {II}}(t)) \right)  \;\geq\; e^{-\varepsilon^c L} 
\]
\end{proposition}
By the monotone convergence theorem, it suffices to prove the above bound with $k<k_{\max}$, uniformly in $k_{\max}$. We will do this to avoid subtleties, at the cost of having the artificial parameter $k_{\max}$.

\subsection{Dressed and Aggregated Resonance Events}\label{sec: dressed and aggregated resonance events}

To streamline what follows, we introduce a notation for \emph{resonance} events, that are defined as the complement of a \emph{non-resonance} event:
$$
\res(t)= (\nonres_I(t) \cap \nonres_{II}(t))^c.
$$
We will need to aggregate resonances corresponding to a given region, and we prepare the ground for this now. Let $\caJ$ be the set of intervals $S$ that intersect $\Lambda_L$ and whose cardinality equals $\lceil 5L_{k+1}\rceil$ for some $0 \le k \leq k_{\max}$.  For any $S\in \caJ$, we write $k(S)$ for the unique $k$ such that $|S|=\lceil 5L_{k+1}\rceil$ and we say that the scale of $S$ is $k$. 
\begin{remark}
 Note that an interval $S\in \caJ$ is not necessarily a subset of $\Lambda_L$, it only has to intersect $\Lambda_L$. This is done to avoid separate treatment of sets $S$ at the boundary of $\Lambda_L$.  
\end{remark}
Next, we associate to any triad $\triad \in {\caT}^{(k)}$ a unique interval $S=S(\triad) \in \caJ$ with scale $k$ by the following rule: $S(\triad)\setminus \overline{I}(t)=J_1\cup J_2$ where $J_1,J_2$ are disjoint intervals (possibly empty) such that $|J_1|-|J_2| \in \{0,1\}$. This is well-defined by the bound \eqref{eq: bound on extended support triads}. 
The specific choice of this rule is not important, except for the fact that it makes sure that $\overline{I}(t) \subset S(\triad)$.

Now to the definition of the \emph{aggregated resonance} events. 
For any $S\in \caJ$, we set 
\[
\fatres_S \;=\;  \bigcup_{\triad \in {\caT}^{(k)}:  S(\triad)=S  }   \res(t), \qquad  k=k(S)
\]
and we also define the \emph{dressed aggregated resonance} event if $k(S)>0$:
\[
\barres_S \;=\;  
\fatres_S \bigcap_{\substack{ S': k(S') <k(S)   \\  S'\cap S\neq \emptyset}}   \fatres_{S'}^c .
\]
We set $\barres_S = \fatres_S$ if $k(S)=0$.
The following lemma expresses that locally a dressed aggregated resonance is unlikely: 
\begin{lemma}\label{lem: bound on barres} 
Let $S \in \caJ$ and let $k=k(S)$. 
There exists $c>0$ such that 
   \begin{equation}\label{eq: bound on barres}
  \prob ( \barres_{S}) \;\leq\; \varepsilon^{cL_k}. 
\end{equation}  
\end{lemma}

\begin{proof}
We recall the event $\fatnonres([t])$ introduced in Section \ref{subsec: main results on inductive bounds}: for any triad $t\in\mathcal T^{(k)}$, 
\[
\fatnonres([t]) \; = \;  \bigcap_{k'<k} \,\mathop{\bigcap}\limits_{\substack{t' \in {\caT}^{(k')}  \\[1mm]  I(\triad') \subset I(t) } } \nonres_{\mathrm{I}}(\triad')\cap \nonres_{\mathrm{II}}(t').
\]
We  check that 
\begin{equation}\label{eq: barres included in previous nonres}
    \barres_{S} \;\subset\;   
    \bigcup\limits_{\substack{[t] \in[\mathcal T^{(k)}]: \\ S(t)=S }} 
     \left( 
    \Big(\bigcup_{t' \in [t] }  
    \res(t') \Big) \;\bigcap\;  {\fatnonres}([t]) \right), 
\end{equation}
with $[\mathcal T^{(k)}]=\{[t]:t\in\mathcal T^{(k)}\}$, 
i.e.\@ the set of all equivalence classes of triads in $\mathcal T^{(k)}$. 
The expression between large round brackets is the event whose probability was estimated in Proposition \ref{prop: prob of resonances}, except for the fact that there the resonance event $\res(t')$ was replaced either by $(\nonres_{\mathrm{I}}(t'))^c$ or by $(\nonres_{\mathrm{II}}(t'))^c$, whereas we now consider the union.  Therefore, by applying Proposition \ref{prop: prob of resonances}, we get
\[
    \prob  (\barres_{S})    
    \; \leq \; 
    \sum\limits_{\substack{[t] \in[\mathcal T^{(k)}]: \\ S(t)=S }} \Const \varepsilon^{c|t|}  
    \; \leq \; \Const \varepsilon^{c\beta L_{k}} 
    \sum_{\substack{[t] \in[\mathcal T^{(k)}]: \\  \min(I(t)) \in S}}  1
\]
where the second equality follows from the lower bound $|t|\ge \beta L_k$ in \eqref{eq: bound on order triads}. 
To estimate the remaining sum on the right-hand side, we use Corollary \ref{cor: number of triad classes} in Section~\ref{sec: from diagrams to triads counting classes}, which yields
\[
    \sum_{\substack{[t] \in[\mathcal T^{(k)}]: \\  \min(I(t)) \in S}}  1 
    \;\leq\; |S|\Const^{L_k} \;\leq\;  \lceil 5L_{k+1} \rceil \Const^{L_k} \;\leq\; \Const^{L_k},
\]
where we updated the value of $\Const$ to get the last inequality.  
The claim of the lemma now follows for small enough $\varepsilon$.
\end{proof}

The above lemma will be applied in the following, more general form: given mutually disjoint sets $S_1,\ldots, S_n \in \mathcal J$, we have 
\begin{equation}\label{eq: bound on dependent products}
  \prob \left( \bigcap_{i=1}^n\barres_{S_i} \right) 
  \;\leq\; \prod_{i=1}^n\varepsilon^{cL_{k(S_i)}},
\end{equation}
despite the fact that these events are not independent. This bound follows upon using \eqref{eq: barres included in previous nonres} to dominate the events $\barres_{S_i}$ by events that are independent. The rest of the proof proceeds then analogously to the proof of \eqref{eq: bound on barres}. 

Finally, it is important to realize that the event of ``absence of all resonances" coincides with the event of ``absence of all dressed resonances", and this is formalized now. Here $\chi(E)$ is the indicator of an event $E \subset \Omega$ and $E^c$ is the complement $\Omega\setminus E$:
\begin{lemma}\label{lem: dressed events replace events}
    \begin{equation}\label{eq: fancified product}
   \prod_{S \in \caJ}  (1- \chi(\fatres_S))
   \; = \; \prod_{S \in \caJ}  (1- \chi(\barres_S))    
\end{equation}
\end{lemma}
\begin{proof}
For any $S \in \caJ$,
\[
  (1- \chi(\fatres_S))  \prod_{\substack{S'\in \caJ\\ k(S')<k(S) }}\chi(\fatres^c_{S'}) 
  \; = \;    
  (1- \chi(\barres_S)) \prod_{\substack{S'\in \caJ\\ k(S')<k(S) }}\chi(\fatres^c_{S'}).
\]
We use this relation on the left-hand side of \eqref{eq: fancified product}
to replace $ (1- \chi(\fatres_S))$ by $ (1- \chi(\barres_S))$ for all $S$ with $k(S)=k_{\max}$. Then we do the same for all $S$ with $k(S)=k_{\max}-1$ and we proceed down to the lowest scale $k=0$.
\end{proof}

Thanks to the above lemma and the remark following Proposition~\ref{thm: total prob of resonances},  we see that it suffices to prove that there exists a constant $c>0$ so that 
\begin{equation}\label{eq: bound on partition function}
Z  \; = \;   {\mathbb{E}}\prod_{S\in \caJ} (1-\chi(\barres_S))   \;\geq\; e^{- \varepsilon^c L}
\end{equation}
uniformly in $k_{\max}$. Here, we introduced the notation $Z$ to stress the similarity with a statistical mechnics partition function.
In the next section we develop the tools to prove this bound.

\subsection{Cluster Expansion}

We define  a \emph{adjacency} relation  as follows: For $S,S' \in \caJ$,
\[
    S\;\sim\; S' 
    \quad \Leftrightarrow \quad 
    \mathrm{dist}({S}, {S'}) \;\leq\; 5(L_{k(S)}+L_{k(S')}).
\]
More generally, we say that two collections  $\caS_1,\caS_2 \subset \caJ$ are adjacent (notation: $\caS_1 \sim \caS_2$) whenever there are $S_1\in\caS_1,S_2\in\caS_2$ such that $S_1 \sim S_2$.
The rationale for this intricate definition is that it ensures that the events $\barres_{S}, \barres_{S'}$ are independent whenever $ S\nsim S'$. Indeed, the event $\barres_{S}$ does not depend on $\theta_x$ if $\mathrm{dist}(S,x)>5L_{k(S)}$.

We introduce \emph{polymers} $\caS$ as collections of sets $S$ that are connected for the adjacency relation $\sim$. That is, a collection $\caS \subset \caJ$ is a polymer if and only if, for any partition of $\caS$ into two collections $\caS_1,\caS_2$, it holds that $\caS_1\sim \caS_2$. We denote the set of polymers as $\mathbb{S}$. 
By writing the product over $S$ in \eqref{eq: bound on partition function} as a sum over subsets of $\caJ$, we derive
\begin{equation}
    \label{eq: polymer expansion}
\expect \prod_{S \in \caJ} (1-\chi(\barres_S))
\; = \;  
1+\sum_{m=1}^{\infty} \frac{1}{m!}  \sum_{(\mathcal S_1 \ldots \mathcal S_m) \in \mathbb{S}^m }  \prod_{i=1}^m w(\mathcal S_i)  \left(\prod_{1\leq i <j \leq m}  \chi(\caS_i \nsim \caS_j) \right), 
\end{equation}
where the weight $w(\mathcal S)$ is defined as
\[
w(\mathcal S) \;=\; (-1)^{|\mathcal S|} \bbE \left( \prod_{S \in \mathcal S}   \chi(\barres_S)\right).
\]
The derivation of \eqref{eq: polymer expansion} follows from the independence of $\barres_{S}, \barres_{S'}$  whenever $ S\nsim S'$.

Let the set of clusters $\mathbb{K}$ consist of collections of polymers $\mathcal K=\{\caS_1,\ldots, \caS_m\}$ such that this collection is connected for the adjacency relation $\sim$. 
The basic result of cluster expansions, originally derived in \cite{kotecky1986cluster}, is
\begin{theorem}\label{thm: koteckypreiss}[Theorem 1 in \cite{ueltschi_cluster}]
Let $a: \mathbb{S}\to\mathbb{R}$ be a positive function on polymers such that, for any $\caS_0 \in \mathbb{S}$,
\begin{equation}\label{eq: kotecky preiss criterion}
\sum_{\caS: \caS \sim \caS_0}  |w(\caS) |e^{a(\caS)} 
\;\leq\; a(\caS_0).
\end{equation}
Then $Z > 0$ as defined in \eqref{eq: bound on partition function}, and there is a function $w^T:\mathbb{K} \to \mathbb{R}$ (called ``truncated weight'') such that  
\[
\log Z \;=\; \sum_{\mathcal K \in  \mathbb{K}}  w^T(\mathcal K)
\]
with 
$w^T(\cdot) $ satisfying, for any $\caS\in \mathbb{S}$,
\begin{equation}\label{eq: conclusion of kotecky preiss}
    \sum_{\mathcal K \sim \caS  } |w^T(\mathcal K)|   \;\leq\;  a(\caS).
\end{equation}
\end{theorem} 

\begin{remark}
Let us relate the notation used here to that in \cite{ueltschi_cluster}.
In \cite{ueltschi_cluster}, polymers are denoted by $A \in \mathbb{A}$, which correspond to our notation $\mathcal{S} \in \mathbb{S}$.
The measure $\mu$ in \cite{ueltschi_cluster} is, in our setting, the measure that assigns to each polymer $\mathcal{S} \in \mathbb{S}$ the weight $w(\mathcal{S})$.
The function $\zeta(A, A')$ in \cite{ueltschi_cluster} implements the adjacency relation $\sim$ in our setting, and is defined by $\zeta(A, A') = -1$ whenever $A \sim A'$, and $\zeta(A, A') = 0$ otherwise.
Finally, the bound \eqref{eq: conclusion of kotecky preiss} is equation (5) in \cite{ueltschi_cluster}.
\end{remark}

Let us explain how we will apply this result.  
Let $\supp(\caS)$ be the smallest interval containing all $S\in \caS $. Then we have
\begin{proposition}\label{prop: properties polymer weights}
There is a $c>0$ such that, for $\varepsilon$ small enough, the condition \eqref{eq: kotecky preiss criterion} holds with
\[
    a(\caS)\; = \; \varepsilon^c |\supp(\caS)|.
\]
\end{proposition}
This is proven in the next subsection. For now, we use it to complete the 
\begin{proof}[Proof of Proposition \ref{thm: total prob of resonances}]
As we already explained, it suffices to prove the bound \eqref{eq: bound on partition function}. 
Let us take a polymer $\caS_0$ whose support is $\Lambda_L$.  Since any $S\subset \caJ$ intersects $\Lambda_L$, it follows that any cluster $\mathcal K \in \mathbb{K}$ has to satisfy $\mathcal K \sim \caS_0$.
Therefore
\[
|\log Z| \;\leq\;  
\sum_{\mathcal K \in \mathbb{K}} |w^T(\mathcal K)|  
\;=\;  
\sum_{ \mathcal K \sim \caS_0} |w^T(\mathcal K)|  
\;\leq\;
a(\caS_0)  
\;=\; 
\varepsilon^c L
\]
which yields the bound in proposition \ref{thm: total prob of resonances}, 
provided Proposition \ref{prop: properties polymer weights} holds.
\end{proof}

\subsection{Proof of Proposition \ref{prop: properties polymer weights}} \label{sec: proof of kotecky preiss}

Let us show that there exists constants $c,c'>0$ such that, for $\varepsilon$ small enough,
\begin{equation}\label{eq: what actually need}
  \sum_{\caS: \supp(\caS)=I} |w(\caS) |  
  \;\leq\; \varepsilon^{2c}  e^{- c'|I|},   
\end{equation}
for any finite interval $I \subset \mathbb{Z}$. 
This will imply Proposition~\ref{prop: properties polymer weights} with the same constant $c$. 
Throughout this section, we fix hence an interval $I$ and we consider polymers $\caS$ such that $I=\supp (\caS)$. 


The next three lemmas provide the tools to bound the number of polymers $\mathcal{S}$ such that $I = \mathrm{supp}(\mathcal{S})$.
The crucial observation is the following: If a polymer $\caS$ contains the intervals $S,S' \in \caJ$ with $k(S')<k(S)$, and 
\[
    S\cap S' \ne \varnothing ,
\]
then $w(\caS)=0$. This is a direct consequence of the definition of the event $\barres_S$. 
In the remainder of this section, we will restrict our attention to polymers $\caS$ such that $w(\caS)\neq 0$, which therefore places stringent restrictions on them.

Given $0\le k \le k_{\mathrm{max}}$,
and given a polymer $\mathcal S$, let us introduce the set
\[
   N_k(\caS) \;=\; \bigcup_{S\in \caS, k(S)=k} S.
\] 
Then the above observation implies directly
\begin{lemma}\label{lem: n sets disjoint}
If $k\neq k'$ and if $w(\mathcal S)\ne 0$, 
the sets $N_k(\caS)$ and $N_{k'}(\caS)$ are disjoint. 
\end{lemma}

Next, let us count polymers $\caS$ with the sets $N_k(\caS)$ kept fixed:
given a collection $(N_k)_{0\leq k\leq k_{\max}}$  of disjoint subsets of $I$, we have
\begin{lemma}\label{lem: number of polymers via collections}
The number of polymers $\caS$ with $\supp(\mathcal S) = I$, with weight $w(\caS) \neq 0$, and such that   $ N_k(\caS)=N_k$ for all $0\le k\le k_{\mathrm{max}}$, is bounded by 
\begin{equation}\label{eq: number of polymers per collection}
   2^{\sum_{k=0}^{k_{\max}} |N_k|} \; \leq  \; 2^{ |I|}. 
\end{equation}  
\end{lemma}
\begin{proof}
On the one hand, at every point in $N_k$, at most one set $S$ can be started (i.e.\@ can have this point as its minimum), namely a set $S$ at scale $k$. On the other hand, for any interval $S\subset\caS$ with $k(S)=k$, we have $ S \subset N_k$. This means that every polymer can be specified by a binary variable for each element of $\bigcup_k N_k$, and hence the number of different polymers corresponding to a given collection $({N_k})_{0\leq k\leq k_{\max}}$ is bounded by \eqref{eq: number of polymers per collection}.
\end{proof}

Let \((N_k)_{0 \leq k \leq k_{\max}}\) be as above, and denote the connected components of each \(N_k\) by \(K_{k,i}\), where \(i\) ranges over a finite index set that may depend on \(k\).
We also include the connected components of the set
\begin{equation}\label{eq: empty space}
E \; = \; I \setminus \bigcup_k N_k,
\end{equation}
which we refer to as the ``empty space''. To these components, we assign the dummy scale \(k = -1\), i.e.\@ we denote them by \(K_{-1,i}\), for \(i\) in a finite set.

To simplify notation, we combine the indices \((k, i)\) into a single label \(\alpha\), and write \(k(\alpha)\) to denote the scale associated with \(\alpha\).
In this way, we obtain a \emph{decorated partition} of \(I\) as
\begin{equation}\label{eq: I partition alpha}
I = \bigcup_{\alpha} K_\alpha,
\end{equation}
where each interval \(K_\alpha\) carries a scale \(k(\alpha) \in \{-1, 0, \dots, k_{\max}\}\).



We can now count the number of such decorated partitions of the interval $I$:
\begin{lemma}\label{lem: number decorated partititons}
    There exists a constant $C$ such that the number of decorated partitions of $I$ is bounded by $C^{|I|}$.
\end{lemma}
\begin{proof}
Let us first find the number of partitions of $I$ into intervals (without decoration). 
It suffices to count the possible minima of the intervals. 
Hence we find that the number of partitions is bounded by $2^{|I|}$. 

Next, given such a partition into intervals, we bound the number of possible decorations. 
Let $\ell_\alpha$ be the cardinality of the interval $K_\alpha$. 
Let also $M \le |I|$ be the number of intervals $K_\alpha$ in the partition~\eqref{eq: I partition alpha}.
We notice that $\sum_\alpha \ell_\alpha = |I|$, and that $5L_{k(\alpha)}  \leq |K_\alpha|$ with the convention $L_{-1}=0$. 
There exists thus a constant $C$ such that 
\begin{equation*}
    k(\alpha) \;\leq\; C\log(\ell_\alpha).
\end{equation*}
Using this, we find that the number of possible decorations is bounded by
\[
    C^M \prod_{\alpha} \log (\ell_\alpha) \;\le\; C^{|I|} e^{|I|} 
\]
where we have used the crude bound $\log x < e^x$ to get the last bound.

The number of decorated partitions is thus bounded by $2^{|I|} C^{|I|} e^{|I|}$.
%
\end{proof}

The lemmas~\ref{lem: number of polymers via collections} and \ref{lem: number decorated partititons} imply that the number of non-zero terms in \eqref{eq: what actually need} is upper-bounded by $C^{|I|}$ for some constant $C$.
The next lemma provides a bound on the weights $w(\mathcal S)$, and allows us to conclude the proof of \eqref{eq: what actually need}.

\begin{lemma} 
There exists a constant $c$ such that, for any $\caS$ with $\supp(\caS)=I$, 
\[
|w(\caS)| \leq  \varepsilon^{c|I|}.
\]
\end{lemma}
\begin{proof}
We first bound the size of $E$ defined in \eqref{eq: empty space}.  
To each $x \in E$, we associate a parent component $K_{\alpha}$ with $k(\alpha) \geq 0$ and such that $\dist(K_{\alpha},x) \leq  5L_{k(\alpha)} $.  Such a parent component always exists (it need not be unique, but we choose one in an arbitrary way).  On the other hand, any $K_{\alpha}$ with $k(\alpha) \geq 0$ is the parent component for at most 
$10L_{k(\alpha)}\leq 2 |K_{\alpha}|$ elements of $E$.  
Therefore,
\[
    |E| \;\leq\; 2 |I\setminus E| \; = \; 
    2 \sum_{\alpha: k(\alpha)\geq 0} |K_\alpha|.
\]
Hence 
\begin{equation} \label{eq: bound by bigk sum}
    |I| \;\leq\; |I\setminus E| +|E| 
    \;\leq\; 
    3\sum_{\alpha: k(\alpha)\geq 0} |K_\alpha| .
\end{equation}
Then, we remark that for any $K_\alpha$ with $\alpha \geq 0$, there is a subcollection $\caS_\alpha \subset \caS$ such that
\begin{enumerate}
    \item 
    for any $S\in \caS_\alpha$,  $S\subset K_\alpha$, and in particular $k(S)=k(\alpha)$,
    \item all elements of $\caS_\alpha$ are disjoint,
    \item  $\sum_{S \in \caS_\alpha}|S | \geq |K_\alpha|/2$.
\end{enumerate}
Now we bound
\[
|w(\caS)| 
\;\leq\; 
\prod_{\alpha: k(\alpha)\geq 0} \prod_{S \in \caS_\alpha}   \varepsilon^{c|S|}  
\;\leq\; \prod_{\alpha: k(\alpha)\geq 0} \varepsilon^{c|K_\alpha|} 
\;\leq\;
\varepsilon^{c|I|}
\]
where the first inequality is by \eqref{eq: bound on dependent products} and the last inequality is by \eqref{eq: bound by bigk sum}, upon readjusting $c>0$. 
\end{proof}

\section{Proof of Main Theorems}\label{sec: proof of main theorems}
In this section, we prove the theorems stated in Section \ref{sec: model and main results}, using the technical results established earlier. The arguments here are standard and do not rely on the specific language or methods introduced in the previous sections.

Let us define the \emph{full non-resonance} event
\begin{equation}\label{eq: fully non resonant}
    \mathrm{FNR}
    \; = \;   
    \bigcap_{k \in \mathbb{N}} \bigcap_{\triad\in {\caT}^{(k)} }  \nonres_{\mathrm I}(\triad) 
    \cap 
    \nonres_{\mathrm{II}}(\triad).
\end{equation}
Using explicitly the relation $\varepsilon = \gamma^{a(\beta)}$ in \eqref{eq: delta epsilon power law gamma}, the bound in Proposition~\ref{thm: total prob of resonances}, Section \ref{sec: probability of absence of resonances}, implies that
\begin{equation}\label{eq: bound on FNR}
    \mathbb P(\mathrm{FNR}) \; \ge \; e^{-\gamma^{c}L}
\end{equation}
for some constant $c>0$ (that depends on $a(\beta)$).

In Sections~\ref{sec: bounds on generators} and \ref{sec: proof of theorem one}, we will show that the conclusions of Theorem~\ref{thm: locality of u} hold on this event, and we will establish Theorem~\ref{thm: absence of conduction} as a corollary in Section~\ref{sec: absence heat conduction}. 



\subsection{Constructing the Diagonalizing Unitary}\label{sec: bounds on generators}

The statement of Theorem \ref{thm: locality of u}, involves a diagonalizing unitary $U$, which we construct now as an infinite product
\begin{equation}\label{eq: explicit diagonalising unitary}
    U  \;=\; \lim_{k\to \infty} U^{(k)}, 
    \qquad  
    U^{(k)}= e^{-A^{(k)}}\dots e^{-A^{(1)}}
\end{equation}
on the event $\mathrm{FNR}$ introduced in \eqref{eq: fully non resonant}.

We first check that this limit is well-defined on $\mathrm{FNR}$.
From here onward, it is convenient to decompose $V^{(k)}$ and $A^{(k)}$ as a sum of local terms: 
\begin{align*}
    &  V^{(k)} \; = \; 
    \sum_{\substack{I\subset\Lambda_L,\\ \text{interval}}}
    V^{(k)}_I, 
    \qquad 
    V^{(k)}_I 
    \; = \; 
    \sum_{\substack{g\in\mathcal G^{(k)}:\\ {I}(g)=I}} V^{(k)}(g),\\
    & A^{(k+1)} \; = \; 
    \sum_{\substack{I\subset\Lambda_L,\\ \text{interval}}}
    A^{(k+1)}_I,
    \qquad
    A^{(k+1)}_I 
    \; = \; 
    \sum_{\substack{t\in\mathcal T^{(k)}:\\ {I}(t)=I}} A^{(k+1)}(t)
\end{align*}
for any $k\ge 0$ (note that the operators $V^{(k)}_I$ and $A^{(k)}_I$ are supported on $\overline I = [\min I - 1, \max I +1 ]\cap \Lambda_L$).
Recalling that $\gamma/\varepsilon\delta < 1$ by \eqref{eq: implied constraint gamma delta epsilon}, 
we prove 
\begin{proposition}\label{prop: bound on a}
Let the event $\mathrm{FNR}(1,L)$ be true.
For all $k\ge 0$ and all interval $I\subset\Lambda_L$, 
\[
    \|V^{(k)}_I\|, \, \|A^{(k+1)}_I\| 
    \; \le \;
    \left(\frac{\gamma}{\varepsilon \delta}\right)^{\max\{|I|,\beta L_{k}\}}.
\]
\end{proposition}
\begin{proof}
Let $k\ge 0$ and let $I\subset\Lambda_L$ be an interval. 
Let us prove the claim for $\|V^{(k)}_I\|$.
Since $V^{(k)}(g)$ and $\widetilde V^{(k)}(g)$ coincide on $\mathrm{FNR}$, Proposition~\ref{pro: main inductive bounds} implies 
\[
    \| V^{(k)}(g) \|
    \; \le \; 
    \frac{1}{g!} \delta^{\|g\|} 
    \left(\frac{\gamma}{\varepsilon \delta}\right)^{|g|}.
\]
Therefore
\[
    \| V^{(k)}_I \|
    \; \le \; 
    \sum_{\substack{g\in\mathcal G^{(k)}:\\ {I}(g)=I}} 
    \|V^{(k)}(g)\|
    \; \le \; 
    \left(\frac{\gamma}{\varepsilon \delta}\right)^{\max\{|I|,L_k\}}
    \sum_{\substack{g\in\mathcal G^{(k)}:\\ {I}(g)=I}}
    \frac{\delta^{\|g\|}}{g!} ,
\]
where the second inequality follows from the bounds $|g|\ge |I(g)|=|I|$ and $|g|\ge L_k$, valid for all $g\in\mathcal G^{(k)}$.
It remains to bound the last sum in the right-hand side by 1 for $\delta$ small enough, 
which follows from Proposition~\ref{thm: counting diagrams}: 
\[
    \sum_{\substack{g\in\mathcal G^{(k)}:\\ {I}(g)=I}}
    \frac{\delta^{\|g\|}}{g!}
    \; \le \; 
    \sum_{w\ge |I|} \delta^w \sum_{x\in I}
    \sum_{\substack{g\in \mathcal G^{(k)}, \|g\|=w,\\\min{I(g)}=x}} \frac{1}{g!}
    \; \le \; \sum_{w\ge \frac13|I|} \delta^w w C^w \; \le \; 1,
\]
provided $\delta$ is small enough, 
where we have used the bound $|I| = |I(g)|\le |g| \le \| g\|$.

The bound on $\|A^{(k)}_I\|$ is obtained in a similar manner; however, in this case, we only have the bound $|t| \geq \beta L_k$ (instead of $|g| \geq L_k$), which accounts for the appearance of the factor $\beta$ in our claim.
\end{proof}

We are now ready to verify that the unitary $U$ on the left-hand side of \eqref{eq: explicit diagonalising unitary} is well-defined on $\mathrm{FNR}$. The main idea, also to be used in the next subsection, is to interpret the diagonalizing unitary $U$ as the time-1 evolution generated by a time-dependent generator $A_s$ with $s \in [0,1]$.

Let us define a sequence of times $(s_k)_{k\ge 0}$ by $s_0=0$ and  $s_k= s_{k-1}+ (\tfrac{1}{2})^{k}$ for $k\ge 1$, such that $s_{k}\to 1$ as $k\to\infty$. Then we construct the time-dependent generator 
\begin{equation}\label{eq: time dependent A s}
    A_s \; = \;  -2^{k+1} A^{(k+1)}, \qquad  s_{k}\leq s < s_{k+1}
\end{equation}
for $k\ge 0$. 
In particular, from Proposition \ref{prop: bound on a}, we deduce that 
\[
    \| A^{(k+1)}\|
    \;\leq\; 
    CLL_k 
    \left(\frac{\gamma}{\varepsilon\delta}\right)^{\beta L_k}.
\]
Hence the map $s\mapsto \| A_s \|$ is bounded on $[0,1)$ provided that $\gamma/\varepsilon\delta < 1$.
Therefore,  the integral form of the Schrödinger equation
\begin{equation}\label{eq: schrodinger integral}
    U_s\; = \; 1+ \int_0^s du \,  A_u U_u, \qquad  s\in [0,1]
\end{equation}
has a unique solution $(U_s)_{s\in [0,1]}$ and
\[
    U_1\; = \; \lim_{k \to \infty}   
    e^{-A^{(k)}} \dots e^{-A^{(1)}}.
\]
The unitary $U_1$ coincides with our diagonalizing unitary in \eqref{eq: explicit diagonalising unitary}, and we conclude that the limit exists. 

Finally, we can check that $U$ diagonalizes the Hamiltonian.
Indeed, from \eqref{eq: Hk = Ek + Vk},  
\[
   (U^{(k)})^\dagger H  U^{(k)} \; = \; H^{(k)} \; = \;  E^{(k)}+V^{(k)}.
\]
The operator $E^{(k)}$ is diagonal in the $Z$-basis, and $\|V^{(k)}\| \to 0$ as $k\to\infty$ by Proposition \ref{prop: bound on a}.
Since $U^{(k)}\to U$, we conclude that $E^{(k)}\to E$ where $E$ is a diagonal operator. 
We conclude that $U$ diagonalizes $H$. 

It is worth noting that the existence of the limit $U$ in \eqref{eq: explicit diagonalising unitary}, which diagonalizes the Hamiltonian for any fixed total length $L$, does not provide information about the behavior of $U$ as $L$ becomes large. This issue is addressed in the next subsection.

\subsection{Locality of the Diagonalizing Unitary} \label{sec: proof of theorem one}

In \eqref{eq: time dependent A s}, we introduced a time-dependent generator $A_s$ with $s \in [0,1]$, and constructed the unitary $U$ via the time-integrated Schrödinger equation~\eqref{eq: schrodinger integral}. We now apply Lieb-Robinson bounds to establish the locality-preserving properties of $U$.
We first describe the classical framework for studying dynamics generated by time-dependent generators, as presented in \cite{Nachtergaele2019}, building on the original work in \cite{Lieb1972}. We then demonstrate how to apply this framework to our specific setup.

\paragraph{Lieb Robinson Bounds and Locality.}
We consider functions $\Psi$, assigning to any subset $S \subset \Lambda_L$ a Hermitian operator $\Psi(S)$, supported in $S$.
To quantify the spatial decay of $\Psi$, we fix a family of positive functions 
$F_a: \mathbb{N} \to \mathbb{R}_+: r \mapsto \frac{a^r}{(1+r)^3} $, 
with $a \in (0,1)$,
and we define the family of norms
\[
    \|\Psi\|_{a} 
    \;=\; 
    \sup_{x,y \in \Lambda_L} \frac{1}{F_a(|x-y|)}  
    \sum_{ \substack{K \subset \Lambda_L \\ x,y \in K }} \|\Psi(K)\|.
\] 
We denote by $\mathfrak B_a$ the (finite-dimensional) Banach space determined by the norm $\|\cdot\|_a$.

Let us now consider a function $\Phi \in \caC([0,1],\mathfrak{B}_a)$, i.e.\ a continuous function with values in $\mathfrak B_a$. It defines a
unitary family $U^{(\Phi)}_s$ by the time-dependent Schr{\"o}dinger equation 
\begin{equation}\label{eq: schrodinger equation}
   \frac{d}{ds} U^{(\Phi)}_s= -i H_{\Phi_s}U^{(\Phi)}_s, \qquad  U^{(\Phi)}_0=1  
\end{equation}
where
\begin{equation}\label{eq: H Phi s}
H_{\Phi_s}= \sum_{K \subset \Lambda_L} \Phi_s(K).
\end{equation}
We are ready to state the Lieb-Robinson bound, first proven in \cite{Lieb1972}. 
Our next Theorem is Theorem 3.1 in \cite{Nachtergaele2019}, with the simplifications that the set $\Lambda_L$ is finite and that the constant $C_F$ satisfies $C_F\geq 1$: 
\begin{theorem}\label{thm: lieb robinson}
Let the operators $O,O'$ have support in $X,X' \subset \Lambda_L$ respectively, with $X\cap X'=\emptyset$. Let $\Phi$ be as defined above and let $\|\Phi\|_{1,a}=\int_0^1 ds \|\Phi_s\|_a$.  
Then, for $s\in [0,1]$,  
\[
    \| [(U^{(\Phi)}_s)^\dagger OU^{(\Phi)}_s ,O'] \| 
    \;\leq\; 
    2  \| O\| \|O'\| e^{2 C_F \|\Phi\|_{1,a}} \sum_{x\in X,x'\in X'}  F_a(|x-x'|) 
\]
where $C_F\geq 1$ is a constant that can be computed from the function $F_a$ with $a=1$.
\end{theorem}

A standard consequence of this Lieb-Robinson bound is that, for $O$ supported in an interval $X$ and for any $s\in[0,1]$, we can decompose, using the abbreviation 
$O_s=(U^{(\Phi)}_s)^\dagger OU^{(\Phi)}_s$,
\begin{equation} \label{eq: decomposition lr one}
O_s \; = \;   \sum_{{n= 0}}^\infty  (O_s)_n 
\end{equation}
where the term $(\cdot)_n$ is supported in $X_n=\{x | \dist(x,X)\leq n \}$, and 
\begin{equation} \label{eq: decomposition lr two}
\| (O_s)_n \| \;\leq\;  4C'_F \|O\| e^{2 C_F \|\Phi\|_{1,a}} a^n , \qquad n>0, 
\end{equation}
and $\|(O_s)_0\|=\|O\|$. 
Here, $C'_F$ is another constant that is computed from $F_a$ with $a=1$.

Let us briefly review this. 
For any $X\subset \Lambda_L$, let $\mathrm{Tr}_{X^c}$ be the normalized partial trace on the complement $X^c$, such that, for any observable $A$, $\mathrm{Tr}_{X^c}(A)$ is supported on $X$. 
We write 
\begin{equation}\label{eq: trace as integral}
    \mathrm{Tr}_{X^c}(A)
    \;=\; \int_{\mathcal U_{X^c}} dU \, 
 UAU^\dagger 
 \;=\;  A+ \int_{\mathcal U_{X^c}} dU \, 
 [U,A]U^\dagger .
\end{equation}
Here $\mathcal U_{X^c}$ is the group of unitaries whose support is in $X^c$, and $dU$ is the Haar probability measure on this group. 
Then one sets, noting that $X_0=X$,
\[
(O_s)_n \;=\; 
\begin{cases}   
\mathrm{Tr}_{X_n^c}(O_s) -  \mathrm{Tr}_{X_{n-1}^c}(O_s),  & n >0,  \\
\mathrm{Tr}_{X^c}(O_s),  & n =0.
\end{cases}
\]
Hence, $(O_s)_n=\mathrm{Tr}_{X_n^c}(O_s- \mathrm{Tr}_{X_{n-1}^c}(O_s)))$ for $n>0$ and  
\begin{align*}
    \|(O_s)_n\| 
   &  \;\leq\; \| O_s- \mathrm{Tr}_{X_{n-1}^c}(O_s) \| 
   \\
   & \;\leq\; \int_{\mathcal U_{X_{n-1}^c}} dU \, 
    \|[U,O_s]\|  \\
   & \;\leq\; 2 \|O\| e^{2 C_F \|\Phi\|_{1,a}} a^n  \left(2\sum_{r=1}^{\infty} rF(r)\right).  
\end{align*}
 The second inequality is by \eqref{eq: trace as integral} and the third inequality follows from Theorem~\ref{thm: lieb robinson}, with the expression between brackets being a bound for $\sum_{x\in X,x'\in X'}\ldots$ that is valid because  $X$ is  a discrete interval.
Now,  (\ref{eq: decomposition lr one}, \ref{eq: decomposition lr two}) indeed follow.

For our purposes, we need a slight extension of the above result, which we describe now.
Let $\Phi$ be an element of $L^{1}([0,1],\mathfrak B_a)$ instead of $\caC([0,1],\mathfrak B_a)$. Then, the differential equation \eqref{eq: schrodinger equation} should be replaced by the integral equation 
\begin{equation}\label{eq: integral schrodinger}
        U_s \; = \; 1-i\int_0^s du  H_{\Phi_u}U^{(\Phi)}_u.
\end{equation}
In this setting, we have 
\begin{corollary}\label{cor: lieb robinson and locality measurable}
  For $\Phi\in L^{1}([0,1],\mathfrak B_a)$, the unitary family $(U_s)_{s\in [0,1]}$, defined by \eqref{eq: integral schrodinger}, satisfies the bounds \eqref{eq: decomposition lr one} and \eqref{eq: decomposition lr two} for any observable $O$ supported in an interval $X$, and any $s\in[0,1]$.
\end{corollary}
\begin{proof}
We recall that the space $\mathfrak B_a$ is finite-dimensional because $L$ is finite.
For any $s\in [0,1]$, the map $\Phi \mapsto U_s^{(\Phi)}$, defined by \eqref{eq: integral schrodinger},  is continuous as a map from $ L^{1}([0,1],\mathfrak B_a)$ to finite matrices. This follows e.g.\ by the Gronwall inequality. 
Since  $\caC([0,1],\mathfrak B_a)$ is dense in $ L^{1}([0,1],\mathfrak B_a)$, the result follows then by the corresponding result for continuous functions $\Phi$, which was proven above by invoking Theorem \ref{thm: lieb robinson}.
\end{proof}

Finally, since our results establish the locality of both $U$ and $U^\dagger$, we will need the following lemma:
\begin{corollary}\label{cor: lieb robinson and locality measurable inverted}
With the same assumptions as in Corollary \eqref{cor: lieb robinson and locality measurable}, the bounds  \eqref{eq: decomposition lr one} and \eqref{eq: decomposition lr two} also hold with $O_s=  (U^{(\Phi)}_s)^\dagger OU^{(\Phi)}_s  $ replaced by $U^{(\Phi)}_s O(U^{(\Phi)}_s)^\dagger$
\end{corollary}
\begin{proof}
This follows in the same way as Corollary \ref{cor: lieb robinson and locality measurable} above, once one observes the following symmetry in Theorem \ref{thm: lieb robinson}  :   
\[  
\| [(U^{(\Phi)}_s)^\dagger OU^{(\Phi)}_s ,O'] \| \;=\;   
\| [U^{(\Phi)}_s O' (U^{(\Phi)}_s)^\dagger, O ] \|.  
\]
This symmetry follows from unitarity of $U^{(\Phi)}_s$ and from the fact that the operator norm $\| \cdot\|$ is invariant under unitary conjugation.
\end{proof}

\paragraph{Proof of Theorem~\ref{thm: locality of u}.}

Let us assume that the event $\mathrm{FNR}$ defined in \eqref{eq: fully non resonant} holds. 
Since $\mathbb P(\mathrm{FNR})\ge e^{-\gamma^c L}$ by \eqref{eq: bound on FNR}, it suffices to prove the result on this event. 

We will now apply the above framework 
to the time-dependent generator $A_s$ defined in \eqref{eq: time dependent A s}. 
Given $0\le s < 1$, let us identify $A_s$ with $H_{\Phi_s}$ defined in \eqref{eq: H Phi s}.
For this to work, we define $\Phi_s$ as follows: 
for sets $S\subset \Lambda_L$ that are of the form $S =\overline{I}$ for some interval $I$, we set
\[
    \Phi_s(S) \;=\;  -i 2^{k+1} A^{(k+1)}_{I}, 
    \qquad  s_{k}\leq s < s_{k+1},
\]
for $k\ge 0$, 
and $\Phi_s(S)=0$ for all other sets $S$. 
From the definition of the times $(s_k)_{k\ge 0}$ given above, prior to \eqref{eq: time dependent A s}, and from Proposition \ref{prop: bound on a}, we deduce that there exist constants $C,c$ such that
\[
    \|\Phi \|_{a,1} \;\leq\; \sup_{s\in [0,1)}  \|\Phi_s\|_a
    \;\leq\; C
    \qquad \text{with}\qquad  a= (\gamma/\delta\varepsilon)^c.
\]
Our theorem eventually follows from Corollary \ref{cor: lieb robinson and locality measurable} and Corollary \ref{cor: lieb robinson and locality measurable inverted}.

\subsection{Proof of Absence of Heat Conduction}\label{sec: absence heat conduction}

We now prove Theorem~\ref{thm: absence of conduction}. 
Since the current $J$, defined in~\eqref{eq: current J definition}, can be bounded in norm independently of the length $L$, the theorem holds for any fixed $L$, provided the constant $C$ is taken sufficiently large. 
In what follows, we assume that $L$ is large enough so that all expressions are well-defined and all estimates apply.

Consider two sites $\ell_1,\ell_2$ with $1 < \ell_1 <\ell_2 < L$.
We define the Hamiltonian $H_{\ell_1,\ell_2}$ by removing all terms in the  Hamiltonian $H$, defined in \eqref{eq: main Hamiltonian}, that are not supported in 
\[
    \Lambda_{\ell_1,\ell_2} \; := \; \{\ell_1,\ldots,\ell_2\} \; \subset \; \Lambda_L.
\]
In particular, the Hamiltonian in \eqref{eq: main Hamiltonian}, that was also denoted as $H_{\mathrm{sys}}$ in Section~\ref{sec: absence of conduction}, can now be consistently written as $H_{1,L}$.
The Hamiltonian $H_{\ell_1,\ell_2}$ is supported on the interval 
$\Lambda_{\ell_1,\ell_2}$ 
and depends only on the random variables $(\theta_{x})_{\ell_1 \le x \le \ell_2}$.
We also define the event $\mathrm{FNR}(\ell_1, \ell_2)$ analogously to the event $\mathrm{FNR}$ introduced in~\eqref{eq: fully non resonant}, but starting from $H_{\ell_1, \ell_2}$ instead of $H_{1, L}$.
This event depends only on the variables $(\theta_x)_{\ell_1 \le x \le \ell_2}$, and the conclusions of Theorem~\ref{thm: locality of u} apply to the Hamiltonian $H_{\ell_1, \ell_2}$ on this event.
Moreover, if $1 < \ell_1<\ell_2 <\ell_3<\ell_4 < L$, the events $\mathrm{FNR}({\ell_1,\ell_2})$ and $\mathrm{FNR}({\ell_3,\ell_4})$ are independent. 
The following lemma is a straightforward consequence of the bound~\eqref{eq: bound on FNR} and standard considerations on i.i.d.\@ random variables.
Let $c'$ be the constant featuring in \eqref{eq: bound on FNR}.

\begin{lemma}
On an event of probability at least $1-e^{-{L^{1-\gamma^{c'}}/\log L}}$, there exist random sites $\ell_1,\ell_2$ as above such that the event $\mathrm{FNR}({\ell_1,\ell_2})$ holds and such that $\ell_2 - \ell_1 \ge  \frac12\log L$.
\end{lemma}

From here onward, we assume that the event introduced in the above lemma holds, and we let $\ell_1,\ell_2$ be two sites such that $\mathrm{FNR}({\ell_1,\ell_2})$ holds and such that $\ell_2 - \ell_1 \ge  \frac12\log L$.
The conclusions of Theorem~\ref{thm: locality of u} apply to $H_{\ell_1,\ell_2}$ and, 
as in \eqref{eq: H diagonalization}, we let 
\[
    U^\dagger H_{\ell_1,\ell_2} U
    \; = \; D,
\]
where $D$ is diagonal, and where both $D$ and the unitary $U$ are supported on the interval $\Lambda_{\ell_1,\ell_2}$. 
The operator $D$ can be expanded as
\begin{equation}\label{eq: expansion D operator}
     D \; = \; \sum_{\substack{I\subset \Lambda_{\ell_1,\ell_2},\\ \text{interval}}} \underline{D}_I,
\end{equation}
where we have used the notation $\underline{D}_I$ to emphasize that it represents a diagonal operator, rather than merely a coefficient as in the expansion \eqref{eq: D sum local}.
We check that the analog of the bound \eqref{eq: bound local D} remains valid: 
\begin{equation}\label{eq: bound D I underlined}
    \|\underline D_I\| \; \le \; 
    C \min \left\{ 1,  {|I|}^2 \gamma^{c(|I| - 2)/2} \right\}, 
    \qquad I \ne \varnothing. 
\end{equation}
We can now decompose $D$ into a left and a right part: 
with $\ell_* = (\ell_1 + \ell_2)/2$, 
\[
    D_{\mathbf l} 
    \; = \;
    \sum_{\substack{I \subset \Lambda_{\ell_1,\ell_2}:\\  \min I \le \ell_*}}
    \underline{D}_I , 
    \qquad 
    D_{\mathbf r} \; = \; D - D_{\mathbf l}.
\]

Recall the definition of $H_{\mathrm{tot}}$ in \eqref{eq: def of H tot} and the decomposition $H_{\mathrm{tot}} = H_{\mathbf l} + H_{\mathbf r}$ in \eqref{eq: decompoisition H tot left right}. 
We now introduce a different decomposition, and we set
\[
    H_{\mathrm{tot}}
    \; = \; 
    H'_{\mathbf l} + H'_{\mathbf r}
\]
with 
\begin{align*}
    H'_{\mathbf l}
    \; &= \; 
    H_{B,\fral}+V_{B,\fral}\otimes X_1 + 
    U D_{\mathbf l} U^\dagger + R_{\mathbf l}, 
    \qquad
    R_{\mathbf l} \; = \; 
    \sum_{\substack{I\subset \Lambda_L \text{ interval},\\ \min I < \ell_1}} H_I, \\
    H'_{\mathbf r}
    \; &= \; 
    H_{B,\frar}+V_{B,\frar}\otimes X_L + 
    U D_{\mathbf r} U^\dagger + R_{\mathbf r}, 
    \qquad
    R_{\mathbf r} \; = \; 
    \sum_{\substack{I\subset \Lambda_L \text{ interval},\\ \min I \ge  \ell_1, \max I > \ell_2}} H_I
\end{align*}
where we have decomposed $H_{\mathrm{sys}} = H_{1,L}$, which appears in \eqref{eq: def of H tot}, according to \eqref{eq: decomposition H sys intervals} in order to define $R_{\mathbf l}$ and $R_{\mathbf r}$.
Since $[D_{\mathbf l},D_{\mathbf r}] = 0$, we find 
\begin{equation}\label{eq: commutator 3 parts}
    [H'_{\mathbf l},H'_{\mathbf r}]
    \; = \; 
    [ UD_{\mathbf l}U^\dagger, R_{\mathbf r} ]
    + 
     [ R_{\mathbf l} , U D_{\mathbf r} U^\dagger ] 
    + 
    [R_{\mathbf l},R_{\mathbf r}]
    + 
    [X_1,R_{\mathbf r}]
    + 
    [R_{\mathbf l},X_L].
\end{equation}

Applying Theorem~\ref{thm: locality of u} once more, and expressing $D$ as in \eqref{eq: expansion D operator}, we can expand $UDU^\dagger$ as
\[
    U D U^\dagger 
    \; = \; 
    \sum_{\substack{I\subset \{\ell_1,\dots,\ell_2\},\\ \text{interval}}} 
    U \underline{D}_I U^\dagger
    \; = \; 
    \sum_{\substack{I\subset \{\ell_1,\dots,\ell_2\},\\ \text{interval}}} 
    \underline{D}_I' 
\]
where the bound \eqref{eq: bound D I underlined} leads to
\[
    \|\underline D_I'\| \; \le \; 
    C \min \left\{ 1,  {|I|}^4 \gamma^{c(|I| - 2)/2} \right\}, 
    \qquad I \ne \varnothing. 
\]
This decomposition of $D$ in local terms carries over directly to $D_{\mathbf l}$ and $D_{\mathbf r}$, with the extra condition that the support of the terms in the decomposition of $D_{\mathbf l}$ intersect $\Lambda_{\ell_1,\lceil \ell_*\rceil}$, and the support of the terms in the decomposition of $D_{\mathbf r}$ intersect $\Lambda_{\lceil l_\star\rceil,l_2}$.
We can now bound the norm of the five terms in the right-hand side of \eqref{eq: commutator 3 parts} and find 
\begin{equation}\label{eq: bound commutator prime}
    \|[H'_{\mathbf l},H'_{\mathbf r}]\|
    \; \le \;
    C \gamma^{(c/5) \log L} \; = \; C L^{-\frac{c}5\log (1/\gamma)}.
\end{equation}
where we have used the bound $\ell_2  - \ell_2 \ge \frac12\log L$.

To conclude the proof, we write 
\begin{align*}
    \frac1T \int_0^T dt J(t)
    \; &= \; 
    \frac{i}{T}\int_0^T dt [H_{\mathrm{tot}},H_{\mathbf l}(t)]\\
    \; &= \; 
    \frac{i}{T}\int_0^T dt [H_{\mathrm{tot}},H'_{\mathbf l}(t)]
    + 
    \frac{i}{T}\int_0^T dt [H_{\mathrm{tot}},H_{\mathbf l}(t) - H'_{\mathbf l}(t)].
\end{align*}
The norm of the first term in the right-hand side can be bounded thanks to \eqref{eq: bound commutator prime}, 
since 
$\|[H_{\mathrm{tot}},H'_{\mathbf l}(t)]\|=\|[H'_{\mathbf l},H'_{\mathbf r}]\|$. 
The second one is the time integral of a total derivative of an observable that does not involve the baths: 
\[
    \frac{i}{T}\int_0^T dt [H_{\mathrm{tot}},H_{\mathbf l}(t) - H'_{\mathbf l}(t)]
    \; = \; 
    \frac1T \left( ( H_{\mathbf l}(T) - H'_{\mathbf l}(T)) - (H_{\mathbf l}(0) - H'_{\mathbf l}(0)  \right)
\]
with 
\[
    H_{\mathbf l} - H'_{\mathbf l}
    \; = \; 
    H_{\mathrm{sys},\mathbf l} - UD_{\mathbf l} U^\dagger - R_{\mathbf l}. 
\]
Therefore 
\[
    \frac1T \left\|\int_0^T dt J(t) \right\|
    \; \le \; 
    C L^{-\frac{c}5\log (1/\gamma)} + 
    \frac{CL}{T}.
\]
The bound on $|\langle J\rangle_{\mathrm{ness}}|$ follows by letting $T\to\infty$. 
Finally, to get the claim about the expectation value, we use the a priori bound $\|J\| \leq C$ and the dominated convergence theorem.

\appendix

\section{Smooth Approximation of Some Indicator Functions}\label{sec: smoothing lemma}

In this section, we denote by $n$ the dimension of the sample space, i.e.\@ we let $\Omega = [0,1]^n$. Crucially, we notice that $\Omega$ is a convex set.
Given $p$ smooth functions $f_1,\dots,f_p$ on $\Omega$ and given a threshold $\eta > 0$, the function $Q$ constructed in the lemma bellow can be thought of as a smooth approximation of the set $\{|f_1|\le \eta, \dots, |f_p|\le \eta\}$ while $S$ can be viewed as a smooth approximation of the set $\{|f_1|\ge\eta,\dots,|f_p|\ge\eta\}$:

\begin{lemma}\label{lem: abstract lemma smoothing}
Let $f_1,\dots,f_p$ be $p$ smooth functions on $\Omega$. 
Assume that there exists a number $B\ge 0$ such that 
\[
	\left|\frac{\partial f_k}{\partial \theta_i} (\theta) \right| \; \le \; B, \qquad 1 \le i \le n, \qquad 1 \le k \le p,  \qquad \theta \in \Omega. 
\]
Let finally $\eta > 0$.
There exist smooth functions $Q,S$ on $\Omega$ with the following properties: 
\begin{enumerate}
    \item 
    $0 \le Q(\theta),S(\theta) \le 1$.
	
    \item 
    For all $\theta\in\Omega$, 
    \[
    Q(\theta)>0 \; \Rightarrow \; \big(|f_k(\theta)|\le 2\eta \quad \forall 1\le k \le p \big), 
    \qquad 
    S(\theta)>0 \; \Rightarrow \; \big(|f_k(\theta)|\ge \eta/2 \quad \forall 1\le k \le p \big)
    \]
	
    \item 
    For all $\theta\in\Omega$,
    \[
    \big(|f_k(\theta)|\le \eta \quad \forall 1\le k \le p\big) \; \Rightarrow \; Q(\theta)=1, 
    \qquad 
    \big(|f_k(\theta)|\ge \eta \quad \forall 1\le k \le p \big)\; \Rightarrow \; S(\theta)=1
    \]
	
	\item 
	There exists a universal constant $\Const$ such that 
	\begin{equation}\label{eq: derivative smooth indicator}
		\left|\frac{dQ}{d\theta_i} (\theta) \right| , \left|\frac{dS}{d\theta_i} (\theta) \right|\; \le \; \Const \frac{B n}{\eta}, \qquad 1 \le i \le n, \qquad \theta \in \Omega. 
	\end{equation}
\end{enumerate}
\end{lemma}

\begin{proof}
	We construct the function $Q$. The construction of the function $S$ is analogous. 
	Let $a>0$ be the width of the smoothing, that we will need to fix, and let us define the function $Q$ as 
	\[
		Q(\theta) \; = \; \int_\Omega d\theta' \rho(\theta,\theta') 1_{\{|f_1|\le3\eta/2\}}(\theta') \dots 1_{\{|f_p|\le3\eta/2\}}(\theta')
	\]
	for some smooth kernel $\rho$ on $\Omega^2$ that has the following properties:
    \begin{enumerate}
        \item 
       $\rho \ge 0$, 

        \item 
        $\rho(\theta,\theta')= 0$ as soon as $|\theta - \theta'|_\infty>a$ for all $\theta,\theta' \in\Omega$, 

        \item 
        $\int_\Omega d\theta' \rho(\theta,\theta') = 1$ for all $\theta\in\Omega$, 

        \item 
        For all $\theta\in\Omega$ and all $1\le i \le n$, 
        \[
        \int_\Omega d\theta' 
        \left|\frac{\partial\rho}{\partial\theta_i}(\theta,\theta')\right|
        \;\le\;\frac{\Const}{a}.
        \]
    \end{enumerate}
    The kernel $\rho$ can be defined as
    \[
        \rho(\theta,\theta') 
        \; = \; 
        \varphi(\theta_1,\theta_1') \dots \varphi(\theta_n,\theta_n')
    \]
    for some smooth kernel $\varphi$ on $[0,1]^2$. 
    The four properties above will be satisfied for $\rho$ if they are satisfied for $\varphi$, replacing $\Omega$ by $[0,1]$ and taking $n=1$. 
    To construct $\varphi$, we consider a positive function $u$ on $\R$, symmetric (i.e.\@ $u(x)=u(-x)$ for all $x\in\R$), non-negative, supported on $[-a,a]$, such that $\int u = 1$ and such that $\int|u'| \le \Const/a$. 
    We cannot define $\varphi(\theta,\theta')$ to simply be $u(\theta-\theta')$ because the 3rd property above would not be satisfied for $\theta$ near the boundary. Instead, assuming $a<1/2$, we may set 
    \[
    \varphi(\theta,\theta') \; = \; 
    u(\theta - \theta') + u(-\theta - \theta') + u((2-\theta) -\theta'), 
    \qquad \theta,\theta'\in [0,1]. 
    \]
    The kernel $\varphi$ satisfies then the four required properties.

It follows from this definition that $0 \le Q \le 1$ (item 1 in the claim) and that 
\[
    \left|\frac{dQ}{d\theta_i} (\theta) \right| 
    \; \le \; \frac{\Const}{a}, \qquad 1 \le i \le n, 
    \qquad \theta \in \Omega,
\]
so that item 4 in the claim will be satisfied if $a\ge \eta/2Bn$, and we will see that this value is small enough for the two other items.
	
	For item 2, we observe that if $Q(\theta)>0$, then there exists $\theta'$ such that $|f_k(\theta')|\le 3\eta/2$ for all $1\le k \le p$ and $|\theta-\theta'|_\infty \le a$. 
	Therefore, it is enough to show that $|f_{k}(\theta)-f_{k}(\theta')|\le \eta/2$ for all $1\le k \le p$. 
	Since $\Omega$ is convex, we know that
	\[
		|f_{k}(\theta)-f_{k}(\theta')| \;\le\; \sup_{\widetilde\theta\in\Omega}|\langle\nabla f_{k}(\widetilde\theta),(\theta - \theta')\rangle| \; \le \; n B |\theta-\theta'|_\infty \le n Ba
	\] 
    for all $1\le k \le p$.
	This imposes $a\le \eta/2Bn$. 
	
	For item 3, let us show that if $\theta$ is such that $|f_k(\theta)|\le\eta$ for all $1 \le k \le p$, and if $\theta'$ is such that $|\theta-\theta'|_\infty\le a$, 
	then $|f_k(\theta')|\le 3\eta/2$ for all $1 \le k \le p$, which implies $Q(\theta)=1$.  
	As for the previous item, it is enough to prove that $|f_k(\theta)-f_k(\theta')|\le \eta/2$ for all $1 \le k \le p$, which will hold if $a\le \eta/2Bn$. 
\end{proof}

\section{Spectral Bounds for Perturbed Triangular Matrices}\label{sec: appendix spectrum}

We consider a $d\times d$ matrix $M=(M_{i,j})_{i,j=1,\ldots,d}$ of the form $M=D+N+E$ where
\begin{enumerate}
\item $D$ is diagonal,
    \item  $N$ is upper triangular and off-diagonal, and its matrix elements are bounded, more precisely: 
    \[
    |N_{i,j}| \;\leq\;  \chi(j>i), \qquad  i,j =1,\ldots,d,
    \]
    \item  $E$ is lower triangular and off-diagonal, and its matrix elements decay in the distance to the diagonal, more precisely: 
    \[
    |E_{i,j} | \;\leq\; \epsilon^{|i-j|+1} \chi(i>j), \qquad i,j=1,\ldots,d
    \]
for some $0<\epsilon<1$.
\end{enumerate}
To state the result, let us write $\mathrm{spec}(D)$ for the set of eigenvalues of $D$. 
\begin{lemma}\label{lem: spectral lemma}
There exists a universal constant $C$ such that any eigenvalue $\lambda$ of $D+N+E$  satisfies
$$
\mathrm{dist}(\lambda,\mathrm{spec}(D)) \;\leq\; \Const \epsilon.
$$
\end{lemma}
\begin{proof}
We define the unperturbed resolvent $R_0(z)=\frac{1}{(z-D)}$ for $z \notin \mathrm{spec}(D) $ and the perturbed resolvent $R(z)=(z-(D+N+E))^{-1}$ for  $z \notin \mathrm{spec}(D+N+E) $. Whenever the Neumann series
\[
R(z) \;=\;   R_0(z) + R_0(z)(N+E)  R_0(z) +  R_0(z)(N+E)  R_0(z) (N+E) R_0(z) +\ldots 
\]
is absolutely convergent, it follows that $z$ is not an eigenvalue of $D+N+E$. 
We consider a matrix element $(R(z))_{i,j}$ and we dominate the series by a sum over finite walks on the indices $\{1,\ldots,d\}$:
\[
    |(R(z))_{ij}|
    \;\leq\; 
    {\frac1\eta +} \frac1\eta\sum_{k=1}^{\infty}  \,   
    \sum_{\substack{1 \le i_{{0}},\ldots,i_k \le d \\ i_{{0}}=i, i_k{=}j}} \,  \prod_{\ell={1}}^k w(i_{\ell}-i_{\ell-1})   
\]
where $\eta=\mathrm{dist}(z,\mathrm{spec}(D))$ and where
\[
    w(m) \; = \; 
    \begin{cases} 
    1/\eta & m>0, \\
    0 & m=0, \\
    \epsilon^{1+|m|}/\eta  &  m<0.
\end{cases}
\]

Let $\mathbf w = (i_0,\dots,i_k)$ be a walk with $k = k(\mathbf w)$ steps and $1 \le i_0,\dots,i_k \le d$. 
We don't impose $i_0=i$ nor $i_k=j$ but, because of the form of $w(\cdot)$ above, we can require the walk to never stand still, i.e.\@ $i_{\ell}\neq i_{\ell-1}$ for all $1 \le \ell \le k$.
Let $k_+ = k_+(\mathbf w)$ and $k_- = k_-(\mathbf w)$ be the number of steps in positive and negative direction respectively, that is, corresponding to $m>0$ and $m<0$. 
Then $k = k_+ + k_-$.
Let also $r_+ = r_+(\mathbf w)$ and $r_- = r_-(\mathbf w)$ be the total distance travelled by the walk when going in positive and negative direction respectively.
Then the weight $W(\mathbf w):=\prod_{\ell={1}}^{k} w(i_{\ell}-i_{\ell-1}) $ can also be written as
\[
    W (\mathbf w) \;=\; (\epsilon/\eta)^{k_-} \epsilon^{r_-} (1/\eta)^{k_+}.
\]
The following inequalities hold
\[
    r_-\geq r_+-d, \qquad   r_+\geq k_+, 
\]
where the first one follows because the walk cannot exit the interval $\{1,\ldots,d\}$.  
We can therefore dominate the weight as 
\begin{align*}
    W(\mathbf w)
    \; &= \; 
    \frac{\epsilon^{k_- + r_-}}{\eta^k}
    \; = \; 
    \left(\frac12\right)^{k_- + r_-} \frac{(2\epsilon)^{k_- + r_-}}{\eta^k}
    \;\leq\;  
    \left(\sqrt 2 \epsilon\right)^{-d}  
    \left(\frac1{\sqrt{2}}\right)^{r_++r_-} \left(\frac{2\epsilon}{\eta}\right)^{k}\\
    \; &=: \; C(\epsilon,d) \widetilde W(\mathbf w)
\end{align*}
where we have obtained the inequality using the bounds
\[
    k_- + r_- \; \ge \, r_- \; \ge \; \frac12 (r_- + r_+ - d), 
    \qquad 
    k_- + r_- \; \ge \; k_- + r_+ - d \; \ge \; k - d.
\]

For $k\ge 1$, let now 
\[
    S_k \; = \; \sum_{\mathbf w : k(\mathbf w) = k} \widetilde W(\mathbf w). 
\]
Provided that 
\[
    \frac{2\epsilon}{\eta} \sum_{r\ge 1} \left(\frac{1}{\sqrt{2}}\right)^r
    \; = \; 
     \frac{2\epsilon}{\eta} \frac{1}{\sqrt{2}-1} 
     \;\le \; \frac12, 
\]
we conclude that $S_{k+1} \; \le \; S_k/2$ for all $k \ge 1$. 
This implies that the Neumann series for $R(z)$ converges, and this shows the lemma with $C = 4/(\sqrt 2 - 1)$. 
%
\end{proof}

\bibliographystyle{plain}
\bibliography{bibliography}
\end{document}